\documentclass{JFM-FLM_Au}
\makeatletter
\def\arxivfoot{\hbox to \textwidth{\hfil{\cppagefont\thepage}\hfil}}
\def\ps@headings{\let\@mkboth\markboth
  \def\@oddhead{\hfill{\itshape\@righttitle}\hfill}%
  \def\@evenhead{\hfill\itshape\@lefttitle\hfill}%
  \def\@oddfoot{\arxivfoot}\def\@evenfoot{\arxivfoot}%
  \def\sectionmark##1{\markboth{##1}{}}\def\subsectionmark##1{\markright{##1}}}
\def\ps@titlepage{\leftskip\z@\let\@mkboth\@gobbletwo\vfuzz=5\p@
  \def\@oddhead{\vbox{\vspace*{-4pt}\hbox to \textwidth{\@j@urnal\hfill}}}%
  \let\@evenhead\@oddhead
  \def\@oddfoot{\arxivfoot}\let\@evenfoot\@oddfoot
  \def\sectionmark##1{}\def\subsectionmark##1{}}
\AtBeginDocument{\pagestyle{headings}}
\makeatother

\usepackage{xcolor}
\usepackage{tikz}
\usetikzlibrary{calc,arrows.meta}
\usepackage{subcaption}

\definecolor{jfmblue}{RGB}{0,56,168}

\hypersetup{
  colorlinks=true,
  citecolor=blue,
  linkcolor=blue,
  urlcolor=blue
}

\lefttitle{Sharma and Kerstein}
\righttitle{Journal of Fluid Mechanics}

\title{Strain-coupled one-dimensional turbulence: a single-line closure, its limits, and the consequence for leading-edge noise}

\author{Sparsh Sharma\aff{1} \and Alan R. Kerstein\aff{2}}

\affiliation{
\aff{1}German Aerospace Center (DLR), 38108 Braunschweig, Germany
\aff{2}Consultant, Danville, CA 94526, USA
}

\corresau{Sparsh Sharma, sparsh.sharma@dlr.de}

\begin{document}

\maketitle

\begin{abstract}
Turbulence approaching a solid surface is strained before it arrives,
and the distortion of its spectrum is often what governs the downstream
response; yet reduced-order predictions, leading-edge-noise theory among
them, still take the incoming spectrum as frozen. We develop a
strain-coupled formulation of one-dimensional turbulence (ODT) that
evolves strained turbulence with full spatial resolution at a cost
compatible with parametric study, carrying the rapid distortion and the
nonlinear relaxation together. The obstacle is that the rapid
pressure--strain is non-local and three-dimensional while an ODT line
resolves only the flow state along its own coordinate.
Splitting the distortion physics by structure captures the phenomenology: the
off-axis production and a spectrally closed rapid operator enter as
continuous-in-time forcing along the model coordinate, while `eddy events' capture the on-line production and
the turbulence-driven relaxation.
Carried through to leading-edge noise, the evolved inflow shifts the
predicted level by $5$--$20$\,dB at the energy-containing scales relative
to the conventional frozen-spectrum assumption.
\end{abstract}


\section{Introduction}\label{sec:intro}
The distortion of turbulence by a mean velocity gradient is among the
most basic problems of turbulence dynamics. When a turbulent flow is
strained---accelerated through a contraction, deflected around a bluff
body, compressed towards a stagnation region---the mean deformation
reorganises the fluctuating field: it amplifies some velocity components
and suppresses others, redistributes energy among them through the
fluctuating pressure, and rescales the spectral content of the motion.
When the strain accumulates in a time short compared with the eddy
turnover time, the nonlinear interactions cannot keep pace and the
response is governed, to leading order, by the linear dynamics of rapid
distortion, first analysed by \citet{BatchelorProudman1954} and
consolidated by \citet{TownsendBook1976}, extended to flow around
obstacles and near boundaries by \citet{Hunt1973} and
\citet{HuntGraham1978}, and given its modern systematic form by
\citet{HuntCarruthers1990} and \citet{SagautCambon2018}. The turbulence that emerges from a strained
region differs from the turbulence that entered it---in intensity,
anisotropy and spectral shape---and that difference is often the
quantity of practical interest, from grid turbulence in wind-tunnel
contractions \citep{MillsCorrsin1959,TuckerReynolds1968} and its
numerical counterpart \citep{LeeReynolds1985} to the distortion of
oncoming turbulence approaching a lifting surface
\citep{HuntGraham1978,GoldsteinAtassi1976}.

Capturing this distortion at a cost compatible with parametric study is
difficult. Scale-resolving simulation represents it directly but
precludes broad sweeps over intensity, length-scale ratio and Reynolds
number; linear rapid-distortion theory (RDT) is exact in the rapid limit
but neglects, by construction, the inter-scale transfer and the slow
return to isotropy that act over any finite strain duration
\citep{SagautCambon2018}; second-moment closures restore the nonlinear
relaxation but discard the spectral information that the distortion
redistributes \citep{LaunderReeceRodi1975,Pope2000}. Between these sits
a gap: a reduced-order model that \emph{evolves} strained turbulence
with full spatial resolution, retaining the rapid distortion and the
nonlinear relaxation together.

The application motivating this work is turbulence--aerofoil
interaction noise. The prevailing
prediction framework \citep{Amiet1975,PatersonAmiet1977,RogerMoreau2010}
obtains the radiated spectrum as a linear functional of the upwash
velocity spectrum of the incoming turbulence---but the spectrum it
operates on is prescribed, not computed, usually as an isotropic
von~K\'arm\'an form \citep{vonKarman1948}, whereas the turbulence
actually reaching the leading edge is strained turbulence. The
importance of this pre-impact distortion for the prediction is by now
well established: measurements and scale-resolved simulations show the
near-edge spectrum systematically altered
\citep{dosSantos2023,Ribeiro2023,Piccolo2024}, and corrected predictions
have been obtained by prescribing RDT-transformed inflow spectra
\citep{deSantana2016,dosSantos2023} or by sampling the spectrum close to
the edge \citep{Piccolo2024}. What these corrections share is that the
distorted spectrum is imposed---from linear theory, or from a
measurement that a prediction tool does not generate internally. Synthetic-turbulence
generators \citep{KimHaeri2015} share the limitation by construction:
they reproduce a prescribed spectrum rather than evolving one. A model
that predicts the distorted spectrum, including the nonlinear relaxation
that linear transforms omit, would supply the missing element, here and
wherever distorted turbulence feeds a downstream prediction. The subject
of this paper is that distortion problem; the acoustic stage enters
only as the established linear functional that operates on the model's
output.

Map-based stochastic one-dimensional models offer a route. The
one-dimensional-turbulence (ODT) model
\citep{Kerstein1999,Kerstein2001}, built on the triplet map of the
linear-eddy model \citep{Kerstein1991}, evolves a three-component
velocity field along a notional line, resolving the full range of length
scales. Turbulent advection is represented as a stochastic sequence of
mapping events interleaved with a continuous-in-time, one-dimensional
specialisation of viscous transport; within each event a kernel mechanism
redistributes energy among components, emulating pressure scrambling
\citep{WunschKerstein2001,Kerstein2001}.

Generalised to variable density
\citep{AshurstKerstein2005}, adaptive meshes \citep{Lignell2013,
StephensLignell2021} and cylindrical geometry \citep{Lignell2018},
embedded as a subgrid closure \citep{SchmidtKerstein2010}, and coupled
to acoustic analogies for jet noise
\citep{SharmaKleinSchmidt2022,MedinaMendez2023}, ODT is a general
low-cost generator of scale-resolved turbulence statistics. But standard ODT
cannot represent distortion by a mean strain: its pressure treatment
models only the slow, turbulence-driven return to isotropy, not the
rapid part driven directly by the mean strain, and the
apparent remedy---importing the RDT amplitude equation, whose pressure
response is the solenoidal projection in three-dimensional wavenumber
space---is unavailable on a line that resolves the wavenumber only along
its own coordinate. Reconciling the inherently three-dimensional,
non-local pressure physics of rapid distortion with the one-dimensional,
locally resolved structure of ODT is the central modelling problem addressed here.

The paper advances two claims: that ODT can be coupled to a mean strain
by means of a conservative, rapid-distortion-consistent transform, whose
distorted spectrum is validated directly against an independent
scale-resolving benchmark; and that the validated distortion changes the
leading-edge-noise prediction relative to the conventional frozen-spectrum
assumption. The supporting results are the following.
\begin{itemize}
\item A strain-coupled formulation of ODT (SC-ODT)
(\S\,\ref{sec:formulation}), obtained as the restriction to the line of
the Rogallo transform of the ODT formulation to the frame deforming with
the mean. The off-axis production and the rapid pressure--strain enter as
a continuous-in-time forcing built from an explicit spectral
closure, and the wavevector kinematics as a dilatation of the domain, while the
on-line upwash production, the one channel that varies along the model
coordinate, is carried by the eddy events acting on the full velocity,
which deposit it at the scale of the eddy motion. Each event conserves
momentum and energy---the one-dimensional analogues of the Navier--Stokes
invariants---so the construction is a transform of one-dimensional
turbulence, not a departure from it; standard ODT is recovered when the
strain is removed, and under the stated closure the component energies
follow the rapid-distortion evolution with no adjustable constant
(\S\,\ref{sec:verification}).

\item A treatment of what a single line can determine about the rapid
pressure--strain term (a non-local, three-dimensional transformation). The
limitation is stated exactly, sharp bounds are obtained from line data
alone, and the error of the axisymmetry hypothesis is quantified: zero for
axisymmetric distortion, finite under plane strain (\S\,\ref{sec:gate}).

\item Validation against the strained-turbulence experiment of
\citet{ChenMeneveauKatz2006} (CMK), summarised in \S\,\ref{sec:cmk}
with details in the supplementary material (\S\,S1).

\item The linear reference appropriate to the model's observables: exact
RDT projected onto a line is broadband rather than a rigid spectral
translation, so the reference against which any line model must be judged
already differs from the model's own kinematics at the linear level
(\S\,\ref{sec:spec-linref}).

\item An independent spectral benchmark consisting of a strained-box direct
numerical simulation together with exact-RDT companion calculations,
reduced to the same one-dimensional observables
(\S\,\ref{sec:val-spectrum}). Against it the formulation reproduces the
sign and wavenumber dependence of the strained upwash anisotropy at
moderate strain rate, and a
distance-from-linear-theory envelope is supplied as a function of strain
rapidity $Sk_t/\varepsilon$ (strain rate $S$, turbulent kinetic energy
$k_t$, dissipation rate $\varepsilon$).

\item The consequence for the motivating application
(\S\,\ref{sec:lenoise}): the computed distorted spectrum is supplied to
Amiet's leading-edge-noise theory in place of the frozen inflow of standard
practice, both as a ratio from which aerofoil geometry and observer position cancel
and as an absolute prediction at the flow conditions of the reference
configuration. The
calculation quantifies the $5$--$20$\,dB term omitted by prescribed-spectrum
practice; the assembled chain is compared with three published experiments
and lies within $1$--$2$\,dB rms for the facilities with the
best-characterised inflow.
\end{itemize}

The paper falls into two parts, following the two claims.
\begin{itemize}
\item The model and what a line can know: \S\,\ref{sec:formulation}
introduces the ODT model and derives the strain-coupled formulation as the
Rogallo transform of it; \S\,\ref{sec:gate} treats the single-line closure
question---what a line can and cannot determine about the rapid
pressure--strain term; \S\,\ref{sec:verification} verifies the
formulation against exact rapid-distortion solutions and validates it at
the moment level against the DNS of \citet{LeeReynolds1985}; and
\S\,\ref{sec:val-spectrum} constructs the line-projected linear reference
and carries out the strained-box DNS benchmark of the distorted spectrum.

\item \S\,\ref{sec:lenoise} applies the computed distortion as a
correction to the leading-edge-noise input.
\end{itemize}
Section~\ref{sec:conclusions} concludes. The supplementary material gives
the full comparison with the strained-turbulence experiment of
\citet{ChenMeneveauKatz2006} (\S\,S1), the operational steps of the
leading-edge-noise chain (\S\,S2), the verification of the implementation
(\S\,S3), a space--time view of the strained line (\S\,S4) and a
cross-check of the Rogallo-frame implementation against a lab-frame one
(\S\,S5).

\section{Strain-coupled one-dimensional turbulence}\label{sec:formulation}

The linear rapid-distortion equations, recalled below, contain two
distinct ingredients: an evolution of the Fourier \emph{amplitudes},
comprising production and the solenoidal pressure projection, and an
evolution of the \emph{wavevectors} that carries the kinematic straining
of scales. The pressure part of the amplitude equation requires the full
three-dimensional wavevector and is therefore not directly representable
on a one-dimensional line; that incompatibility is the focus of this section.
The distortion physics is split according to its structure and each part
assigned to the mechanism that can carry it. The production is
local in physical space, but only its on-line channel---the wall-normal
mean gradient $A_{22}$, which acts along the model coordinate---represents
mean momentum the line can transport; it is carried by the eddy events
themselves, which act on the full velocity, so that the construction is a
transform of one-dimensional turbulence to the strained frame rather than
a departure from it (\S\,\ref{sec:event-strain}). The off-axis
production and the rapid pressure--strain, the latter non-local and not
representable pointwise, enter as a continuous forcing
of the line velocity, the rapid operator being constructed to reproduce
the homogeneous rapid-distortion evolution of the component energies
under an explicit spectral closure
(\S\S\,\ref{sec:strain-forcing}--\ref{sec:rapid-kernel}). The wavevector
kinematics are recovered by a consistent straining of the domain
(\S\,\ref{sec:domain-strain}). The slow, turbulence-driven
pressure--strain and the inertial cascade are, as in standard ODT,
carried by the discrete eddy events, reviewed first.

Throughout, $A_{ij}=\partial U_i/\partial x_j$ denotes the imposed mean
velocity gradient, taken uniform across the line at each instant; the
medium is incompressible, $A_{ii}=0$. The line carries the velocity
$v_i(y,t)\equiv U_i+u_i$, the sum of the imposed mean $U_i$ and a
fluctuation $u_i$ with $\overline{u_i}=0$; $R_{ij}\equiv\overline{u_iu_j}$ is the
Reynolds-stress tensor, the overline denoting the average over the
statistically homogeneous directions (equivalently, the ensemble of
independent realisations), and $k_t\equiv\tfrac12R_{ii}$ the turbulent
kinetic energy. In the homogeneous configurations used to establish and
test the formulation the spatial line average and the ensemble average
coincide. The eddy events act on the velocity $v_i$
(\S\,\ref{sec:odt-review}); the continuous evolution between events acts on
the fluctuation $u_i$ (\S\,\ref{sec:strain-forcing}), which denotes the
fluctuation throughout. The fluctuating pressure is written $p'$.

\subsection{Vector-valued ODT}\label{sec:odt-review}
\begin{figure}
  \centering
\begin{tikzpicture}[
  x=1cm, y=1cm,
  line width=0.4pt,
  font=\fontsize{8}{9}\selectfont,
  ev/.style={-{Latex[length=1.4mm,width=1.1mm]}},
  ax/.style={-{Latex[length=1.4mm,width=1.1mm]}, gray!55, line width=0.4pt},
  box/.style={draw, rectangle, rounded corners=1pt, text width=2.45cm,
              align=center, inner xsep=3pt, inner ysep=4pt, line width=0.4pt,
              minimum height=0.9cm},
]

\begin{scope}[shift={(0,0)}]
  \fill[black!6] (1.4,1.70) rectangle (2.9,4.85);
  \draw[black!25] (1.4,1.70) -- (1.4,4.85);
  \draw[black!25] (2.9,1.70) -- (2.9,4.85);
  \node[font=\fontsize{8}{8}\selectfont] at (1.4,5.06) {$y_0$};
  \node[font=\fontsize{8}{8}\selectfont] at (2.9,5.06) {$y_0+l$};
  \draw[ax] (0.10,4.45) -- (3.85,4.45);
  \draw[ax] (0.10,2.55) -- (3.85,2.55);
  \node[anchor=east,font=\fontsize{8}{8}\selectfont] at (3.99,4.45) {$y$};
  \node[anchor=east,font=\fontsize{8}{8}\selectfont] at (3.99,2.55) {$y$};
  \draw[black!40, line width=0.5pt] plot coordinates {(0.250,4.615) (0.263,4.615) (0.277,4.616) (0.290,4.616) (0.303,4.616) (0.317,4.615) (0.330,4.615) (0.343,4.614) (0.357,4.613) (0.370,4.612) (0.383,4.610) (0.397,4.609) (0.410,4.607) (0.423,4.605) (0.436,4.603) (0.450,4.601) (0.463,4.599) (0.476,4.597) (0.490,4.595) (0.503,4.592) (0.516,4.590) (0.530,4.588) (0.543,4.586) (0.556,4.583) (0.570,4.581) (0.583,4.579) (0.596,4.577) (0.610,4.574) (0.623,4.572) (0.636,4.570) (0.650,4.569) (0.663,4.567) (0.676,4.565) (0.690,4.564) (0.703,4.562) (0.716,4.561) (0.730,4.560) (0.743,4.559) (0.756,4.558) (0.769,4.558) (0.783,4.557) (0.796,4.557) (0.809,4.557) (0.823,4.556) (0.836,4.557) (0.849,4.557) (0.863,4.557) (0.876,4.558) (0.889,4.558) (0.903,4.559) (0.916,4.560) (0.929,4.561) (0.943,4.562) (0.956,4.563) (0.969,4.564) (0.983,4.565) (0.996,4.566) (1.009,4.567) (1.023,4.568) (1.036,4.570) (1.049,4.571) (1.063,4.572) (1.076,4.573) (1.089,4.574) (1.103,4.575) (1.116,4.576) (1.129,4.577) (1.142,4.578) (1.156,4.578) (1.169,4.579) (1.182,4.579) (1.196,4.579) (1.209,4.579) (1.222,4.579) (1.236,4.578) (1.249,4.577) (1.262,4.577) (1.276,4.575) (1.289,4.574) (1.302,4.572) (1.316,4.571) (1.329,4.569) (1.342,4.566) (1.356,4.564) (1.369,4.561) (1.382,4.558) (1.396,4.554) (1.409,4.551) (1.422,4.547) (1.436,4.543) (1.449,4.539) (1.462,4.534) (1.475,4.530) (1.489,4.525) (1.502,4.519) (1.515,4.514) (1.529,4.509) (1.542,4.503) (1.555,4.497) (1.569,4.491) (1.582,4.485) (1.595,4.479) (1.609,4.473) (1.622,4.466) (1.635,4.460) (1.649,4.453) (1.662,4.447) (1.675,4.441) (1.689,4.434) (1.702,4.428) (1.715,4.421) (1.729,4.415) (1.742,4.409) (1.755,4.402) (1.769,4.396) (1.782,4.390) (1.795,4.385) (1.808,4.379) (1.822,4.374) (1.835,4.368) (1.848,4.363) (1.862,4.358) (1.875,4.354) (1.888,4.349) (1.902,4.345) (1.915,4.342) (1.928,4.338) (1.942,4.335) (1.955,4.332) (1.968,4.329) (1.982,4.326) (1.995,4.324) (2.008,4.323) (2.022,4.321) (2.035,4.320) (2.048,4.319) (2.062,4.318) (2.075,4.318) (2.088,4.318) (2.102,4.318) (2.115,4.318) (2.128,4.319) (2.142,4.320) (2.155,4.321) (2.168,4.323) (2.181,4.325) (2.195,4.327) (2.208,4.329) (2.221,4.331) (2.235,4.334) (2.248,4.336) (2.261,4.339) (2.275,4.342) (2.288,4.345) (2.301,4.349) (2.315,4.352) (2.328,4.355) (2.341,4.359) (2.355,4.362) (2.368,4.366) (2.381,4.369) (2.395,4.373) (2.408,4.376) (2.421,4.380) (2.435,4.383) (2.448,4.387) (2.461,4.390) (2.475,4.393) (2.488,4.396) (2.501,4.399) (2.514,4.402) (2.528,4.405) (2.541,4.408) (2.554,4.410) (2.568,4.412) (2.581,4.415) (2.594,4.417) (2.608,4.419) (2.621,4.420) (2.634,4.422) (2.648,4.423) (2.661,4.425) (2.674,4.426) (2.688,4.427) (2.701,4.428) (2.714,4.428) (2.728,4.429) (2.741,4.429) (2.754,4.429) (2.768,4.429) (2.781,4.429) (2.794,4.429) (2.808,4.429) (2.821,4.429) (2.834,4.429) (2.847,4.428) (2.861,4.428) (2.874,4.428) (2.887,4.427) (2.901,4.427) (2.914,4.426) (2.927,4.426) (2.941,4.425) (2.954,4.425) (2.967,4.425) (2.981,4.425) (2.994,4.425) (3.007,4.425) (3.021,4.425) (3.034,4.425) (3.047,4.425) (3.061,4.426) (3.074,4.427) (3.087,4.428) (3.101,4.429) (3.114,4.430) (3.127,4.431) (3.141,4.433) (3.154,4.435) (3.167,4.437) (3.181,4.439) (3.194,4.442) (3.207,4.445) (3.220,4.448) (3.234,4.451) (3.247,4.454) (3.260,4.458) (3.274,4.462) (3.287,4.466) (3.300,4.470) (3.314,4.474) (3.327,4.479) (3.340,4.484) (3.354,4.489) (3.367,4.494) (3.380,4.499) (3.394,4.504) (3.407,4.510) (3.420,4.516) (3.434,4.521) (3.447,4.527) (3.460,4.533) (3.474,4.539) (3.487,4.545) (3.500,4.551) (3.514,4.556) (3.527,4.562) (3.540,4.568) (3.553,4.574) (3.567,4.580) (3.580,4.585) (3.593,4.591) (3.607,4.596) (3.620,4.601) (3.633,4.607) (3.647,4.611) (3.660,4.616) (3.673,4.621) (3.687,4.625) (3.700,4.629)};
  \draw[black, line width=0.6pt] plot coordinates {(0.250,2.685) (0.263,2.685) (0.277,2.686) (0.290,2.686) (0.303,2.686) (0.317,2.685) (0.330,2.685) (0.343,2.684) (0.357,2.683) (0.370,2.682) (0.383,2.680) (0.397,2.679) (0.410,2.677) (0.423,2.675) (0.436,2.673) (0.450,2.671) (0.463,2.669) (0.476,2.667) (0.490,2.665) (0.503,2.662) (0.516,2.660) (0.530,2.658) (0.543,2.656) (0.556,2.653) (0.570,2.651) (0.583,2.649) (0.596,2.647) (0.610,2.644) (0.623,2.642) (0.636,2.640) (0.650,2.639) (0.663,2.637) (0.676,2.635) (0.690,2.634) (0.703,2.632) (0.716,2.631) (0.730,2.630) (0.743,2.629) (0.756,2.628) (0.769,2.628) (0.783,2.627) (0.796,2.627) (0.809,2.627) (0.823,2.626) (0.836,2.627) (0.849,2.627) (0.863,2.627) (0.876,2.628) (0.889,2.628) (0.903,2.629) (0.916,2.630) (0.929,2.631) (0.943,2.632) (0.956,2.633) (0.969,2.634) (0.983,2.635) (0.996,2.636) (1.009,2.637) (1.023,2.638) (1.036,2.640) (1.049,2.641) (1.063,2.642) (1.076,2.643) (1.089,2.644) (1.103,2.645) (1.116,2.646) (1.129,2.647) (1.142,2.648) (1.156,2.648) (1.169,2.649) (1.182,2.649) (1.196,2.649) (1.209,2.649) (1.222,2.649) (1.236,2.648) (1.249,2.647) (1.262,2.647) (1.276,2.645) (1.289,2.644) (1.302,2.642) (1.316,2.641) (1.329,2.639) (1.342,2.636) (1.356,2.634) (1.369,2.631) (1.382,2.628) (1.396,2.624) (1.409,2.616) (1.422,2.603) (1.436,2.588) (1.449,2.571) (1.462,2.553) (1.475,2.534) (1.489,2.515) (1.502,2.495) (1.515,2.476) (1.529,2.458) (1.542,2.442) (1.555,2.427) (1.569,2.414) (1.582,2.404) (1.595,2.396) (1.609,2.391) (1.622,2.388) (1.635,2.388) (1.649,2.391) (1.662,2.395) (1.675,2.402) (1.689,2.410) (1.702,2.420) (1.715,2.430) (1.729,2.440) (1.742,2.451) (1.755,2.461) (1.769,2.470) (1.782,2.478) (1.795,2.485) (1.808,2.491) (1.822,2.495) (1.835,2.498) (1.848,2.499) (1.862,2.499) (1.875,2.499) (1.888,2.498) (1.902,2.497) (1.915,2.498) (1.928,2.499) (1.942,2.499) (1.955,2.499) (1.968,2.497) (1.982,2.494) (1.995,2.490) (2.008,2.484) (2.022,2.476) (2.035,2.468) (2.048,2.458) (2.062,2.448) (2.075,2.438) (2.088,2.427) (2.102,2.417) (2.115,2.408) (2.128,2.400) (2.142,2.394) (2.155,2.390) (2.168,2.388) (2.181,2.388) (2.195,2.392) (2.208,2.398) (2.221,2.406) (2.235,2.417) (2.248,2.431) (2.261,2.446) (2.275,2.463) (2.288,2.481) (2.301,2.500) (2.315,2.520) (2.328,2.539) (2.341,2.558) (2.355,2.576) (2.368,2.592) (2.381,2.606) (2.395,2.619) (2.408,2.617) (2.421,2.604) (2.435,2.589) (2.448,2.572) (2.461,2.554) (2.475,2.536) (2.488,2.516) (2.501,2.497) (2.514,2.478) (2.528,2.460) (2.541,2.443) (2.554,2.428) (2.568,2.415) (2.581,2.404) (2.594,2.396) (2.608,2.391) (2.621,2.388) (2.634,2.388) (2.648,2.390) (2.661,2.395) (2.674,2.402) (2.688,2.410) (2.701,2.419) (2.714,2.429) (2.728,2.440) (2.741,2.450) (2.754,2.460) (2.768,2.470) (2.781,2.478) (2.794,2.485) (2.808,2.491) (2.821,2.495) (2.834,2.498) (2.847,2.499) (2.861,2.499) (2.874,2.499) (2.887,2.498) (2.901,2.497) (2.914,2.496) (2.927,2.496) (2.941,2.495) (2.954,2.495) (2.967,2.495) (2.981,2.495) (2.994,2.495) (3.007,2.495) (3.021,2.495) (3.034,2.495) (3.047,2.495) (3.061,2.496) (3.074,2.497) (3.087,2.498) (3.101,2.499) (3.114,2.500) (3.127,2.501) (3.141,2.503) (3.154,2.505) (3.167,2.507) (3.181,2.509) (3.194,2.512) (3.207,2.515) (3.220,2.518) (3.234,2.521) (3.247,2.524) (3.260,2.528) (3.274,2.532) (3.287,2.536) (3.300,2.540) (3.314,2.544) (3.327,2.549) (3.340,2.554) (3.354,2.559) (3.367,2.564) (3.380,2.569) (3.394,2.574) (3.407,2.580) (3.420,2.586) (3.434,2.591) (3.447,2.597) (3.460,2.603) (3.474,2.609) (3.487,2.615) (3.500,2.621) (3.514,2.626) (3.527,2.632) (3.540,2.638) (3.553,2.644) (3.567,2.650) (3.580,2.655) (3.593,2.661) (3.607,2.666) (3.620,2.671) (3.633,2.677) (3.647,2.681) (3.660,2.686) (3.673,2.691) (3.687,2.695) (3.700,2.699)};
  \draw[jfmblue, line width=1.0pt] plot coordinates {(1.902,2.497) (1.915,2.498) (1.928,2.499) (1.942,2.499) (1.955,2.499) (1.968,2.497) (1.982,2.494) (1.995,2.490) (2.008,2.484) (2.022,2.476) (2.035,2.468) (2.048,2.458) (2.062,2.448) (2.075,2.438) (2.088,2.427) (2.102,2.417) (2.115,2.408) (2.128,2.400) (2.142,2.394) (2.155,2.390) (2.168,2.388) (2.181,2.388) (2.195,2.392) (2.208,2.398) (2.221,2.406) (2.235,2.417) (2.248,2.431) (2.261,2.446) (2.275,2.463) (2.288,2.481) (2.301,2.500) (2.315,2.520) (2.328,2.539) (2.341,2.558) (2.355,2.576) (2.368,2.592) (2.381,2.606) (2.395,2.619)};
  \draw[ev,black] (1.00,3.85) -- (1.00,3.05);
  \node[anchor=west,text width=1.55cm,align=left] at (1.14,3.45)
       {instantaneous eddy event};
  \draw[densely dotted] (1.9,2.90) -- (1.9,1.80);
  \draw[densely dotted] (2.4,2.90) -- (2.4,1.80);
  \node[anchor=west] at (0.10,1.34) {$K(y)=y-f(y)$};
  \draw[black!18] (1.4,0.85) -- (2.9,0.85);
  \draw[black, line width=0.5pt] plot coordinates {(1.4,0.85) (1.9,0.65) (2.4,1.05) (2.9,0.85)};
  \node[anchor=west,font=\itshape] at (0.02,5.35) {(a)};
  \node[anchor=west,black!50] at (0.18,4.80) {before};
  \node[anchor=west] at (0.18,2.92) {after};
  \node at (1.65,2.02) {1};
  \node[jfmblue] at (2.15,2.02) {2};
  \node at (2.65,2.02) {3};
  \node[font=\itshape\fontsize{7.2}{8}\selectfont,jfmblue] at (2.15,1.78) {reversed};
  \draw[-{Latex[length=1.2mm,width=0.9mm]},jfmblue,line width=0.5pt]
       (2.52,1.86) .. controls (2.58,2.20) .. (2.33,2.42);
\end{scope}

\begin{scope}[shift={(4.55,0)}]
  \foreach \yy in {1.55, 2.55, 3.55}
    \draw[-{Latex[length=1mm,width=0.8mm]},black!22] (0.75,\yy) -- (3.70,\yy);
  \draw[-{Latex[length=1.4mm,width=1.1mm]}] (0.50,0.75) -- (4.05,0.75);
  \draw[-{Latex[length=1.4mm,width=1.1mm]}] (0.50,0.75) -- (0.50,4.85);
  \node[anchor=north east,font=\itshape\fontsize{8}{8}\selectfont] at (4.05,0.62) {t};
  \node[anchor=south east,font=\itshape\fontsize{8}{8}\selectfont] at (0.40,4.85) {y};
  \foreach \x/\a/\bb in {0.90/1.35/2.05, 1.20/2.65/3.45, 1.55/1.45/2.95,
    2.00/3.05/3.95, 2.40/1.55/2.55, 2.80/2.75/3.85, 3.20/1.25/1.95,
    3.62/2.95/4.05}
    \draw[line width=0.9pt] (\x,\a) -- (\x,\bb);
  \node[anchor=west,font=\itshape] at (0.02,5.35) {(b)};
\end{scope}

\begin{scope}[shift={(9.05,0)}]
  \node[box]              (b1) at (1.35,4.30) {sample next candidate time; advance the line to it};
  \node[box]              (b2) at (1.35,2.66) {sample candidate eddy $(y_0,l)$; infer timescale $\tau$ from the local field};
  \node[box]              (b3) at (1.35,1.30) {accept with probability $\propto 1/\tau$};
  \node[box,fill=black!6] (b4) at (1.35,0.05) {apply event: triplet map $+$ kernels};
  \draw[ev] (b1) -- (b2);
  \draw[ev] (b2) -- (b3);
  \draw[ev] (b3) -- (b4) node[midway,anchor=west,font=\itshape\fontsize{7.2}{8}\selectfont,black!60] {accept};
  \draw[ev] (b3.east) -- ++(0.45,0) |- (b1.east);
  \node[font=\itshape\fontsize{7.2}{8}\selectfont,black!55,rotate=-90] at (2.98,2.66) {reject};
  \draw[ev] (b4.west) -- ++(-0.45,0) |- (b1.west);
  \node[anchor=west,font=\itshape] at (-0.30,5.35) {(c)};
\end{scope}
\end{tikzpicture}
  \caption{The model at a glance.
(\textit{a}) Instantaneous eddy event on $[y_0,y_0+l]$: the triplet map compresses and triplicates the local field; the additive kernel $K(y)$ redistributes energy among components, both conserving momentum and total energy.
(\textit{b}) Eddy events (Poisson), the vertical bars marking the spatial range of each event, punctuating the continuous diffusion and mean-strain forcing (grey arrows).
(\textit{c}) Advancement cycle at a state-dependent rate; the $\kappa^{-5/3}$ cascade emerges as an outcome of repeated sampling, consistent with inertial-range phenomenology.}
  \label{fig:odt-schematic}
\end{figure}
This subsection reviews the ODT model of \citet{Kerstein1999} as extended
to the vector-valued formulation of \citet{Kerstein2001}; nothing in it is original. During a simulated realisation, the model evolves the three-component velocity $v_i(y,t)$,
$i=1,2,3$, on a one-dimensional domain whose coordinate $y$ is identified
with the direction $x_2$. Ensemble statistics are gathered over many independent realisations.
Two interleaved mechanisms advance the state
(figure~\ref{fig:odt-schematic}). Between events the fields
obey ordinary one-dimensional diffusion,
\begin{equation}
  (\partial_t-\nu\,\partial_y^2)\,v_i=0,
  \label{eq:odt-diffusion}
\end{equation}
with kinematic viscosity $\nu$; and a stochastic sequence of
instantaneous \emph{eddy events} represents the effect of individual
turbulent overturns.

Because a continuum rotational motion cannot be represented on a line, an
eddy event is an instantaneous, measure-preserving rearrangement of the
velocity profiles together with an additive change,
\begin{equation}
  v_i(y)\;\longrightarrow\;v_i\!\big(f(y)\big)+c_i\,K(y),
  \label{eq:event}
\end{equation}
where $f$ is the \emph{triplet map}
\citep{Kerstein1991,Kerstein1999}: on an interval $[y_0,y_0+l]$ of size
$l$ it compresses the profile to one third of its length, inserts three
copies, reverses the central copy, and leaves the exterior unchanged.
It is \emph{measure preserving} (the non-local analogue of a
divergence-free rearrangement), \emph{continuous}, and \emph{scale
local} (it changes gradients by at most an order-unity factor)
\citep{Kerstein2001}---the
three attributes whose preservation under an imposed mean strain is
established in \S\,\ref{sec:strain-forcing}. The additive term is built
from the \emph{kernel}
\begin{equation}
  K(y)\equiv y-f(y),
  \qquad
  \int K(y)\,\mathrm{d}y = 0,
  \qquad
  \int K^2(y)\,\mathrm{d}y = \tfrac{4}{27}\,l^{3},
  \label{eq:kernel-identities}
\end{equation}
non-zero only within the eddy interval; the first identity guarantees
that the additive change conserves the momentum of each component, and
the second fixes the energetics of the redistribution.

The kernel term is the model surrogate for the pressure-mediated
redistribution of energy among velocity components
\citep{WunschKerstein2001,Kerstein2001}. Requiring that an event
redistribute kinetic energy among components while conserving the total,
with a rule invariant under permutation of the component indices, gives
\begin{equation}
  \Delta E_i = \alpha\sum_j T_{ij}\,Q_j,
  \qquad
  \mathsf{T}=\frac{1}{2}
  \begin{pmatrix}
    -2 & 1 & 1\\ 1 & -2 & 1\\ 1 & 1 & -2
  \end{pmatrix},
  \qquad
  Q_i\equiv\tfrac{27}{8}\,\rho_0\, l\, v_{i,K}^2,
  \label{eq:redistribution}
\end{equation}
where $v_{i,K}\equiv l^{-2}\!\int v_i(f(y))\,K(y)\,\mathrm{d}y$, $Q_i$ is
the available energy of component $i$ within the eddy, $\rho_0$
the constant density, and $\alpha\in[0,1]$ the parameter setting the
fraction of extractable energy transferred towards isotropy. Solving the
per-component energy balance for the amplitudes yields
\begin{equation}
  c_i = \frac{27}{4l}\left[-v_{i,K}
        +\operatorname{sgn}(v_{i,K})
        \sqrt{v_{i,K}^2+\alpha\sum_j T_{ij}\,v_{j,K}^2}\,\right],
  \label{eq:ci}
\end{equation}
real for $0\le\alpha\le1$. The transfer matrix $\mathsf{T}$ introduces a
tendency to equalise component energies. The value $\alpha=2/3$ maximises
this tendency and is the value typically used in applications. Here, eddies
with this value of $\alpha$ and with $\mathsf{T}$ as shown in
\eqref{eq:redistribution} are termed \emph{isotropic}. For present purposes,
an additional eddy type is introduced in \S\,\ref{sec:eddy-types}.

The construction is a model of the
return-to-isotropy (\emph{slow}) part of the pressure--strain
interaction, a correspondence made precise in \S\,\ref{sec:rapid-slow}.
Standard ODT contains no counterpart of the rapid,
mean-strain-driven part.

The event sequence is a Poisson process whose rate density
$\lambda(y_0,l;t)$ follows the standard eddy time-scale construction
governed by the $Q_i$, with a rate constant, a viscous cutoff and a
large-eddy suppression parameter as the only tunable constants; the
inertial-range cascade and the $\kappa^{-5/3}$ spectrum are \emph{outcomes}
of the sampling, not inputs
\citep{Kerstein1999,StephensLignell2021}. One diagnostic caveat of
\citet{Kerstein2001} is respected throughout: the line velocity does
not advect fluid, so diagonal component spectra and variances are read
from the fields while shear stresses are read from eddy-induced
transport.

\subsection{Two eddy types}\label{sec:eddy-types}

The kernel in \eqref{eq:event} redistributes kinetic
energy among all three velocity components in a manner that tends to
reduce anisotropy. \citet{Fistler2020} generalise it by letting each event
select an \emph{eddy type} that restricts the redistribution to a chosen
subset of the components, where the probabilities of the types are degrees
of freedom that can be adjusted to match the component anisotropy of a
given flow. The present geometry singles out two types, chosen with
probabilities $p$ and $1-p$ respectively. With probability $p$ the event
is \emph{planar}: the kernel couples only the $(v_1,v_2)$ pair---the plane
of straining under the plane strain \eqref{eq:plane-strain}, the
axis-normal plane under the axisymmetric strain \eqref{eq:axi-strain}---at
an exchange fraction $\alpha_\mathrm{p}$, with the third component $v_3$
unaffected ($c_3=0$); the transfer matrix is then
\begin{equation}
  \mathsf{T}_\mathrm{p}=\frac{1}{2}
  \begin{pmatrix}
    -1 & 1 & 0\\ 1 & -1 & 0\\ 0 & 0 & 0
  \end{pmatrix},
  \label{eq:redistribution-planar}
\end{equation}
in place of the isotropic $\mathsf{T}$ of \eqref{eq:redistribution}. The two types are thus distinguished by the difference in the transfer matrix and by the parameter $\alpha_\mathrm{p}$ that does not necessarily equal the isotropic-eddy value $\alpha = 2/3$. With probability $1-p$ the eddy is
\emph{isotropic}, as described in \S\,\ref{sec:odt-review}.

The construction adds two parameters, $\alpha_\mathrm{p}$ and the type probability $p$; $p=0$ recovers the one-type
model. The energy identities \eqref{eq:kernel-identities}--\eqref{eq:ci}
hold event-by-event in either type, each being the same amplitude solve
restricted to its coupled components, so conservation is maintained. The
planar type withholds the $v_3$ sink for the strained in-plane energy;
\S\,\ref{sec:validated} measures what this supplies for the fine-scale upwash
anisotropy and fixes $p$ and $\alpha_\mathrm{p}$ 
against the benchmark. The two eddy types change how each event distributes
energy among the components; like the isotropic kernel, each conserves the
total kinetic energy event-by-event. They act mainly on the fine-scale
component anisotropy (\S\,\ref{sec:validated}), whereas the moment-level
results of \S\,\ref{sec:verification} and the broadband spectral shape of
\S\,\ref{sec:val-spectrum} are governed by the production, the rapid
pressure--strain and the cascade. SC-ODT denotes this two-eddy-type model throughout; the one-eddy-type
specialization is used only in the comparison of these alternatives in
\S\,\ref{sec:validated}.

\subsection{The strain-coupled line equation and its structural
admissibility}\label{sec:strain-forcing}

Standard ODT carries no mean strain: between events the field obeys the
pure diffusion \eqref{eq:odt-diffusion}, and the eddy events supply only
the turbulence-driven cascade and return to isotropy. The strain-coupled
formulation generalises the \emph{continuous} advancement to carry the
distortion by an imposed mean velocity gradient $A_{ij}$, which is
prescribed---the input to the problem---rather than evolved. Two
irrotational homogeneous strains are studied, both admitting exact
rapid-distortion solutions and isolating the physics of interest.
\emph{Plane strain},
\begin{equation}
  \mathsf{A}=\operatorname{diag}(s,-s,0),\qquad s>0,
  \label{eq:plane-strain}
\end{equation}
is the strain of a two-dimensional stagnation region: with $x_1$
tangential to the surface, $x_2$ normal to it and $x_3$ spanwise, the
component whose scales are compressed, and whose amplitude is amplified,
is the wall-normal (upwash) component $u_2$; that amplification is the
pre-impact distortion relevant to leading-edge noise. \emph{Axisymmetric
strain},
\begin{equation}
  \mathsf{A}=\operatorname{diag}(-\tfrac12\gamma,-\tfrac12\gamma,\gamma),
  \qquad \gamma>0,
  \label{eq:axi-strain}
\end{equation}
is the strain of an axisymmetric contraction and
provides an independent test in which two components are amplified and
one suppressed. For both, $\mathsf{A}$ is diagonal in the
$(x_1,x_2,x_3)$ axes, so the line along $x_2$ is a principal axis, as
\S\,\ref{sec:domain-strain} requires. More generally, any diagonal
incompressible strain $\mathsf{A}=\operatorname{diag}(A_{11},A_{22},
-A_{11}-A_{22})$ with $x_2$ a principal axis is a two-parameter family of
which \eqref{eq:plane-strain} and \eqref{eq:axi-strain} are one-parameter
members; the construction below applies throughout it, so the scope for
validation is in principle wider than the two cases examined here. The modelling task is to capture the four processes through which $A_{ij}$ acts: the off-axis production and
the modelled rapid pressure--strain, added to the continuous advancement
below; the on-line production, carried by the eddy events acting on the
full velocity (\S\,\ref{sec:event-strain}); and the wavevector
kinematics, by domain straining (\S\,\ref{sec:domain-strain}).

Between eddy occurrences the molecular evolution is replaced by
\begin{equation}
  \partial_t u_i
  = \nu\,\partial_y^2 u_i \;-\; A^{\perp}_{ij}\,u_j \;+\; B_{ij}\,u_j,
  \qquad
  A^{\perp}_{ij}\equiv A_{ij}-A_{22}\,\delta_{i2}\delta_{j2},
  \label{eq:cont-evolution}
\end{equation}
the line form of the fluctuation momentum equation for a spatially
uniform mean gradient: molecular diffusion as in standard ODT; the
off-axis production $-A^{\perp}_{ij}u_j$---a pointwise linear combination
of the local velocity components, exactly representable on the line---and,
in place of the fluctuating pressure gradient, the modelled rapid
pressure--strain operator $B_{ij}u_j$ constructed in
\S\,\ref{sec:rapid-kernel}. The on-line production $-A_{22}u_2$, the one
channel that varies along the model coordinate, is carried instead by
the eddy events acting on the full velocity (\S\,\ref{sec:event-strain}),
so that the continuous and event-borne parts together reproduce the exact
production $\mathcal{P}_{ij}$ of \eqref{eq:rstransport}. The advancement is
split between the continuous phase \eqref{eq:cont-evolution} and the
discrete event sequence.

The strain forcing preserves the structural properties on which the model
rests; we record the argument as a consistency check. Measure
preservation, scale locality and the kernel identities
\eqref{eq:kernel-identities} are geometric properties of $f$ and $K$
on a fixed domain, and the operator splitting applies the discrete
event exactly as in standard ODT, so they hold verbatim, apart from the two-eddy extension. Both
$A_{ij}$, the imposed mean gradient, and $B_{ij}$, formed from the
line-averaged Reynolds stress (\S\,\ref{sec:line-realisation}), are
uniform along the line, so $\mathcal{A}_{ij}\equiv-A_{ij}+B_{ij}$ is a
constant matrix; over a non-diffusive sub-step it advances the flow state by the $y$-independent linear transformation
\begin{equation}
  u_i(y,t)=\Big[\exp\!\int_{t_0}^{t}\mathcal{A}\,\mathrm{d}t'\Big]_{ij}\,
  u_j(y,t_0),
  \label{eq:strain-map}
\end{equation}
a multiplication of the local velocity vector by a matrix set by the
accumulated strain, uniform along the line. Continuous profiles remain
continuous while the molecular term is smoothing; and the
forcing displaces no fluid along $y$, leaving the domain measure
untouched. The deliberate domain straining is the separate operation
of \S\,\ref{sec:domain-strain}.

\subsection{Rapid and slow pressure--strain, and linear rapid
distortion}\label{sec:rapid-slow}

The production in \eqref{eq:cont-evolution} is exact; the rapid operator
$B_{ij}$ is not. Constructing it requires the Reynolds-stress and
rapid-distortion physics it must reproduce, set out here. The
Reynolds-stress transport equation for a fluctuating field about a
mean with constant gradient $A_{ij}$ reads, in homogeneous turbulence,
\begin{equation}
  \frac{\mathrm{d}R_{ij}}{\mathrm{d}t}
  = \mathcal{P}_{ij} + \Pi_{ij} - \varepsilon_{ij},
  \qquad
  \mathcal{P}_{ij}=-A_{ik}R_{kj}-A_{jk}R_{ki},
  \quad
  \Pi_{ij}=\Big\langle \tfrac{p'}{\rho}
            \big(\partial_j u_i+\partial_i u_j\big)\Big\rangle,
  \label{eq:rstransport}
\end{equation}
with $\varepsilon_{ij}$ the dissipation tensor \citep{Pope2000}. The
fluctuating pressure satisfies a Poisson equation whose source splits
into a part linear in the product of the mean velocity gradient and the
fluctuating velocity gradient and a part quadratic in the
fluctuations \citep{Pope2000},
\begin{equation}
  \frac{1}{\rho}\nabla^2 p'
  = \underbrace{-2\,A_{lm}\,\partial_l u_m}_{\text{rapid}}
    \;-\;
    \underbrace{\partial_l\partial_m\!\left(u_l u_m
     -\langle u_l u_m\rangle\right)}_{\text{slow}} ,
  \label{eq:poisson}
\end{equation}
and the pressure--strain correlation inherits the decomposition
$\Pi_{ij}=\Pi^{(\mathrm r)}_{ij}+\Pi^{(\mathrm s)}_{ij}$ into a
\emph{rapid} part driven by the mean strain and a \emph{slow} part driven
by the turbulence alone. The slow part relaxes the Reynolds stress
towards isotropy---the mechanism that the kernel
\eqref{eq:redistribution}--\eqref{eq:ci} reproduces in model form. The
rapid part is linear in $A_{lm}$ and is determined by the velocity
spectrum tensor $\Phi_{ij}(\boldsymbol{\kappa})$; standard ODT
represents neither $\mathcal{P}_{ij}$ nor $\Pi^{(\mathrm r)}_{ij}$.

Of the two, the production enters exactly, split by component: the
off-axis forcing $-A^{\perp}_{ij}u_j$ of \eqref{eq:cont-evolution} and the
on-line production carried by the events (\S\,\ref{sec:event-strain})
together evolve the Reynolds stress at fixed domain length by
\begin{equation}
  \left.\frac{\mathrm{d}R_{ij}}{\mathrm{d}t}\right|_{\text{prod}}
  = -A_{ik}R_{kj}-A_{jk}R_{ki}
  \;\equiv\; \mathcal{P}_{ij},
  \label{eq:production-rate}
\end{equation}
the production tensor of \eqref{eq:rstransport}. The production is
therefore represented with no model assumption. Because the kernels
\eqref{eq:redistribution}--\eqref{eq:redistribution-planar} and the rapid
operator of \S\,\ref{sec:line-realisation} both conserve the total kinetic
energy, the turbulent kinetic energy $k_t=R_{ii}/2$ of the model obeys
\begin{equation}
  \frac{\mathrm{d}k_t}{\mathrm{d}t}=\mathcal{P}(t)-\varepsilon(t),
  \qquad
  \mathcal{P}(t)=\tfrac12\mathcal{P}_{ii}(t)=-A_{ij}R_{ij}(t),
  \label{eq:tke-budget}
\end{equation}
with the line-averaged molecular dissipation rate
$\varepsilon=\varepsilon_{ii}/2$ as the only
sink. The kernels and the rapid operator enter only the budgets of the
individual component energies $R_{11}/2$, $R_{22}/2$ and $R_{33}/2$, among
which they redistribute energy. Of these redistribution terms, only the
rapid pressure--strain $\Pi^{(\mathrm r)}_{ij}$, determined by the spectrum
tensor, requires a closure; its construction is the subject of
\S\,\ref{sec:rapid-kernel}.

When the strain outruns the eddy turnover---the rapid limit---the
response is given exactly by linear theory
\citep{BatchelorProudman1954,Hunt1973,HuntCarruthers1990}. Writing the
fluctuation as Fourier modes with wavevectors convected by the mean
flow, each mode of time-dependent wavevector $\boldsymbol{\kappa}(t)$ and amplitude
$\hat{u}_i$ evolves, in the inviscid rapid limit, according to
\begin{equation}
  \frac{\mathrm{d}\kappa_i}{\mathrm{d}t}=-A_{ji}\,\kappa_j,
  \qquad
  \frac{\mathrm{d}\hat{u}_i}{\mathrm{d}t}
  =\left(\frac{2\,\kappa_i \kappa_m}{\kappa^2}-\delta_{im}\right)
   A_{mn}\,\hat{u}_n .
  \label{eq:rdt}
\end{equation}
The bracketed operator combines the production $-A_{im}\hat u_m$ with the
rapid pressure response, the solenoidal projection
$P_{ij}=\delta_{ij}-\kappa_i\kappa_j/\kappa^2$ applied to the strained
field; the factor of two is fixed by continuity, since only that
coefficient makes $\mathrm{d}(\kappa_i\hat u_i)/\mathrm{d}t$ vanish under
\eqref{eq:rdt}. Two features are central to what follows. First, the
pressure response is \emph{non-local} and intrinsically
\emph{three-dimensional}: it requires all of
$\kappa_1,\kappa_2,\kappa_3$, whereas a single line resolves the
wavenumber only along its own coordinate, so \eqref{eq:rdt} cannot be
evaluated on the line as written. Second, the rapid/slow decomposition
furnishes the organising principle of the formulation: the slow part is
already carried by the eddy kernel and the production by
\eqref{eq:cont-evolution}; what remains is to construct the rapid
redistribution $B_{ij}$ consistent with \eqref{eq:rdt} from quantities
available on the line---the task of \S\,\ref{sec:rapid-kernel}.

\subsection{The rapid redistribution operator}\label{sec:rapid-kernel}

Were the continuous evolution to consist of production and diffusion only
($B_{ij}=0$), the line would reproduce
$\mathrm{d}R_{ij}/\mathrm{d}t=\mathcal{P}_{ij}$, whereas the exact
inviscid rapid-distortion balance is
$\mathrm{d}R_{ij}/\mathrm{d}t=\mathcal{P}_{ij}+\Pi^{(\mathrm r)}_{ij}$.
The omission is not small: for turbulence isotropic at the onset of
straining the classical result recovered below gives
$\Pi^{(\mathrm r)}_{ij}=-\tfrac35\mathcal{P}_{ij}$, so production acting
alone overstates the initial distortion rate by a factor of $\tfrac52$;
under sustained strain it also lets a component energy grow without the
bound that incompressibility imposes. Representing
$\Pi^{(\mathrm r)}_{ij}$ is necessary.

\subsubsection{The exact rapid pressure--strain and the closure problem}
\label{sec:closure-problem}

The exact rapid term follows from the amplitude equation. Writing
\eqref{eq:rdt} as $\mathrm{d}\hat u_i/\mathrm{d}t=M_{in}\hat u_n$ with
$M_{in}=-A_{in}+2(\kappa_i\kappa_m/\kappa^2)A_{mn}$, and noting that the
incompressible wavevector dynamics preserve the measure
$\mathrm{d}\boldsymbol{\kappa}$ (their phase-space divergence is
$-A_{ii}=0$), the spectrum tensor
$\Phi_{ij}=\langle\hat u_i^{*}\hat u_j\rangle$ yields
\begin{equation}
  \frac{\mathrm{d}R_{ij}}{\mathrm{d}t}
  = \int\!\big(M_{in}\Phi_{nj}+M_{jn}\Phi_{in}\big)\,
    \mathrm{d}\boldsymbol{\kappa}
  = \mathcal{P}_{ij} + \Pi^{(\mathrm r)}_{ij},
  \label{eq:dRdt-rdt}
\end{equation}
in which the production separates from the pressure contribution,
\begin{equation}
  \Pi^{(\mathrm r)}_{ij}
  = 2A_{mn}\big(\mathcal{R}_{ijmn}+\mathcal{R}_{jimn}\big),
  \qquad
  \mathcal{R}_{ijmn} \equiv
  \int \frac{\kappa_i \kappa_m}{\kappa^2}\,
  \Phi_{nj}(\boldsymbol{\kappa})\,\mathrm{d}\boldsymbol{\kappa}.
  \label{eq:exact-rps}
\end{equation}
The rapid term $\Pi^{(\mathrm r)}_{ij}$ is traceless---by solenoidality $\kappa_n\Phi_{nj}=0$, the
trace integrand contains $\kappa_i\Phi_{ni}=0$---so the rapid term
redistributes energy among components without changing the total, in
keeping with the energy-conserving character of the kernel mechanism.
Equation \eqref{eq:exact-rps} also exposes the central difficulty: the
rapid term is a functional of the spectrum tensor through the
orientation-weighted integral $\mathcal{R}_{ijmn}$, not of the
single-point Reynolds stress---the classical closure problem of the rapid
pressure--strain, common to every single-point model
\citep{Crow1968,LaunderReeceRodi1975,Pope2000} and intrinsic to the
physics. A line carries $R_{ij}$ and its own one-dimensional spectra, but
not the distribution of energy over wavevector \emph{orientation} that
$\mathcal{R}_{ijmn}$ requires; a spectral-tensor closure is therefore
unavoidable, and we make it explicit. (What a line's data \emph{do}
determine about this integral, and how sharply, is the analytical
question taken up in \S\,\ref{sec:gate}.)

\subsubsection{Closure of the pressure contribution}\label{sec:iso-closure}

The minimal closure takes the spectrum tensor instantaneously isotropic,
$\Phi_{nj}=\frac{E(\kappa)}{4\pi\kappa^2}
\big(\delta_{nj}-\kappa_n\kappa_j/\kappa^2\big)$ with
$\int E\,\mathrm{d}\kappa=k_t$. With the angular averages
$\overline{e_ie_m}=\tfrac13\delta_{im}$ and
$\overline{e_ie_me_ne_j}=\tfrac1{15}(\delta_{im}\delta_{nj}
+\delta_{in}\delta_{mj}+\delta_{ij}\delta_{mn})$ for
$e_i=\kappa_i/\kappa$,
\begin{equation}
  \mathcal{R}_{ijmn}^{\,\mathrm{iso}}
  = k_t\Big[\tfrac{4}{15}\delta_{im}\delta_{nj}
    -\tfrac1{15}\big(\delta_{in}\delta_{mj}
    +\delta_{ij}\delta_{mn}\big)\Big],
  \label{eq:iso-closure}
\end{equation}
and substitution into \eqref{eq:exact-rps} gives, with the mean strain
rate $S_{ij}=\tfrac12(A_{ij}+A_{ji})$ and $A_{nn}=0$,
\begin{equation}
  \Pi^{(\mathrm r),\mathrm{IP}}_{ij}
  = \tfrac{4}{5}\,k_t\,S_{ij}
  \;=\; -\tfrac{3}{5}\,\mathcal{P}_{ij}
  \quad\text{(isotropic $R_{ij}$)},
  \label{eq:crow}
\end{equation}
which is the exact response of initially isotropic turbulence to a weak rapid
strain \citep{Crow1968} and the content of the
isotropisation-of-production (IP) model with coefficient $\tfrac35$
\citep{Pope2000}---obtained here, with no adjustable constant, as the
isotropic limit of \eqref{eq:exact-rps}.

For anisotropic turbulence $\mathcal{R}_{ijmn}$ must be modelled as a
tensor function of $R_{ij}$; the general representation linear in the
anisotropy is the quasi-isotropic model of \citet{LaunderReeceRodi1975}
(LRR),
\begin{equation}
  \Pi^{(\mathrm r),\mathrm{LRR}}_{ij} = C_2\,k_t S_{ij}
  + C_3\,k_t\big(b_{ik}S_{jk}+b_{jk}S_{ik}
  -\tfrac{2}{3}b_{mn}S_{mn}\delta_{ij}\big)
  + C_4\,k_t\big(b_{ik}W_{jk}+b_{jk}W_{ik}\big),
  \label{eq:lrr}
\end{equation}
with anisotropy $b_{ij}=R_{ij}/(2k_t)-\tfrac13\delta_{ij}$, mean rotation
$W_{ij}=\tfrac12(A_{ij}-A_{ji})$ (zero for the irrotational strains
considered), and $C_2=4/5$, $C_3=7/4$, $C_4=1.31$; the value
$C_2=4/5$ is fixed by requiring that \eqref{eq:lrr} reduce at isotropy
to the exact $\tfrac45k_tS_{ij}$ of \eqref{eq:crow}, so both closures
share the constant-free onset slope. The isotropy-based form
\eqref{eq:crow} is retained only for the onset consistency check
of \S\,\ref{sec:onset}; every simulation result in the paper
uses the LRR closure \eqref{eq:lrr}.

\subsubsection{Realisation on the line}\label{sec:line-realisation}

Given the closed rapid term---symmetric and traceless---we require a
pointwise operator $B_{ij}$ whose action on the line reproduces it.
Demanding that the operator evolve the Reynolds stress by exactly
$\Pi^{(\mathrm r)}_{ij}$ implies
\begin{equation}
  \big(\mathsf{B}(t)\,\mathsf{R}(t)+\mathsf{R}(t)\,\mathsf{B}(t)\big)_{ij}
  = \Pi^{(\mathrm r)}_{ij}(t),
  \label{eq:lyapunov}
\end{equation}
a continuous Lyapunov equation with, for positive-definite
$\mathsf{R}$, a unique symmetric solution $\mathsf{B}$. The solution
conserves kinetic energy automatically---contracting \eqref{eq:lyapunov}
gives $2B_{ij}R_{ij}=\Pi^{(\mathrm r)}_{ii}=0$, so the rate of work done by
the operator on the line vanishes. In the isotropic case
$\mathsf{R}=\tfrac23k_t\mathsf{I}$ it reduces to the closed form
\begin{equation}
  B_{ij} = \tfrac{3}{5}\,S_{ij},
  \qquad
  \mathcal{A}_{ij} = -A_{ij} + \tfrac{3}{5}\,S_{ij}.
  \label{eq:iso-B}
\end{equation}
$B_{ij}$ is evaluated from the running Reynolds stress at each sub-step
and enters the continuous advancement \eqref{eq:cont-evolution}
alongside the production. Phenomenologically it represents the part of
the redistribution driven by the mean gradient, which proceeds
continuously at the strain rate, in contrast to the slow return to
isotropy that the eddy kernel supplies at the eddy rate. Strain-driven effects (production and
rapid pressure--strain) are continuous and enter through
$\mathcal{A}_{ij}$; turbulence-driven effects (the inertial cascade and
the slow return to isotropy) remain relegated to the discrete events.

The sense in which the formulation is rapid-distortion consistent can now
be stated precisely, again as a consistency check rather than a theorem.
By \eqref{eq:production-rate} the production term contributes
$\mathcal{P}_{ij}$ to $\mathrm{d}R_{ij}/\mathrm{d}t$ and by construction
\eqref{eq:lyapunov} the operator contributes $\Pi^{(\mathrm r)}_{ij}$;
their sum is the exact balance \eqref{eq:dRdt-rdt} for the prescribed
closure, with no adjustable constant, recovering \eqref{eq:crow} at the
onset of distortion of isotropic turbulence. Two qualifications delimit
this statement: the consistency holds at the level of the
\emph{component energies} and is \emph{conditional on the
spectral closure}; the isotropic closure is exact only at onset, and its
accuracy for anisotropic states is that of the
isotropisation-of-production approximation it reproduces. The
\emph{spectral} distribution of the distorted energy is not enforced by
\eqref{eq:cont-evolution}: the operator acts uniformly on every resolved
wavenumber, and the distribution of energy across scales emerges from
its interaction with diffusion, the cascade and the domain straining. It
is therefore a genuine prediction, tested in
\S\S\,\ref{sec:verification}--\ref{sec:val-spectrum}---where the
uniform, wavenumber-blind action of $B_{ij}$ is also identified as the
model's measured spectral deficiency (\S\,\ref{sec:validated}).

\subsection{Wavevector kinematics by domain straining}
\label{sec:domain-strain}

The second ingredient of \eqref{eq:rdt}, the wavevector evolution, carries
the kinematic compression and stretching of scales and is independent of
the amplitude evolution: in the method-of-characteristics reading of
linear theory, the wavevectors are the characteristics along which the
amplitudes evolve. On the line the relevant wavenumber is $\kappa_2$. We
restrict attention to a line aligned with a principal axis of the mean
strain, so that $A_{12}=A_{32}=0$ and a material element of the line
remains a coordinate line---as holds for the irrotational mean flow of a
stagnation region. Then
$\mathrm{d}\kappa_2/\mathrm{d}t=-A_{22}\,\kappa_2$, and since a material
element of length $\ell$ stretches as $\dot\ell/\ell=A_{22}$ while
wavenumbers scale as $1/\ell$, a uniform dilatation of the domain
reproduces the wavevector kinematics exactly:
\begin{equation}
  \frac{1}{L}\frac{\mathrm{d}L}{\mathrm{d}t}=A_{22},
  \qquad\Longrightarrow\qquad
  \kappa_2(t)=\kappa_2(t_0)\,
  \exp\!\Big(\!-\!\int_{t_0}^{t}A_{22}\,\mathrm{d}t'\Big),
  \label{eq:dilatation}
\end{equation}
with $L$ the domain length, implemented within the Lagrangian
adaptive-mesh advancement of the model \citep{Lignell2013} by imposing
the mean-flow stretching on the mesh. The line is thus a material line of
the mean flow, contracted by the mean advection along it.

This step is not redundant with the amplitude operator. The Reynolds
stress is an intensive second moment and is invariant under the uniform
dilatation: the dilatation relabels the scales that carry the energy
without changing the component energies, while $\mathcal{A}_{ij}$ changes
the component energies without relabelling scales. The two operations are
the physical-space counterparts of the characteristic motion and the
amplitude evolution in \eqref{eq:rdt}, and together they account for the
mean straining without double counting. Because the dilatation equally
compresses the flow state and the domain, the integral scale remains a
fixed fraction of the domain length under strain---numerically constant to
within $2\%$ over the range studied ($e\le4$)---so the linear strain forcing does
not make the domain size a physical parameter of the strained state over
the reported window. Two limitations are inherent to
representing three-dimensional kinematics on one axis. First, the
transverse wavenumbers also strain but are not represented; their
effect is subsumed in the spectral closure, which already supplies the
orientation information the line lacks. Second, a rotational mean flow,
or a line not aligned with a principal axis, tilts the material line
out of the coordinate direction and lies outside the formulation.

\subsection{The eddy event under strain}\label{sec:event-strain}

The domain straining of \S\,\ref{sec:domain-strain} accounts for the mean
\emph{advection} along the line and carries the wavevector kinematics,
but it does not remove the mean velocity \emph{gradient}: with the line
along $x_2$ and $U_2(y)=A_{22}\,(y-y_c)$ the only mean component that
varies along the model coordinate, the gradient $A_{22}$ remains as the
representation of transportable mean momentum on the line. The off-axis
gradients $A_{11}$ and $A_{33}$ act in directions orthogonal to the line
and present no such variation. This asymmetry fixes how the strain must
enter the two parts of the advancement.

The triplet map rearranges the velocity, because a turbulent overturn
transports total momentum; the event
\eqref{eq:event}--\eqref{eq:redistribution-planar} therefore acts on the
full velocity $v_i=U_i+u_i$, whose on-line component $v_2=U_2+u_2$
contains the mean. In standard ODT the event likewise acts on the full
velocity; for cases in which the mean is uniform along the line, a
Galilean transformation removes it, so that the event may equivalently be applied
to the fluctuation. Under an imposed strain $U_2$ varies along the line,
cannot be removed in this way, and is rearranged by the map together with
the fluctuation. Between events the mean profile is maintained by the
imposed mean pressure gradient, so that after an event the fluctuation,
the departure of $v_2$ from $U_2(y)$, contains the residue
$U_2(f(y))-U_2(y)$. That residue is the down-gradient turbulent transport
of mean momentum---the on-line production---delivered at the scale of the
event rather than uniformly across wavenumber. The off-axis production
($-A_{11}u_1$, $-A_{33}u_3$) and the rapid pressure--strain remain the
continuous forcing of
\S\S\,\ref{sec:strain-forcing}--\ref{sec:rapid-kernel}; accordingly the
production carried continuously in \eqref{eq:cont-evolution} is the
off-axis part, the on-line channel $-A_{22}u_2$ being supplied only by the
events so that it is not counted twice.

Two properties follow. First, \emph{conservation}: because $f$ is measure
preserving and the kernels
\eqref{eq:redistribution}--\eqref{eq:redistribution-planar} conserve the
total kinetic energy, each event leaves unchanged the momentum
$\int v_i\,\mathrm{d}y$ of each component and the total kinetic energy
$\tfrac12\int v_iv_i\,\mathrm{d}y$, summed over $i$, of the full
velocity---the one-dimensional analogues of the Navier--Stokes invariants.
The component energies are not conserved individually, since the kernels
redistribute energy among them. Second, \emph{reduction}: at $A_{ij}=0$
the mean is uniform, the residue vanishes, and standard ODT is recovered
exactly; and the mean shear $U_2$ generates turbulence in initially
laminar flow, so the strained model does not require a seeded
fluctuation field.

The on-line production is delivered scale locally, at the size of each
event, rather than uniformly across wavenumber; \S\,\ref{sec:val-spectrum}
shows this scale locality to be what the strained spectral anisotropy
requires. This is supplemented by the scale-local contribution to
anisotropy by the two eddy types (\S\,\ref{sec:eddy-types}), whose effect is
quantified in \S\,\ref{sec:validated}.

\subsection{Summary of the formulation}\label{sec:summary}

The strain-coupled model advances the line velocity over a sub-step by
composing a \emph{continuous strain--diffusion phase on a dilating
domain} (step~(a)) with \emph{discrete eddy events} (step~(b)). As an
ordered cycle, the operational steps are:
\begin{enumerate}
\item[(a1)] \emph{strain forcing}---advance the velocity profile by the
  $y$-independent linear transformation \eqref{eq:strain-map}, whose generator
  $\mathcal{A}_{ij}=-A^{\perp}_{ij}+B_{ij}$ combines the off-axis production
  $-A^{\perp}_{ij}u_j$ with the modelled rapid pressure--strain $B_{ij}u_j$ of
  \eqref{eq:cont-evolution}, $B_{ij}$ being given by \eqref{eq:lyapunov}
  from the running Reynolds stress and the spectral closure (the LRR
  closure of \S\,\ref{sec:iso-closure}, with the isotropic form
  \eqref{eq:iso-B} its onset limit);
\item[(a2)] \emph{molecular diffusion} \eqref{eq:odt-diffusion}, as in
  standard ODT;
\item[(a3)] \emph{domain straining}---dilate the mesh by
  \eqref{eq:dilatation}, carrying the wavevector kinematics; and
\item[(b)] \emph{discrete eddy events}---sampled as in standard ODT
  followed by eddy-type sampling, applying the triplet map
  and slow kernel \eqref{eq:event}--\eqref{eq:ci}, with the planar
  $\mathsf{T}$ alternative \eqref{eq:redistribution-planar}, to the full
  velocity $v_i=U_i+u_i$, which on the line carries the on-line
  production (\S\,\ref{sec:event-strain}).
\end{enumerate}
Step (a), the continuous advancement (a1)--(a3), carries the off-axis
production and the rapid pressure--strain; step (b), the discrete events,
carries the turbulence-driven physics---the cascade and the slow return
to isotropy---together with the on-line production, the two being
inseparable once the event acts on the full velocity. The strain
coupling thus enters both parts of the advancement while leaving the
event mechanism itself that of standard ODT, here refined by introducing
two eddy types as described in \S\,\ref{sec:eddy-types}. Standard ODT
is recovered exactly when $A_{ij}=0$, with $p$ in \S\,\ref{sec:eddy-types}
likewise set to zero to reflect the higher degree of symmetry in this case.

The relative importance of the two steps is governed by the strain
rapidity, which is the ratio of the
mean strain rate $S\equiv(2S_{ij}S_{ij})^{1/2}$ to the turbulence
relaxation rate $\varepsilon/k_t$, with $\varepsilon$ the line-averaged dissipation
rate,
\begin{equation}
  \chi \equiv \frac{S\,k_t}{\varepsilon}.
  \label{eq:chi}
\end{equation}
For $\chi\gg1$ the distortion accumulates faster than the turbulence can
respond, the eddy events are negligible over the distortion time, and the
formulation reduces to the rapid-distortion limit of step (a), exact for
the component energies under the closure. For $\chi\ll1$ the turbulence
relaxes between increments of strain and the slow kernel dominates. The
stagnation region of a leading edge is the intermediate regime
$\chi=O(1)$, in which neither limit is uniformly valid and both
mechanisms contribute; the fidelity of the formulation across this range
and that of the emergent distorted spectrum are assessed against
rapid-distortion solutions and scale-resolving data in
\S\S\,\ref{sec:verification}--\ref{sec:val-spectrum}.

\subsection{The Rogallo representation}\label{sec:rogallo}

Section~\ref{sec:summary} specifies the SC-ODT model. This section gives an
alternative mathematical representation of exactly the same model: the
transform to the frame
convecting and deforming with the mean
\citep[Appendix~B of][]{Rogallo1981}, applied to the event as well as to
the continuous evolution. In this frame the line coordinate is
$\eta=\mathrm{e}^{-A_{22}t}(y-y_c)$, on which the line is static; the
dilatation \eqref{eq:dilatation} is its image in physical coordinates. The
variables carried on the line are the fluctuations $u_i$ rather than the
full velocity. The continuous advancement is unchanged,
\eqref{eq:cont-evolution} being already written for the fluctuation, and
the force that maintains the mean in the physical frame drops out. The
event \eqref{eq:event} applied to the full velocity gives the post-map
fluctuation $u_2^{+}$ of the on-line component,
\begin{equation}
  u_2^{+}(y)\;=\;v_2\!\big(f(y)\big)\;-\;U_2(y)
  \;=\;u_2\!\big(f(y)\big)\;+\;U_2\!\big(f(y)\big)\;-\;U_2(y),
  \label{eq:rogallo-event}
\end{equation}
with the kernel term added as in \eqref{eq:event} and the event rate and
kernels evaluated using the full pre-map velocity $v_2(y)=u_2(y)+U_2(y)$.
The mean is added to the pre-map fluctuation and subtracted at the post-map
positions, so that the mean profile is preserved and the residue
$U_2(f(y))-U_2(y)$, the on-line production of \S\,\ref{sec:event-strain},
is added to the post-map fluctuation. Equation \eqref{eq:rogallo-event} is the
Rogallo-transformed form of the vector ODT model of \citet{Kerstein2001}.
Owing to its equivalence to the physical-frame event, each Rogallo-frame
event conserves the momentum and the total
kinetic energy of the full velocity exactly, while the fluctuation energy
changes by the on-line production.

All reported SC-ODT results are obtained using the Rogallo frame with the
exception of an independent lab-frame implementation using the map as
shown in \eqref{eq:event}, which reproduces the corresponding Rogallo-frame
results when the latter are back-transformed for comparison (supplementary
material, \S\,S5).

\section{Closure of the rapid pressure--strain term on a single line}\label{sec:gate}

The strain-coupled formulation of \S\,\ref{sec:formulation} evaluates the
rapid pressure--strain redistribution by the algebraic moment closure
\citep{LaunderReeceRodi1975}, applied component by component on the
resolved line. That closure is a modelling input, and the redistribution
it represents is, in its exact form, a \emph{nonlocal, three-dimensional}
transformation: the rapid pressure at a point is a volume integral over the entire
fluctuating field, and the pressure--strain correlation is an integral
over the full three-dimensional spectrum---one that can be evaluated scale
by scale only where the whole field is available, as \citet{ClayYeung2016}
do in a $4096^3$ simulation of axisymmetric contraction, separating its
rapid and slow parts spectrally. What a single one-dimensional line of
turbulence can say about such a transformation is the analytical question of this
paper. It is settled here---before the model is exercised on strained
cases in \S\S\,\ref{sec:verification}--\ref{sec:val-spectrum}---without a
closure claim: this section states what obstructs a single-line
representation, derives bounds on the rapid term from line data alone,
measures the axisymmetry hypothesis under which the term collapses, and
identifies which physical content the moment closure discards.

The section proceeds as follows. The exact term and the obstruction are
stated first (\S\,\ref{sec:exact-pi}--\S\,\ref{sec:obstruction}); the
axisymmetric collapse follows, together with the statement of what it does
\emph{not} deliver (\S\,\ref{sec:assumption}); the bounds that line data do
deliver, with their falsifiers, come next (\S\,\ref{sec:bounds}); the
axisymmetry assumption is then measured directly in the strained-box DNS,
including a control experiment that isolates its failure mode
(\S\,\ref{sec:a1meas}); the component-ordering deficiency raised by the
exact solution is measured and localised (\S\,\ref{sec:order}); the
reduction to the moment closure and the disjointness from the eddy
redistribution close the argument
(\S\,\ref{sec:lrr-reduction}--\S\,\ref{sec:doublecount}).

\subsection{The exact rapid pressure--strain term}\label{sec:exact-pi}

The starting point is the Poisson equation \eqref{eq:poisson} for the
fluctuating pressure, whose source splits into a part linear in the
fluctuations and the mean velocity gradient (the \emph{rapid} part) and a
part quadratic in the fluctuations (the \emph{slow} part)
\citep{Pope2000}. The rapid part is linear in $\boldsymbol{u}'$ and
proportional to the mean gradient, which for the present homogeneous
plane strain is the constant tensor $A_{kl}$. Writing
$p'_{\mathrm r}$ for the rapid pressure and inverting the Laplacian with the
free-space Green's function $G(\boldsymbol{x})=-1/4\pi|\boldsymbol{x}|$,
\begin{equation}
  p'_{\mathrm r}(\boldsymbol{x})
  = \frac{\rho}{2\pi}\,A_{kl}\!\int
    \frac{1}{|\boldsymbol{x}-\boldsymbol{x}'|}\,
    \frac{\partial u_l}{\partial x'_k}(\boldsymbol{x}')\,
    \mathrm{d}\boldsymbol{x}' .
  \label{eq:prapid}
\end{equation}
The rapid pressure--strain correlation that redistributes energy among the
Reynolds-stress components is
$\Pi^{(\mathrm r)}_{ij}
   = \overline{(p'_{\mathrm r}/\rho)\,(\partial u_i/\partial x_j
     + \partial u_j/\partial x_i)}$.
Substituting \eqref{eq:prapid} and using homogeneity to pass to the
velocity spectrum tensor
$\Phi_{mn}(\boldsymbol{\kappa})$
(the Fourier transform of $\overline{u_m(\boldsymbol{x})\,
u_n(\boldsymbol{x}+\boldsymbol{r})}$),
one obtains the standard exact result
\citep{BatchelorProudman1954,Crow1968,HuntCarruthers1990}
\begin{equation}
  \Pi^{(\mathrm r)}_{ij}
  = A_{kl}\,M_{ijkl},
  \qquad
  M_{ijkl}
  = \int_{\mathbb{R}^3}
    \Bigl[\,
      \frac{\kappa_j\kappa_k}{\kappa^2}\,\Phi_{il}(\boldsymbol{\kappa})
      + \frac{\kappa_i\kappa_k}{\kappa^2}\,\Phi_{jl}(\boldsymbol{\kappa})
    \,\Bigr]\,\mathrm{d}\boldsymbol{\kappa},
  \label{eq:Mijkl}
\end{equation}
where $\kappa=|\boldsymbol{\kappa}|$. Equations
\eqref{eq:poisson} and \eqref{eq:Mijkl} are exact; no modelling has entered. The
tensor $M_{ijkl}$---the nonlocal rapid-pressure projection also written
$\mathcal{R}_{ijmn}$ in \eqref{eq:exact-rps}---is a linear functional of the
\emph{full three-dimensional}
spectrum $\Phi_{mn}(\boldsymbol{\kappa})$ through the nonlocal projection
operator $\kappa_a\kappa_b/\kappa^2$, which is the wavenumber-space image of
the inverse Laplacian in \eqref{eq:prapid}. This nonlocality and
three-dimensionality are the obstruction that the single-line model must
confront.

\subsection{The obstruction: what an ODT line resolves}\label{sec:obstruction}

In the present formulation the line carries the fluctuation velocity
$\boldsymbol{u}'(y)$ about a prescribed uniform mean gradient---standard
ODT can also evolve a mean profile---on a single line aligned
with the resolved direction, here taken along the contraction axis of the
strain, $y\equiv x_2$. From it the model has access to the one-dimensional
spectra
\begin{equation}
  \phi_{mn}(\kappa_2)
  = \int_{\mathbb{R}^2}
      \Phi_{mn}(\boldsymbol{\kappa})\,
      \mathrm{d}\kappa_1\,\mathrm{d}\kappa_3,
  \label{eq:phi1d}
\end{equation}
obtained by integrating $\Phi_{mn}$ over the two
unresolved wavenumber components. The integrand of \eqref{eq:Mijkl}, by
contrast, carries the \emph{direction-resolved} weight
$\kappa_a\kappa_b/\kappa^2$, which cannot be reconstructed from
$\phi_{mn}(\kappa_2)$ alone: integrating out $\kappa_1,\kappa_3$ as in
\eqref{eq:phi1d} destroys exactly the angular information that the projection
operator acts on. The reduction $M_{ijkl}\!\to\!$ (functional of
$\phi_{mn}$) is therefore not available without an assumption about how the
energy is distributed over the unresolved wavevector directions. This is the
crux of the present question, and we make the required assumption explicit
rather than implicit.

The obstruction just stated does not disappear under any symmetry
assumption short of full isotropy. As \S\,\ref{sec:assumption} shows,
axisymmetry about the line reduces the unknown spectrum to two scalar
functions of $(\kappa_2,\kappa_\perp)$---but the line data are two weighted
$\kappa_\perp$-integrals of those functions per resolved wavenumber, so the
map from what the line carries to what the rapid term requires retains an
infinite-dimensional null space. The single-line question is therefore to
state this null space and extract what survives it.

\subsection{The closure assumption and the collapsed form}
\label{sec:assumption}

\paragraph{Assumption A1 (axisymmetry of the undistorted spectrum about the
line).} The spectrum tensor of the incoming turbulence, before the rapid
distortion acts, is axisymmetric about the resolved direction $\boldsymbol{e}_2$:
$\Phi_{mn}(\boldsymbol{\kappa})$ depends on $\boldsymbol{\kappa}$ only through
$\kappa_2$ and $\kappa_\perp=(\kappa_1^2+\kappa_3^2)^{1/2}$, and its
polarisation is invariant under rotation about $\boldsymbol{e}_2$.

This is a statement about the \emph{turbulence}, not about the strain. It is
justified for the present configuration because the inflow is grid turbulence,
which is closely isotropic at the grid and is subsequently contracted along a
single axis as it approaches the stagnation plane; isotropy is the special
case of assumption A1 (see Appendix~\ref{app:twoscalar}), and uniaxial contraction preserves axisymmetry about the
contraction axis. The plane strain $A_{kl}=\mathrm{diag}(\tfrac12,-\tfrac12,0)$
acting at the leading edge is itself \emph{not} axisymmetric, so A1 is imposed
on the state on which the rapid operator acts, consistent with the
rapid-distortion ordering in which $\Pi^{(\mathrm r)}$ is evaluated on the
current (here, axisymmetric to leading order) field. The error incurred when
the accumulated distortion has driven the spectrum away from axisymmetry is
quantified in \S\,\ref{sec:a1meas}.

Under A1 the angular integral in \eqref{eq:Mijkl} can be carried out. Writing
$\boldsymbol{\kappa}=(\kappa_\perp\cos\theta,\ \kappa_2,\
\kappa_\perp\sin\theta)$ and using that, by A1, $\Phi_{mn}$ is independent of
$\theta$, the azimuthal averages of the projection weights reduce to
\begin{equation}
  \Bigl\langle \frac{\kappa_1^2}{\kappa^2}\Bigr\rangle_\theta
  = \Bigl\langle \frac{\kappa_3^2}{\kappa^2}\Bigr\rangle_\theta
  = \frac{1}{2}\,\frac{\kappa_\perp^2}{\kappa^2},
  \qquad
  \Bigl\langle \frac{\kappa_2^2}{\kappa^2}\Bigr\rangle_\theta
  = \frac{\kappa_2^2}{\kappa^2},
  \qquad
  \Bigl\langle \frac{\kappa_1\kappa_3}{\kappa^2}\Bigr\rangle_\theta = 0,
  \label{eq:angular}
\end{equation}
with $\kappa^2=\kappa_2^2+\kappa_\perp^2$, and the weights mixing a resolved
and an unresolved index (e.g.\ $\kappa_1\kappa_2/\kappa^2$) averaging to zero.
Substitution of \eqref{eq:angular} into \eqref{eq:Mijkl} removes the azimuthal
dependence and leaves a two-dimensional integral over $(\kappa_2,\kappa_\perp)$
of the axisymmetric spectrum. Defining the diagonal one-dimensional spectra
along the line,
\begin{equation}
  E_n^{\,\|}(\kappa_2)
   = \int_0^\infty \Phi_{nn}(\kappa_2,\kappa_\perp)\,
     2\pi\kappa_\perp\,\mathrm{d}\kappa_\perp
   \quad(\text{no sum}),
   \label{eq:Epar}
\end{equation}
the rapid pressure--strain term reduces to a two-dimensional integral over the
half-plane $(\kappa_2,\kappa_\perp)$ of the axisymmetric spectrum. The
axisymmetric, solenoidal spectrum is fixed by two scalar functions of
$(\kappa_2,\kappa_\perp)$ \citep{Batchelor1946,ChandrasekharAxisym1950}; we
denote them $a(\kappa_2,\kappa_\perp)$ and $c(\kappa_2,\kappa_\perp)$, where
$a$ is the isotropic-like scalar (the sole survivor when the turbulence is
isotropic) and $c$ measures the axisymmetric departure from isotropy. Carrying
out the azimuthal average \eqref{eq:angular} and specialising to the plane
strain $A=\mathrm{diag}(\tfrac12,-\tfrac12,0)$, the diagonal rapid
pressure--strain components are
\begin{equation}
  \Pi^{(\mathrm r)}_{nn}
  = \int_{0}^{\infty}\!\!\int_{0}^{\infty}
      g_{n}(x)\,\bigl[\,a(\kappa_2,\kappa_\perp),\,c(\kappa_2,\kappa_\perp)\,\bigr]\,
      2\pi\kappa_\perp\,\mathrm{d}\kappa_\perp\,\mathrm{d}\kappa_2
   \quad(\text{no sum on }n),
  \label{eq:Piclosed}
\end{equation}
with $x\equiv\kappa_2^2/\kappa^2$, $\kappa^2=\kappa_2^2+\kappa_\perp^2$, and the
kernels, linear in the two spectral scalars,
\begin{align}
  g_{1}(x) &= a\!\left(\tfrac18 + \tfrac34 x - \tfrac78 x^2\right)
            + c\,\tfrac78\,x\,(1-x)^2, \label{eq:g1}\\[2pt]
  g_{2}(x) &= -\tfrac32\,x\,\bigl[\,a\,(1-x) + c\,(1-x)^2\,\bigr],
            \label{eq:g2}\\[2pt]
  g_{3}(x) &= a\!\left(-\tfrac18 + \tfrac34 x - \tfrac58 x^2\right)
            + c\,\tfrac58\,x\,(1-x)^2. \label{eq:g3}
\end{align}
Equations \eqref{eq:Piclosed}--\eqref{eq:g3} constitute the axisymmetric collapse:
\emph{under A1, the rapid pressure--strain term is a
functional of the axisymmetric spectral scalars $a,c$ on the resolved
half-plane and the mean strain}, with no reference to the unresolved azimuthal
structure beyond what A1 fixes. The kernels satisfy two exact constraints that
verify the reduction. First, energy conservation: pressure--strain
redistributes without changing the trace, and indeed
$g_1(x)+g_2(x)+g_3(x)=0$ identically for all $x$ and for both scalars, so
$\Pi^{(\mathrm r)}_{11}+\Pi^{(\mathrm r)}_{22}+\Pi^{(\mathrm r)}_{33}=0$.
Second, the isotropic limit: setting $c=0$ with $a$ a function of $\kappa$
alone, the integral of $g_3$ over the half-plane vanishes, so
$\Pi^{(\mathrm r)}_{33}=0$---as it must, since the plane strain has no
component along $x_3$ and the isotropic rapid term is proportional to $S_{ij}$
\citep{Crow1968}. Both constraints hold identically in the derived kernels,
confirming the tensor reduction.

The status of \eqref{eq:Piclosed}--\eqref{eq:g3} must be stated precisely,
because it is here that the closure question is decided. Under A1 the rapid
term is a linear functional of the two axisymmetric scalars
$a(\kappa_2,\kappa_\perp)$ and $c(\kappa_2,\kappa_\perp)$ on the resolved
half-plane---a \emph{two-dimensional} representation, not yet a single-line
closure, since the line determines only the $\kappa_\perp$-integrals
\eqref{eq:Epar} of fixed combinations of $a$ and $c$ and not the functions
themselves. 
Equations \eqref{eq:Piclosed}--\eqref{eq:g3} are therefore the collapsed form on which any single-line statement must be built. What the line data actually determine—and how sharply—is established next.

\subsection{What line data determine: sharp bounds on the rapid term}
\label{sec:bounds}

Under A1 the solenoidal spectrum tensor takes the two-scalar form
$\Phi_{mn}=a\,P_{mn}+c\,\xi_m\xi_n$, with $P_{mn}$ the solenoidal projector
and $\xi_m$ the projected axis direction; realisability is the pair of
pointwise conditions $a\ge0$ and $a+(1-x)\,c\ge0$ with
$x=\kappa_2^2/\kappa^2$. The line carries $\phi_{11}(\kappa_2)
=\phi_{33}(\kappa_2)$ and $\phi_{22}(\kappa_2)$: per resolved wavenumber,
two weighted $\kappa_\perp$-integrals of the two unknown functions. The
rapid term \eqref{eq:Piclosed} is linear in $(a,c)$, so its extreme values
over everything consistent with the line data, solenoidality and
realisability are obtained from a family of linear programs, one per
$\kappa_2$, yielding sharp bounds
$[\Pi^{-}_{nn},\,\Pi^{+}_{nn}]$ and a kinematic indeterminacy
$\mathcal I_n=(\Pi^{+}_{nn}-\Pi^{-}_{nn})/2k_tS$ that quantifies exactly
how much the null space costs. The forward map from $(a,c)$ to the line
spectra and the derivation of the bound kernels are given in
Appendix~\ref{app:forward}.

Evaluated on an isotropic von K\'arm\'an spectrum (after checks of
tracelessness, the isotropic angular averages
$+\tfrac15:-\tfrac15:0$, the longitudinal--transverse relation and the
exact $\tfrac45 k_t S_{ij}$), the indeterminacy is

\begin{table}
  \centering\small
  \begin{tabular}{lccc}
  constraint level & $\mathcal I_1$ & $\mathcal I_2$ & $\mathcal I_3$\\[2pt]
  line data $+$ solenoidality $+$ realisability & 0.77 & 1.28 & 0.86\\
  $+$ log-Lipschitz slope cap ($\lambda=4$) & \multicolumn{3}{c}{$<5\%$ narrower}\\
  $+$ polarisation cap $|c|\le\gamma a$, $\gamma=0$ (i.e.\ $c=0$) & 0.34 & 0.58 & 0.24\\
  \end{tabular}
  \caption{Kinematic indeterminacy
  $\mathcal I_n=(\Pi^{+}_{nn}-\Pi^{-}_{nn})/2k_tS$ of the three rapid
  pressure--strain components under successively tighter side conditions,
  calibrated on an isotropic von K\'arm\'an spectrum. Line data with
  solenoidality and realisability leave an $O(1)$ indeterminacy; the
  log-Lipschitz smoothness cap barely narrows it, whereas the polarisation
  cap ($\gamma=0$, i.e.\ $c=0$) roughly halves it.}
  \label{tab:indeterminacy}
\end{table}
Two features of table~\ref{tab:indeterminacy} matter. Smoothness assumptions buy almost
nothing: the null space is not a roughness artefact. The polarisation cap
is the only side condition that narrows the bound appreciably;
\S\,\ref{sec:a1meas} shows that under strain the effective polarisation
reaches order one, so in the strained regime only the uncapped bound applies.

As shown in Appendix~\ref{app:forward}, the representation also supplies two \emph{free falsifiers} that the line
itself reports: under A1, $\phi_{11}(\kappa_2)=\phi_{33}(\kappa_2)$ at
every wavenumber (T1) and the velocity cross-spectra vanish (T2). Any
measured transverse splitting is a wavenumber-resolved violation of A1;
rather than assuming the hypothesis, the line monitors it.

\subsection{The axisymmetry hypothesis, measured}\label{sec:a1meas}

A1 is exact for the isotropic inflow and is degraded by the plane strain
itself. Rather than estimating this error perturbatively, we measure it in
the strained-box DNS---a pseudo-spectral $128^3$ Rogallo-frame simulation
under plane strain, specified in Appendix~\ref{app:dns} and used again for
the spectral validation of \S\,\ref{sec:val-spectrum}---where the full
spectrum is available: the exact azimuthal decomposition of $\Phi_{22}$ yields the true
scalars $(a,c)$, the discarded $m=2$ azimuthal residue (the direct
A1-violation measure, with isotropic sampling floor $0.133$ in this box),
and the effective polarisation $\gamma_{\rm eff}=|c|/a$.

\begin{table}
  \centering\small
  \begin{tabular}{c|cccc|cccc}
  & \multicolumn{4}{c|}{$\chi=0.8$} & \multicolumn{4}{c}{$\chi=16$}\\
  $e$ & $b_{22}$ & $\gamma_{\rm eff}$ & $m{=}2$ & $\tfrac{\phi_{11}}{\phi_{33}}{-}1$
      & $b_{22}$ & $\gamma_{\rm eff}$ & $m{=}2$ & $\tfrac{\phi_{11}}{\phi_{33}}{-}1$\\[1pt]
  0    & 0.003 & 0.10 & 0.133 & $+0.06$ & 0.003 & 0.10 & 0.133 & $+0.06$\\
  0.5  & 0.064 & 0.79 & 0.166 & $-0.20$ & 0.070 & 1.28 & 0.247 & $-0.21$\\
  1    & 0.125 & 1.44 & 0.189 & $-0.32$ & 0.124 & 4.22 & 0.396 & $-0.46$\\
  \end{tabular}
  \caption{A1 measured under plane strain ($128^3$ DNS, four realisations):
  upwash anisotropy, effective polarisation, azimuthal $m=2$ residue of
  $\Phi_{22}$ (isotropic floor $0.133$), and the on-line falsifier T1.}
  \label{tab:a1meas}
\end{table}

Three conclusions follow (table~\ref{tab:a1meas}). First, the A1 error is
\emph{not} perturbatively small in the accumulated anisotropy: the residue
grows from its floor to $0.19$ at $\chi=0.8$ and to $0.40$ at
$16$ by $e=1$---larger at rapid strain, because at moderate strain rate the
nonlinear dynamics partly restore axisymmetry while the strain destroys it.
An error estimate quadratic in the anisotropy ($O(b^2)$), which the symmetric structure of the discarded
harmonic suggests, is superseded by this measurement.
Second, the
polarisation cap that narrows the isotropic bounds is illegitimate under
strain: $\gamma_{\rm eff}$ passes unity by $e\simeq0.5$ at $\chi=16$ and
before $e=1$ at $\chi=0.8$. Third,
correspondingly, the uncapped bound brackets the exact rapid term
(evaluated by direct integration in the DNS) up to $e\approx0.25$--$0.5$
and fails beyond, and the failure is A1's, not the estimator's.

That attribution is tested by a control experiment: the
same box under \emph{axisymmetric compression along the line} (strain
axisymmetric about $\boldsymbol{e}_2$) with matched $A_{22}$, for
which A1 is exact by symmetry. There the $m=2$ residue stays at the
sampling floor and the transverse splitting does not grow, decaying from
its initial sampling offset ($+0.06$) to zero by $e=1$, even at $\chi=16$,
while under plane strain both grow
(figure~\ref{fig:a1control}). Assumption~A1 is exact, by symmetry, for
distortion axisymmetric about the line: with isotropic initial statistics,
the equations of motion and the initial ensemble are invariant under
rotation about $\boldsymbol{e}_2$, so the ensemble statistics remain
axisymmetric about $\boldsymbol{e}_2$ at every strain. The control
experiment confirms this to within the sampling floor. The plane-strain
error of the collapsed form is a strain-geometry property, now quantified, rather than a
defect of the line representation. For a rounded three-dimensional
stagnation region the strain is axisymmetric and A1 holds; for a
two-dimensional leading edge the strain is plane and the correction
measured here applies. The acoustically relevant amplification
$b_{22}(e)$, finally, differs by less than $6\%$ between the two geometries
($0.124$ against $0.117$ at $e=1$, $\chi=16$), so the quantity the
leading-edge response uses is insensitive to the geometry that controls A1.

\begin{figure}
  \centering
  \includegraphics[width=\textwidth]{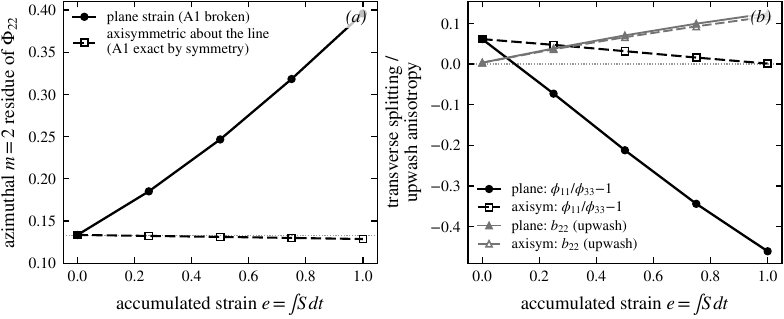}
  \caption{The A1 control experiment in the strained-box DNS at $\chi=16$
  ($128^3$, four realisations): plane strain (solid, filled circles)
  against axisymmetric compression along the line at matched $A_{22}$
  (dashed, open squares). (\textit{a}) Azimuthal $m=2$ residue of
  $\Phi_{22}$ (the $\cos2\theta$ harmonic, $\theta$ the angle about the line
  $x_2$ in the $(x_1,x_3)$ plane): it grows under plane strain but stays at the isotropic
  floor (grey dotted) under axisymmetric strain. (\textit{b}) Transverse
  splitting $\phi_{11}/\phi_{33}-1$ for each geometry, and the upwash
  anisotropy $b_{22}$ (grey triangles, filled and solid for plane strain,
  open and dashed for the axisymmetric case; nearly indistinguishable
  between geometries)---the azimuthal structure that A1 assumes away is what differs,
  while $b_{22}$ is the same in both to within $6\%$.}
  \label{fig:a1control}
\end{figure}

\subsection{Component ordering under plane strain}\label{sec:order}

The exact solution of the rapid problem under plane strain has the neutral
spanwise component eventually overtaking the upwash; algebraic closures
predict the opposite. The exact rapid term evaluated in the DNS settles the
matter: at $e=1$, $\Pi^{(\mathrm r)}_{33}/k_tS = 0.14$ at
$\chi=0.8$ and $0.27$ at $16$, exceeding
$\Pi^{(\mathrm r)}_{11}=0.23$ in the rapid case---the overtaking is real
and sets in at rapid strain. LRR, whose spanwise component is carried by
its $b$-dependent term, matches at $\chi=0.8$
($\Pi^{(\mathrm r)}_{33}=0.143\,k_tS$) but delivers half the exact value at
$16$; IP gives zero identically. The deficiency thus lives in the moment
closure's $b$-term, not in the line model: the ODT moment-level $b_{ij}$
follows the closure to $\sim2\%$, and its $b_{33}$ stays within $0.006$ of
zero in every run. In the spectral representation the ordering is carried
by the $c$-scalar through the kernels
\eqref{eq:g1}--\eqref{eq:g3}; the moment reduction omits it.

\subsection{Reduction to the moment closure}\label{sec:lrr-reduction}

The collapsed form \eqref{eq:Piclosed}--\eqref{eq:g3} is a functional of
the axisymmetric spectral scalars, whereas the distortion of \S\,\ref{sec:val-spectrum} is
obtained with the LRR closure of \S\,\ref{sec:rapid-kernel}. For the two to
be consistent, the spectral kernel must reduce to that closure in the limit where the moment
closure is exact---isotropic turbulence. We show that it does, which both
justifies the closure used in the simulations and fixes the prefactor in an explicit convention (the coefficient of \eqref{eq:crow-spectral} and the remark following it).

Figure~\ref{fig:pi-kernels} plots the kernels $g_1,g_2,g_3$ of
\eqref{eq:g1}--\eqref{eq:g3} against the angular variable
$x=\kappa_2^2/\kappa^2$. Two features are exact and hold for both spectral
scalars. The kernels sum to zero at every $x$,
\begin{equation}
  g_1(x)+g_2(x)+g_3(x)=0 \qquad \forall x,
  \label{eq:gtrace}
\end{equation}
so the rapid term conserves turbulent kinetic energy pointwise in wavenumber,
not merely in the integral; and each kernel vanishes at $x=1$, where the
wavevector aligns with the line and the projection $\kappa_a\kappa_b/\kappa^2$
degenerates. The signs are physically ordered: $g_1>0$ (energy fed to the
streamwise component $E_1$), $g_2<0$ (energy removed from the line/upwash
component $E_2$) over the bulk of the angular range, with $g_3$ changing sign,
consistent with the plane strain stretching $x_1$ and compressing $x_2$.

For \emph{isotropic} turbulence the anisotropy scalar vanishes ($c=0$) and the
isotropic-like scalar $a$ depends only on $\kappa$. Averaging the kernels over
solid angle---equivalently, integrating $g_n(x)$ against the isotropic
distribution of $x$ on the sphere---gives
\begin{equation}
  \langle g_1\rangle : \langle g_2\rangle : \langle g_3\rangle
  = +\tfrac15 : -\tfrac15 : 0,
  \label{eq:giso}
\end{equation}
in the ratio $+1:-1:0$. This is exactly the structure of the plane strain
$S_{ij}=\mathrm{diag}(\tfrac12,-\tfrac12,0)$: the rapid term is proportional to
$S_{ij}$, with $\Pi^{(\mathrm r)}_{33}=0$ because the strain has no component
along $x_3$. The spectral kernel therefore collapses, under isotropy, to the
isotropic-production form of the moment closure,
\begin{equation}
  \Pi^{(\mathrm r)}_{ij}
  = \tfrac{4}{5}\,k_t\,S_{ij}
  \qquad(\text{isotropic limit}),
  \label{eq:crow-spectral}
\end{equation}
the result of \citet{Crow1968} that underlies the isotropisation-of-production
(IP) rapid term of the LRR closure. We note a
convention: with the symmetrised definition of the pressure--strain correlation
adopted here and the factor of two carried in the Poisson source
\eqref{eq:poisson}, the kernels \eqref{eq:g1}--\eqref{eq:g3} yield the
coefficient $\tfrac25 k_t$ for the $11$-component directly (since $S_{11}=\tfrac12$),
equivalent to the standard $\tfrac45 k_t S_{ij}$ of \eqref{eq:crow-spectral}; the ratio
\eqref{eq:giso} and the tracelessness \eqref{eq:gtrace} are independent of this
convention.

The consequence is that the algebraic closure used to obtain the distortion of
\S\,\ref{sec:val-spectrum} is not an independent assumption but the
isotropic-limit reduction of the line-consistent spectral term derived here.
The spectral form \eqref{eq:Piclosed} carries, in addition, the angular
($x$-dependent) and anisotropy ($c$-scalar) information that the moment closure
omits---and the effects of both omissions are measured: the wavenumber dependence
of the linear response in the three-way comparison of
\S\,\ref{sec:validated}, and the component ordering in
\S\,\ref{sec:order}. The moment closure and the collapsed form thus
agree to leading order, and the terms in which they differ are exactly
the deficiencies measured in \S\,\ref{sec:val-spectrum}.

\begin{figure}
  \centering
  \includegraphics[width=0.72\textwidth]{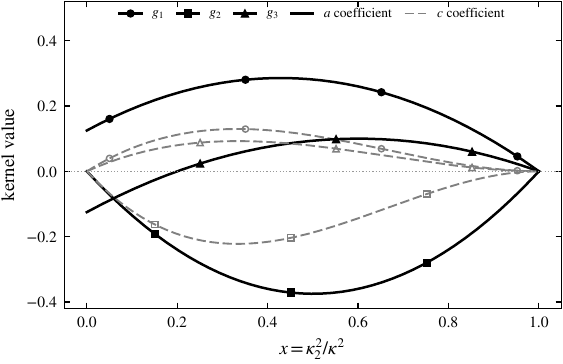}
  \caption{Rapid pressure--strain closure kernels $g_1$ (circles), $g_2$
  (squares), $g_3$ (triangles) (eqs.~\eqref{eq:g1}--\eqref{eq:g3}) against
  the angular variable $x=\kappa_2^2/\kappa^2$. Solid curves: coefficient
  of the isotropic-like spectral scalar $a$; faint grey dashed curves:
  coefficient of the anisotropy scalar $c$. The kernels sum to zero at every $x$ (energy conservation,
  eq.~\eqref{eq:gtrace}) and vanish at $x=1$. Averaged over an isotropic
  spectrum they give the ratio $+\tfrac15:-\tfrac15:0$
  (eq.~\eqref{eq:giso}), recovering the isotropisation-of-production form
  $\Pi^{(\mathrm r)}_{ij}=\tfrac45 k_t S_{ij}$ \citep{Crow1968} that underlies
  the moment closure used in \S\,\ref{sec:val-spectrum}.}
  \label{fig:pi-kernels}
\end{figure}

\subsection{Consistency with the eddy redistribution: avoiding double
counting}\label{sec:doublecount}

A closed rapid term is necessary but not sufficient: it must not duplicate
redistribution already effected by the eddy events. The pressure--strain
correlation decomposes canonically into rapid and slow parts,
$\Pi_{ij}=\Pi^{(\mathrm r)}_{ij}+\Pi^{(\mathrm s)}_{ij}$, corresponding to the two source terms in \eqref{eq:poisson}. The slow part
$\Pi^{(\mathrm s)}_{ij}$ is the return-to-isotropy mechanism, quadratic in the
fluctuations and independent of the mean gradient; the eddy events in ODT,
which act on the fluctuating field without reference to $A_{kl}$, are the
model's representation of exactly this slow redistribution. The rapid part
$\Pi^{(\mathrm r)}_{ij}$, proportional to $A_{kl}$ and linear in the
fluctuations, is by construction absent from the eddy mechanism, which carries
no dependence on the mean strain. The two therefore occupy disjoint terms of
the canonical decomposition, and adding the rapid term as an explicit
mean-strain-driven source alongside the eddies, as the operator $B_{ij}$ of
\eqref{eq:lyapunov} does, does not double-count, provided
the eddy events are calibrated to reproduce the slow return-to-isotropy in the
absence of mean strain---a condition that the no-strain baseline of
\S\,\ref{sec:spec-bvc} verifies directly, since there the centroid relaxes
under the eddies alone with no spurious rapid contribution. This argument establishes consistency at the level of the decomposition; the no-strain baseline of \S\,\ref{sec:spec-bvc} provides the direct numerical confirmation, the centroid there relaxing under the eddy mechanism alone.

The disjointness argument is now also a measured statement, from both
sides. With the strain switched off the eddy events generate no spurious
strain-like signal (the nominal-turbulence null of
Appendix~\ref{app:odtconf}). Conversely, in the rapid limit
($\chi=16$, less than one eddy per line in $e\le1$) every
kernel-energy allocation rule tested---the isotropic kernel, a
strain-biased generalisation, and distinct eddy types---reproduces the
exact rapid-distortion evolution of $b_{22}(e)$ to $\pm0.003$: the rapid
term cannot be reached from inside the eddy events, and the strain forcing
of step~(a1) in \S\,\ref{sec:summary} carries it alone.

\subsection{Summary: what the single-line treatment establishes}
\label{sec:gate-summary}

\begin{enumerate}
\item The rapid term is exactly $A_{kl}M_{ijkl}$; its single-line
  representation is obstructed by an azimuthal null space that is stated
  exactly and survives the axisymmetric collapse at the level of the two
  spectral scalars.
\item Under A1 the rapid term collapses to the kernel form
  \eqref{eq:Piclosed}--\eqref{eq:g3}; A1 holds exactly, by symmetry, for
  distortion axisymmetric about the line, and is \emph{measured} under plane strain, where its violation
  grows with accumulated strain and with strain rapidity
  (table~\ref{tab:a1meas}).
\item Independently of whether A1 holds, line data bound the rapid term
  sharply; the bound is valid---and verified to bracket the exact
  term---up to $e\approx0.25$--$0.5$ under plane strain; under distortion
  axisymmetric about the line it is valid by construction, A1 holding there
  at all computed strains (\S\,\ref{sec:a1meas}).
\item The moment closure used in the simulations is the isotropic reduction
  of the collapsed form (verified, including the exact
  $\tfrac45k_tS_{ij}$); what it discards is the wavenumber dependence of
  the linear response and the component ordering, measured in
  \S\,\ref{sec:val-spectrum} and \S\,\ref{sec:order} respectively.
\item Rapid and slow redistribution are disjoint: verified by the
  strain-off null and by the rapid-limit invariance of every tested kernel-energy
  allocation rule.
\item The acoustically relevant upwash amplification $b_{22}(e)$ is linear
  to $e=1$ and insensitive to the strain geometry that controls the validity of A1.
\end{enumerate}

\section{Verification, and moment-level validation}\label{sec:verification}

We separate two questions. \emph{Verification} asks whether the
formulation and its implementation reproduce a known exact solution in
the regime where one exists---the rapid-distortion limit, tested
analytically and by the numerical solution below. \emph{Validation} asks whether the
model reproduces the behaviour of real strained turbulence outside that
limit; the spectral validation, against the strained-box DNS with its
exact-RDT companion, is the subject of \S\,\ref{sec:val-spectrum}, and
this section closes with the moment-level validation against the
strained-turbulence DNS of \citet{LeeReynolds1985}. Every test in this
section operates at the Reynolds-stress level, by design; no spectral
claim rests on it, and none is made before the spectral benchmarks of
\S\,\ref{sec:val-spectrum}.

One point of identification should be kept in view throughout this section
and the next. The model under test is the strain-coupled formulation
(SC-ODT); \emph{standard} ODT is its $A_{ij}=0$ limit
(\S\,\ref{sec:summary}) and, having no mean-strain coupling---neither
the production forcing, the rapid pressure--strain operator, nor the
domain dilatation---produces none of the strained curves shown here.
Every departure from isotropy in what follows is generated by the added
strain coupling, not by the eddy dynamics of standard ODT, which
supplies only the cascade and the slow return to isotropy. Where the
figures refer to ``SC-ODT'' they mean this formulation; ``standard ODT''
is its strain-off ($A_{ij}=0$, $p=0$) limit. As mentioned earlier, two
eddy types are the SC-ODT default, except that one eddy type (as in
standard ODT) is introduced as a comparison case in \S\,\ref{sec:validated}.

\subsection{The exact onset rate}\label{sec:onset}

The two canonical strains of \S\,\ref{sec:strain-forcing}---plane strain
\eqref{eq:plane-strain} and axisymmetric strain
\eqref{eq:axi-strain}---admit exact rapid-distortion solutions and
isolate the physics of interest.

For turbulence isotropic at the onset of straining the continuous
evolution \eqref{eq:cont-evolution} gives, by
\eqref{eq:production-rate} and \eqref{eq:crow},
\begin{equation}
  \left.\frac{\mathrm{d}R_{ij}}{\mathrm{d}t}\right|_{0}
  = \mathcal{P}_{ij}+\Pi^{(\mathrm r),\mathrm{IP}}_{ij}
  = \tfrac{2}{5}\,\mathcal{P}_{ij}
  = -\tfrac{8}{15}\,k_t\,S_{ij},
  \label{eq:onset-rate}
\end{equation}
the exact rapid-distortion rate for initially isotropic turbulence, with
no adjustable constant. This is a
closed-form prediction the model must satisfy identically. Evaluated for the two canonical strains it gives, under
plane strain~\eqref{eq:plane-strain}, onset rates $\mp\tfrac{8}{15}k_ts$ for the stretched and
compressed components with the neutral component unchanged, and, under
axisymmetric strain~\eqref{eq:axi-strain}, $+\tfrac{4}{15}k_t\gamma$ for each lateral
component against $-\tfrac{8}{15}k_t\gamma$ for the axis. For plane
strain the predicted upwash growth rate $+\tfrac{8}{15}k_ts$ is to be
compared with the $\tfrac43k_ts$ that production alone would give---the
overstatement by a factor of $\tfrac52$ anticipated in
\S\,\ref{sec:rapid-kernel}. A model that satisfies \eqref{eq:onset-rate}
for both strains has its production and rapid terms verified jointly;
failure would indicate an error in the operator $\mathcal{A}_{ij}$ or
its implementation, isolated from any closure or cascade effect.

\subsection{Finite total strain, and the necessity of the rapid term}
\label{sec:finite-strain}

Beyond onset the turbulence becomes anisotropic and the rapid term
\eqref{eq:exact-rps} ceases to be a function of the Reynolds stress
alone, so the verification becomes a test of the spectral closure rather
than of the construction. For each canonical strain we integrate the
Reynolds-stress equation~\eqref{eq:rstransport} forward in time, closing
the rapid pressure--strain two ways---the isotropy-based form
\eqref{eq:crow} (IP) and the anisotropic closure \eqref{eq:lrr} (LRR)
of \S\,\ref{sec:iso-closure}---to obtain the component energies
$R_{ij}(e)$ as functions of the accumulated total strain
$e=\int_0^{t}S\,\mathrm{d}t'$, plotted below. The benchmark is the exact rapid-distortion
solution, obtained by integrating the spectral equations \eqref{eq:rdt}
over a Gauss--Legendre ensemble of initial wavevector directions on the
unit sphere: for initially isotropic turbulence the modal dynamics depend
only on direction, so the Reynolds stress is the spherical average of the
evolved modal covariance. This integration is independent of the IP and
LRR closures and reproduces the onset slope \eqref{eq:onset-rate} to within
$10^{-16}$, verifying the rapid-distortion machinery itself.

Figure~\ref{fig:rdt-components} reports the component-energy fractions
against total strain. Three observations follow. First, at onset all
curves for each component are tangent---the constant-free identity
\eqref{eq:onset-rate} is satisfied---and through $e\lesssim1$ the
closures are nearly indistinguishable from the exact solution. Second,
as anisotropy accumulates the closures depart, LRR uniformly less than
IP: for axisymmetric strain at $e=4$ the lateral fraction is $0.498$
(exact), $0.468$ (LRR), $0.450$ (IP), and the axial fraction $0.003$,
$0.063$, $0.100$. This gap is the closure error in
isolation---the eddy events are inactive in the rapid limit---and it
sets the accuracy of the single-point rapid operators considered here. Third, for plane strain both closures miss
a \emph{qualitative} feature: in the exact solution the neutral spanwise
component grows monotonically, overtakes the upwash near $e\approx3.3$,
and is the largest component by $e=4$
($\overline{u_3^2}/2k_t=0.51$ against $0.47$), whereas both closures
predict it to decay. This redistribution into the neutral direction is
governed by the orientation distribution of the spectral tensor and
cannot be recovered from single-point quantities; it is the moment-level
signature of precisely the transverse-wavenumber information a single
line does not carry. Section~\ref{sec:order} localises this deficiency, with
the strained-box DNS, in the closure's $b$-dependent term.

\begin{figure}
  \centering
  \includegraphics[width=\textwidth]{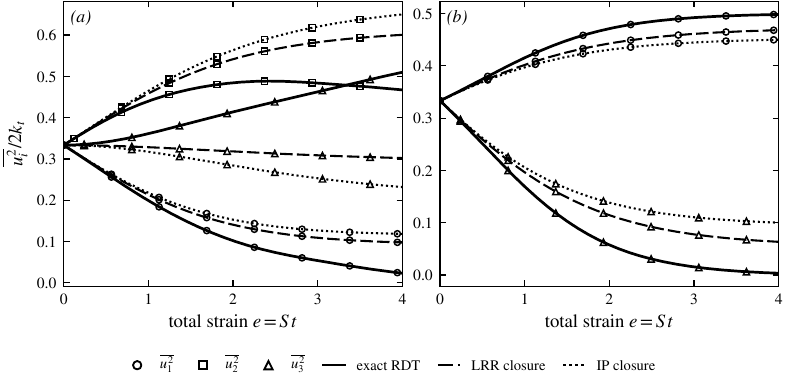}
  \caption{Component energy fractions against total strain for initially
  isotropic turbulence under (\textit{a}) plane strain and (\textit{b})
  axisymmetric strain, normalised to $S=1$: exact rapid-distortion
  theory (wavevector-ensemble integration of \eqref{eq:rdt}) against the
  SC-ODT Reynolds-stress equation~\eqref{eq:rstransport} with the LRR closure \eqref{eq:lrr} and the IP
  closure \eqref{eq:crow}. Line style denotes the source (solid
  exact RDT, dashed LRR, dotted IP) and symbol the component
  ($\circ\,\overline{u_1^2}$, $\square\,\overline{u_2^2}$,
  $\triangle\,\overline{u_3^2}$). For axisymmetric strain the two
  lateral components coincide and only
  $\overline{u_1^2}=\overline{u_2^2}$ (lateral) and $\overline{u_3^2}$
  (axial) are shown. All closures share the exact onset slope; LRR
  tracks the exact solution more closely at finite strain.}
  \label{fig:rdt-components}
\end{figure}

The upwash component, whose spectrum is the input to the leading-edge
response, makes the case for the rapid term quantitative
(figure~\ref{fig:rdt-upwash}). Production acting alone drives the upwash
fraction from $1/3$ to $0.98$ by $e=4$---almost the entire kinetic
energy forced into one component, an essentially non-realisable
state---and overstates the initial amplification rate by the factor
$\tfrac52$. The rapid term removes three fifths of this at onset and
keeps the fraction physical: the exact solution peaks at $0.49$ near
$e\approx2.4$ and then relaxes as energy passes to the spanwise
component. Both closures reproduce the amplification itself---the effect
that the frozen, isotropic input of classical leading-edge-noise
prediction omits---while overpredicting it at large strain (LRR $0.60$,
IP $0.65$ at $e=4$, against the exact $0.47$). The rapid
pressure--strain term is therefore required, and the
accuracy of the method for a given configuration is controlled by the
total strain accumulated along the stagnation streamline.

\begin{figure}
  \centering
  \includegraphics[width=0.62\textwidth]{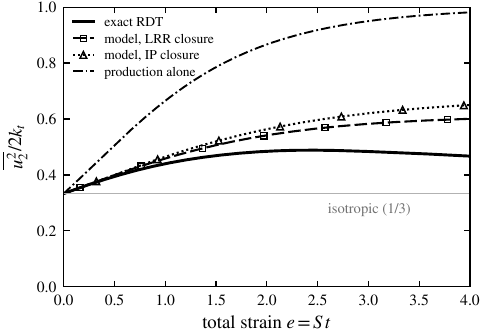}
  \caption{Upwash energy fraction under plane strain for initially
  isotropic turbulence: exact rapid-distortion theory, SC-ODT with
  the LRR and IP closures, and production alone ($B_{ij}=0$). 
  }
  \label{fig:rdt-upwash}
\end{figure}

\subsection{Verification of the implementation}\label{sec:level1a}

The preceding checks verify the \emph{formulation}; the
\emph{implementation}---in which the strain-forcing generator
$\mathcal{A}_{ij}$ is not prescribed but reconstructed at every substep
(line Reynolds stress from the resolved field, closure, Lyapunov solve
\eqref{eq:lyapunov}, pointwise application)---is verified in the supplementary
material (\S\,S3, figure~S3). Run in the rapid-distortion limit from a
statistically isotropic resolved initial field, the solver lies on the
LRR moment trajectory over $0\le e\le4$ to within $10^{-4}$ in each
component fraction and $0.3\%$ in the kinetic-energy amplification, so
the per-substep chain is correct at every state of anisotropy;
agreement with the \emph{exact} solution would indicate a fault, since
with events and diffusion suppressed the line carries no information
beyond the single-point statistics the closure itself uses. The test
requires the resolved spectral initialisation and strain-resolved
timestep bound recorded in Appendix~\ref{app:numerics}. For
definiteness, \emph{every strained computation in the remainder of the
paper uses the LRR closure}, whose error against the exact solution is
quantified in \S\,\ref{sec:finite-strain} at every strain reached; the
sensitivity of the spectral conclusions to this closure choice is not
estimated but measured, through the exact-RDT comparisons of
\S\,\ref{sec:val-spectrum} and the axisymmetry measurement of
\S\,\ref{sec:gate}, which carry the closure deficiency into the quoted
uncertainty rather than leaving it implicit.

\subsection{Moment-level validation against the strained-turbulence DNS
of Lee \& Reynolds}\label{sec:lr-validation}\label{sec:lr-plane}

The DNS of \citet{LeeReynolds1985} is retained as a benchmark for one
specific reason: it was deliberately run at a turbulence Reynolds number
low enough ($q^4/\nu\varepsilon\lesssim100$) to be fully resolvable on a
single line, and it reports the Reynolds-stress anisotropy against total
strain for several irrotational strain modes alongside the analytic
rapid-distortion prediction---a graded, moment-level reference that the
spectral benchmarks of \S\,\ref{sec:val-spectrum} do not provide. The
same programme's later rapid-limit analysis
\citep{LeePiomelliReynolds1986} supplies the analytic reference used in
the CMK comparison; higher-Reynolds-number strained
computations exist, but the spectral role they would serve is filled
here by the dedicated strained-box DNS of \S\,\ref{sec:val-spectrum},
leaving \citeauthor{LeeReynolds1985}'s low-Reynolds-number study as the
appropriate moment-level benchmark for a fully resolved line. It is a
Stanford report rather than a journal paper, but it remains the
irrotational-strain benchmark of that programme's own later work: the
structure-based models of \citet{ReynoldsKassinos1995} and
\citet{KassinosReynoldsRogers2001} are validated against its plane-strain
and axisymmetric runs \citep{KassinosReynolds1998}. More recent direct
simulations of the same deformations at moderate Reynolds number provide
independent, journal-published comparators, and are overlaid on the
comparisons below: \citet{ZusiPerot2013} for plane strain
(figure~\ref{fig:lr-plane-fig}) and \citet{ZusiPerot2014} for
axisymmetric contraction (figure~\ref{fig:lr-axisym-fig}), both digitised
from their published anisotropy figures and shown over the strain range
they reach ($\mathcal{C}\le1.65$, the deformation ratio $\mathcal{C}$ related to the total strain $e=\int S\,\mathrm{d}t$ by $e=2\ln\mathcal{C}$ for plane strain, since $S=(2S_{ij}S_{ij})^{1/2}$ is twice the principal strain rate, and by $e=\sqrt3\ln\mathcal{C}$ for axisymmetric contraction, with $\mathcal{C}$ the axial stretch). A direct simulation of the
CMK straining cycle
\citep{GualtieriMeneveau2010}, which reports the same variances and
pressure--strain budget as that experiment against a
Launder--Reece--Rodi closure, is discussed in the supplementary material
(\S\,S1).
The
comparison is of the anisotropy $b_{ij}$ as a function of the
deformation ratio $\mathcal{C}$, computed from the single-point component
variances along the line, ensemble-averaged over $1000$ realisations;
no spectral transform is involved. Two of their strain modes are run:
plane strain and axisymmetric contraction, the line along a compressed
direction in each case, to the deformation ratios of the DNS ($\mathcal{C}=4$ and
$3.32$). In this subsection and in figures~\ref{fig:lr-plane-fig}--\ref{fig:lr-b11}
we use the axis labels of \citeauthor{LeeReynolds1985}: under plane strain
$x_1$ is their unstrained direction, $x_2$ the compressed direction (the
line) and $x_3$ the stretched direction, which are $x_3$, $x_2$ and $x_1$ of
\eqref{eq:plane-strain}; under axisymmetric contraction their $x_1$ is the
stretched symmetry axis, $x_3$ of \eqref{eq:axi-strain}. The measured line
compression agrees with
$\exp(A_{22}t_f)$, where $t_f$ is the total straining time, to better than $0.3\%$, confirming the kinematic
fidelity of the dilatation. One parameter distinction matters for
interpretation: \citeauthor{LeeReynolds1985} characterise each run by
$S^{*}=Sq^2/\varepsilon$, which equals $2\chi$ identically since
$q^2=2k_t$---and find the axisymmetric anisotropy essentially
independent of $S^{*}$, a parameter-free target, whereas in plane strain
the unstrained component $b_{11}$ is
strongly $S^{*}$-dependent, spanning $b_{11}\approx0$ (slow) to
$+0.12$ (rapid) in the DNS.

Under plane strain (figure~\ref{fig:lr-plane-fig}) the model reproduces
the two strained components: $b_{22}=+0.21$ at $\mathcal{C}=3.86$ against
$\approx+0.20$ in the DNS, and $b_{33}=-0.20$ against $\approx-0.26$, i.e.
within about $25\%$ at the largest strain. These are the components the
DNS identifies as strain-rate-independent, so for them the comparison
does not depend on $S^{*}$. Under axisymmetric contraction the model reproduces the
sign structure and the monotonic approach to a strained state while
under-predicting the anisotropy magnitude by roughly $25\%$
($b_{11}=-0.21$ against $\approx-0.29$ at $\mathcal{C}=3.16$), dominated by the
spectral-closure gap isolated in \S\,\ref{sec:finite-strain} and common
to the whole second-moment family; the line adds only a small artefact,
breaking the transverse symmetry with a residual split
$b_{22}=+0.102$ against $b_{33}=+0.113$ that the symmetry-respecting
mean $\tfrac12(b_{22}+b_{33})$ removes. Figure~\ref{fig:lr-axisym-fig}
shows the comparison, with the independent DNS of \citet{ZusiPerot2014}
tracking both the model and Lee \& Reynolds over the strain range it
reaches.

\begin{figure}
  \centering
  \includegraphics[width=0.8\textwidth]{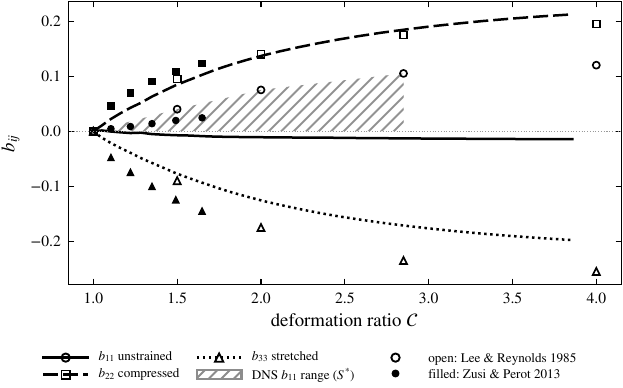}
  \caption{Reynolds-stress anisotropy under plane strain: SC-ODT
  (lines; ensemble mean over $1000$ realisations, standard error within
  the line width) against the DNS of \citet{LeeReynolds1985} (open symbols, digitised
  from their figure~5.4). As noted in \S\,\ref{sec:level1a}, here and
  henceforth all strained computations use the LRR closure. Line style and symbol denote the component
  ($b_{11}$ solid/$\circ$, $b_{22}$ dashed/$\square$, $b_{33}$
  dotted/$\triangle$); open symbols are \citet{LeeReynolds1985}, filled
  symbols the independent DNS of \citet{ZusiPerot2013} (their highest
  rate, over the strained range $\mathcal{C}\le1.65$); the hatched band is the DNS
  $b_{11}$ range over slow-to-fast $S^{*}$. The strained components
  $b_{22}$ and $b_{33}$ are reproduced quantitatively, and the two DNS
  agree where they overlap; the unstrained $b_{11}$ is
  strain-rate-dependent in the DNS, with SC-ODT at its
  slow-strain edge.}
  \label{fig:lr-plane-fig}
\end{figure}

\begin{figure}
  \centering
  \includegraphics[width=0.72\textwidth]{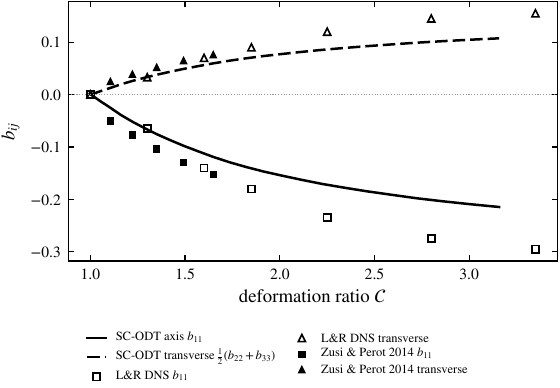}
  \caption{Reynolds-stress anisotropy under axisymmetric contraction:
  SC-ODT (lines; ensemble mean over $1000$ realisations, standard-error
  band) against the DNS of \citet{LeeReynolds1985} (open symbols) and the
  independent DNS of \citet{ZusiPerot2014} (filled symbols, their highest
  rate, over $\mathcal{C}\le1.65$). Solid/$\square$: the stretched axial component
  $b_{11}$; dashed/$\triangle$: the transverse component
  $\tfrac12(b_{22}+b_{33})$. The model tracks the sign and monotonic
  growth of both, and the two DNS agree where they overlap.}
  \label{fig:lr-axisym-fig}
\end{figure}

The one component the model does not reproduce, the plane-strain
unstrained anisotropy $b_{11}$, is not a deficiency specific to the
model, and figure~\ref{fig:lr-b11} locates it exactly. Exact linear
rapid distortion, evaluated by the wavevector-ensemble integration of
\S\,\ref{sec:finite-strain}, gives $b_{11}$ \emph{rising} to $+0.114$
at $\mathcal{C}=4$: although $A_{1k}=0$, the rapid pressure--strain feeds the
neutral direction---the same orientation-carried redistribution as the
spanwise growth in \S\,\ref{sec:finite-strain}. (A kinematic argument
that $\langle u_1^2\rangle$ is conserved because the component is
unstrained holds for production but not for the rapid term, and gives the
wrong sign.) The fast-strain DNS therefore \emph{tracks exact RDT}, as
\citeauthor{LeeReynolds1985} themselves observed, while both
single-point closures give negative $b_{11}$ ($-0.022$ for LRR and
$-0.074$ for IP at $\mathcal{C}=4$): the
closure family misses the neutral-direction feed, which is carried by
the orientation distribution that single-point models discard---the
same deficiency measured spectrally, and localised in the closure's
$b$-dependent term, in \S\,\ref{sec:order}. SC-ODT
follows the closure on which its rapid operator $B_{ij}$ is built
($b_{11}=-0.014$ against $-0.022$ for LRR at the largest strain), so its
near-zero $b_{11}$ is correct for its formulation and matches the DNS
at slow strain rate, where the nonlinear scrambling erases the
neutral-direction feed; at fast strain it inherits the closure
family's known limitation (\S\,\ref{sec:finite-strain}) rather than any line-specific artefact.
The components that govern the upwash spectrum fed to the leading-edge
response are precisely those the model reproduces quantitatively.
Finally, \citeauthor{LeeReynolds1985}'s own transport-budget analysis
finds that in plane strain the intercomponent transfer is carried almost
entirely by the rapid term---independent support for the rapid/slow
decomposition on which \S\,\ref{sec:gate} rests.

\begin{figure}
  \centering
  \includegraphics[width=0.82\textwidth]{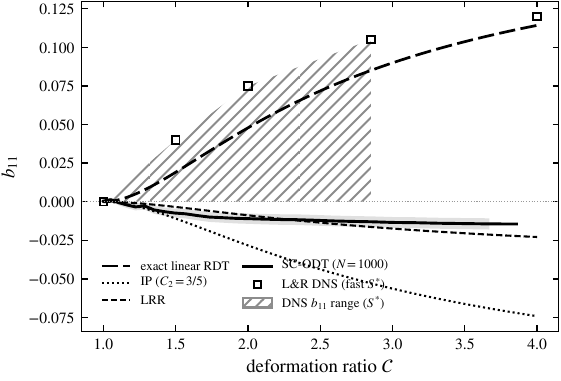}
  \caption{The unstrained-direction anisotropy $b_{11}$ under plane
  strain: SC-ODT (solid, $1000$ realisations, grey standard-error band)
  against exact linear rapid distortion (long-dashed,
  wavevector-ensemble integration), the IP (dotted; coefficient $\tfrac35$
  of \eqref{eq:crow}, marked $C_2=3/5$ in the legend) and LRR
  (short-dashed) closures, and the DNS of \citet{LeeReynolds1985} (open
  squares; hatched band = DNS range over slow-to-fast $S^{*}$). The
  fast-strain DNS tracks exact RDT, whose rapid pressure--strain feeds
  the neutral direction; the single-point closures---and SC-ODT, built
  on them---miss this orientation-carried redistribution, the
  deficiency measured and localised in \S\,\ref{sec:order}.}
  \label{fig:lr-b11}
\end{figure}

\section{The linear reference and the distorted spectrum}\label{sec:val-spectrum}

This section establishes what SC-ODT predicts for the strained spectrum
and measures it against the correct linear reference and against DNS.
It first separates the model's own kinematics---a rigid spectral
translation---from the prediction of linear theory for the
line-projected spectrum, which is \emph{not} rigid
(\S\S\,\ref{sec:spec-rdt}--\ref{sec:spec-linref}); it then examines the
emergent distorted spectrum and fixes the physical operating point
(\S\S\,\ref{sec:spec-bvc}--\ref{sec:spec-op}); and it closes with the
benchmark that addresses the model-versus-itself comparison---a
strained-box DNS with an exact rapid-distortion companion, reduced to
the same one-dimensional observables (\S\,\ref{sec:box-dns}).

\subsection{Compression without cascade: the model's kinematic
reference}\label{sec:spec-rdt}

With the eddy events suppressed and zero molecular viscosity, the mean strain
is the only process acting on the resolved field, and its effect on the line
is purely kinematic. The line length contracts as
$L_y(e)=L_{y,0}\,\mathrm{e}^{A_{22}e}$, so every resolved wavenumber is
rescaled as $\kappa(e)=\kappa(0)\,\mathrm{e}^{-A_{22}e}$---the wavevector kinematics it shares
with rapid-distortion theory \citep{BatchelorProudman1954,HuntCarruthers1990}%
---while the continuous redistribution operator of
\S\,\ref{sec:rapid-kernel} acts uniformly on every resolved wavenumber. The
combination induces a rigid translation of the line spectrum towards
higher wavenumber that transfers no energy between scales. This rigid
translation is a property of the SC-ODT \emph{line}---uniform dilatation
composed with a wavenumber-uniform operator---not of linear distortion
theory itself; the correct linear reference for the projected
spectrum is constructed in \S\,\ref{sec:spec-linref} below. As a check on the
implementation, however, the rigid translation is an exact, parameter-free
target.

In this limit the component spectra translate toward
higher $\kappa$ without change of shape---the peak neither broadens nor develops an
inertial range---and the spectral centroid follows
$\mathrm{e}^{-A_{22}e}$, reaching $7.03$ at $e=3.9$ against the exact
$7.03$, with the three components migrating identically; the domain
length follows $\mathrm{e}^{A_{22}e}$ to four digits. The single-point moments are
unchanged from the no-dilatation case, since a uniform rescaling of position
leaves the length-weighted velocity statistics invariant. Compression
therefore acts entirely on the spectrum and reproduces the model's kinematic
prediction exactly. It is the kinematic reference (the rigid-translation
line in the centroid figures below) against which the full model is
measured.

\subsection{The linear reference for the projected spectrum: exact
RDT}\label{sec:spec-linref}

The rigid translation above must not be mistaken for the prediction of
linear rapid-distortion theory for the quantity the line actually carries.
In the full linear theory the amplitude of each Fourier mode evolves at a
rate that depends on the orientation of its three-dimensional wavevector
\eqref{eq:rdt}; the one-dimensional component spectra are
projections---integrals over the transverse wavenumber plane---and there is
no reason for a projection of orientation-dependent amplification to
translate rigidly. Whether it does is a question of calculation, and the
answer is that it does not. Using the closed-form Cauchy solution of the
linear problem for an isotropic initial spectrum, integrated analytically
over the transverse wavenumber plane and band-limited to the wavenumbers
resolved by the DNS box of \S\,\ref{sec:box-dns}, we define the shape ratio
$R_{nn}(\kappa_2)$ as the projected component spectrum at strain $e$,
normalised by its integral, divided by the rigidly translated initial
spectrum normalised likewise; $R_{nn}\equiv1$ at every $\kappa_2$ if and
only if the translation is rigid. At $e=1$ under plane strain the upwash
ratio $R_{22}$ falls monotonically from $0.86$ at half the initial spectral
centroid to $0.36$ at ten times it: exact linear theory redistributes the
projected upwash spectrum towards low wavenumber by a factor of $2.4$
across the resolved range, because modes aligned with the line
($\kappa\to\kappa_2$) are amplified least. The transverse components are
affected more weakly ($R_{11}$ from $1.07$ to $0.81$; $R_{33}$ within seven
per cent of unity), so the projected anisotropy of linear theory is itself
scale-dependent. Figure~\ref{fig:rdt-projection} shows the ratios.
Panel~\textit{b} adds the DNS at $\chi=0.8$: at the largest scales its
transverse splitting is of the size of the linear one, while towards the
small scales the nonlinear dynamics return it to zero, the scale
dependence that \S\,\ref{sec:validated} measures for the upwash. The
exact-RDT companions of the DNS (\S\,\ref{sec:box-dns}), which share its
initial fields, reproduce these analytic curves to within $0.05$ in the
transverse splitting over the resolved band; the analytic integral is
used because the companions sample the largest scales with one Fourier
line per band.

\begin{figure}
  \centering
  \includegraphics[width=\textwidth]{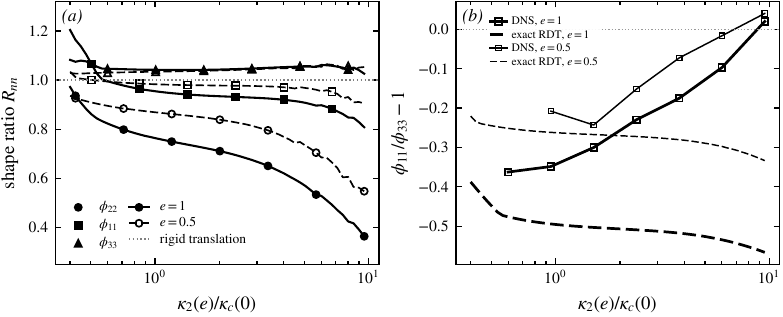}
  \caption{Exact linear theory does not translate the projected spectra
  rigidly. (\textit{a}) Shape ratio $R_{nn}(\kappa_2)$ of the exactly
  distorted, line-projected component spectra to a rigid translation under
  plane strain, at $e=1$ (solid, filled symbols) and $e=0.5$ (dashed, open
  symbols); symbols distinguish the components, and the dotted line is the
  rigid translation $R_{nn}=1$. Closed-form Cauchy solution for an
  isotropic initial spectrum, integrated over the transverse wavenumber
  plane. (\textit{b}) Transverse splitting $\phi_{11}/\phi_{33}-1$ versus
  wavenumber in exact linear theory (dashed) and in the strained-box DNS of
  \S\,\ref{sec:box-dns} at $\chi=0.8$ (solid, squares), at $e=1$ (heavy) and $e=0.5$
  (light).}
  \label{fig:rdt-projection}
\end{figure}

Two consequences run through the remainder of the paper. First, the
comparison of the full model against the rigid translation
(\S\,\ref{sec:spec-bvc}) measures the departure of the full model from its
kinematical part; the departure of the \emph{flow} from linear theory must be
measured against the exact projected reference, and the two differ already
at the linear level. Second, the model's rigid spectral kinematics is
itself a structural approximation---linear theory deepens the projected
anisotropy with wavenumber, which a wavenumber-uniform operator cannot
represent---and \S\,\ref{sec:validated} quantifies the resulting deficiency
against both the exact reference and DNS.


\subsection{The distorted spectrum: strain on against strain
off}\label{sec:spec-bvc}

The central model parameter sensitivity is the strain dependence of
the spectrum at fixed Reynolds number, with the turbulence otherwise
identical. Two ensembles, of $255$ (SC-ODT) and $128$ (baseline)
realisations, are run at $\nu=10^{-5}$:
SC-ODT (eddies and dilatation active) and a baseline identical
in every respect except that the strain machinery is switched off, so
the line does not compress---which is standard ODT, the $A_{ij}=0$
limit of \S\,\ref{sec:summary}. Because the two differ only in the presence
of the mean strain, the difference between them is the distortion,
cleanly separated from the cascade and dissipation the eddies produce on
their own.

Figure~\ref{fig:spec-bvc} shows the result, and it departs from the
kinematic reference in a specific and physically meaningful way.
Consider the centroid first. The no-strain baseline, after the brief
initialisation transient, relaxes monotonically to
$\bar\kappa/\bar\kappa_0<1$: with no compression, the eddies merely
redistribute energy and the centroid drifts below its initial value.
SC-ODT is lifted clearly above this baseline---compression does push
the spectrum towards higher wavenumber---yet it settles at
$\bar\kappa/\bar\kappa_0\approx1.6$--$2.1$, far below the
rigid-translation line, which rises to $7.4$ at $e=4$. The reason is
visible in the spectra: the SC-ODT spectrum is not a rigid
translation of the baseline but a \emph{broadband} distortion. It
extends to substantially higher wavenumber than the baseline---energy is
genuinely carried to small scales by compression and cascade---but the
bulk of the energy remains at the energetic scales, and what reaches
high wavenumber is dissipated there rather than accumulated.

Two distinct effects are mixed in this departure, and
\S\,\ref{sec:spec-linref} requires that they be separated. The first is
the eddy cascade and dissipation, which the strain-on/strain-off
construction isolates cleanly at the level of the model. The second is
that exact linear theory itself, once projected onto the line, is
already broadband (figure~\ref{fig:rdt-projection}): part of what a
rigid-translation reference attributes to nonlinearity is present in the
linear physics. The strain-on/strain-off comparison therefore measures
what the mean strain does \emph{within the model}; how the model's
distorted spectrum stands against exact linear theory and against DNS on
the same footing is answered by the three-way comparison of
\S\,\ref{sec:validated}. What survives either way is the qualitative
point that a frozen, rigidly strained upwash spectrum misrepresents the
distorted inflow, which motivates supplying a distorted spectrum to
Amiet's theory in place of it, in the spirit of the
rapid-distortion-modified inflow spectra of \citet{deSantana2016} but
without the frozen-shape assumption.

\begin{figure}
  \centering
  \includegraphics[width=\textwidth]{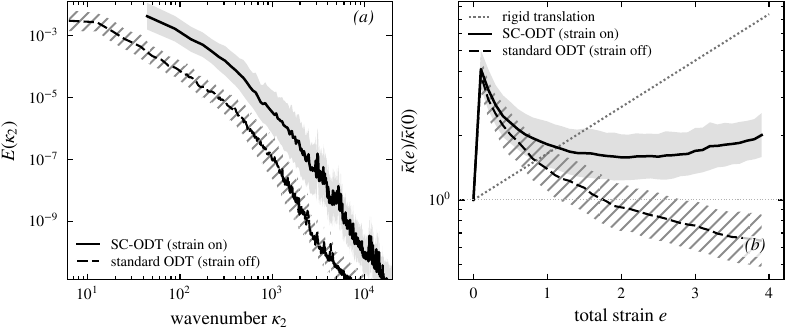}
  \caption{SC-ODT (strain on) against the matched strain-off baseline,
  standard ODT, at
  $\nu=10^{-5}$ (ensemble means over $255$ and $128$ realisations
  respectively). The shaded bands span $\pm1$ geometric standard deviation
  across the realisations: grey fill for SC-ODT, hatched for the baseline.
  (\textit{a}) Final-strain total spectrum
  $E(\kappa_2)$: SC-ODT (strain on, solid) extends to higher wavenumber
  than the baseline (strain off, dashed) but remains broadband, not a
  rigid translation. (\textit{b}) Centroid migration: the baseline relaxes
  below unity, SC-ODT is lifted above it by the compression, yet
  stays far below the rigid-translation line (grey dotted). The gap
  between the strain-on and strain-off curves is the distortion
  attributable to the mean strain at this Reynolds number.}
  \label{fig:spec-bvc}
\end{figure}

\subsection{The physical operating point}\label{sec:spec-op}

The molecular viscosity is, in these units, the inverse integral-scale
Reynolds number, and it is a physical property of the represented
inflow, not a tunable parameter: a sweep over
$\nu=10^{-4}$--$10^{-5}$ shows the cascade deepening monotonically, with
the centroid settling below the rigid-translation centroid line. The operating
point is the viscosity matching the target inflow, resolution-checked. For the
reference configuration---a NACA0012 at $U_\infty=20$\,m\,s$^{-1}$ in
grid turbulence with intensity $u'/U_\infty\approx6\%$ and integral
scale $\mathcal{L}\approx55$\,mm at the leading-edge plane---the
integral-scale Reynolds number is
$Re_{\mathcal{L}}\approx4.4\times10^3$, which maps through the
initial-condition normalisation to $\nu\approx3\times10^{-6}$, run on a
refined grid ($16000$ cells at initialisation, adapted thereafter;
table~\ref{tab:resolution}) on which the dissipation rolloff is
resolved with margin; the ensemble comprises $254$ realisations.

Figure~\ref{fig:spec-op} shows the result. The component spectra fall
smoothly across four decades, the upwash
component $E_2$ amplified above the streamwise $E_1$---the spectral
signature of the plane strain---and no extended $\kappa^{-5/3}$ range, as
expected physically at this Reynolds number
\citep{Kolmogorov1991,Pope2000}. The centroid settles after the
initialisation transient and migrates to $\approx2.4$ by $e=4$, far
below the rigid-translation value of $7.4$: the broadband, dissipative
distortion of \S\,\ref{sec:spec-bvc} at the physical operating point. The regime is
quantified with an estimator-free dissipation: because the eddy events
conserve energy and the production is exact, molecular dissipation is
the only sink, so it is fixed by the budget \eqref{eq:tke-budget},
rearranged as
\begin{equation}
  \varepsilon(t) = \mathcal{P}(t) - \frac{\mathrm{d}k_t}{\mathrm{d}t}
              = -A_{ij}R_{ij}(t) - \frac{\mathrm{d}k_t}{\mathrm{d}t},
  \label{eq:eps-budget}
\end{equation}
every term available on the line. By this budget the strain-rate
parameter settles at $\chi\approx1.2$ over the strained
trajectory (ensemble of $254$ realisations), corroborated within a few
per cent by the direct gradient estimate
$\langle2\nu(\partial u_i/\partial y)^2\rangle\approx1.15$: the
operating point sits at $\chi=O(1)$, bracketed by the two strained-box
DNS rapidities $\chi=0.8$ and $16$ used in
\S\,\ref{sec:validated}, the regime the leading-edge
stagnation flow occupies, where neither the rapid-distortion nor the
near-equilibrium limit applies and where improvement over linear theory is
meaningful. The spectral \emph{peak} corroborates the regime reading
independently of the dissipation estimator: the energy-containing
wavenumber itself migrates far less than the rigid law, so the eddy
relaxation holds back the energy-containing range, not only the
dissipation tail. The operating-point ensemble requires a few hundred
core-hours, which permits sweeps over inflow conditions.

\begin{figure}
  \centering
  \includegraphics[width=\textwidth]{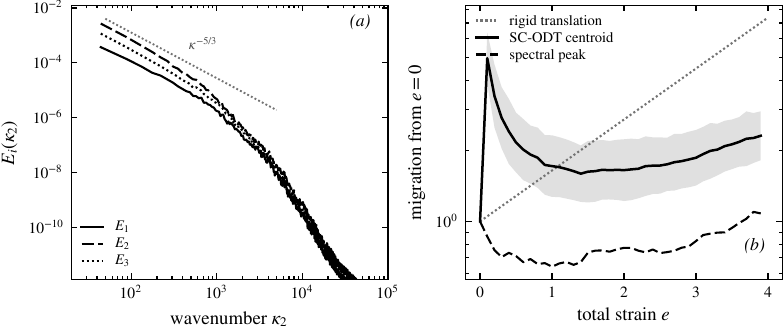}
  \caption{The physical operating point ($\nu\approx3\times10^{-6}$,
  $Re_{\mathcal{L}}\approx4.4\times10^3$, $16000$ initial cells; ensemble over
  $254$ realisations, geometric $\pm\sigma$).
  (\textit{a}) Final-strain component spectra (solid $E_1$, dashed
  $E_2$, dotted $E_3$): smooth rolloff; the
  upwash $E_2$ amplified above the streamwise $E_1$; grey dotted
  reference slope $\kappa^{-5/3}$.
  (\textit{b}) Migration of the total-spectrum centroid (solid) and the
  spectral peak (dashed) against the rigid-translation law (grey dotted), with the geometric $\pm\sigma$ ensemble spread shaded:
  the broadband distortion relative to the rigid-translation reference, at
  the operating point and with the noise averaged out.}
  \label{fig:spec-op}
\end{figure}

\subsection{Validation against the CMK strained-turbulence experiment}\label{sec:cmk}

The strain-on/strain-off comparison is model-internal; an external
benchmark at laboratory Reynolds number is provided by the CMK plane-strain
experiment, whose straining phase
(deformation ratio $\mathcal{C}\approx3$, $e\approx2.20$ in the present notation)
we reproduce. In the rapid-distortion limit (eddy events suppressed,
inviscid) the model follows its kinematic reference---component
centroids on the rigid-translation law, $\mathcal{C}=3.00$ at
$e=2.20$---and the moment closure: the kinetic-energy amplification is
$1.72$ against the exact $1.69$ of the analytic reference of
\citet{LeePiomelliReynolds1986} that CMK use, the residual attributed to
the moment closure, not the implementation (\S\,\ref{sec:level1a}). With
eddies and viscosity active the component centroids diverge under
strain---streamwise migration $1.9$ against $1.5$ for the upwash at the
end of the straining phase, both far below the rigid value $3$, the
upwash held back most. This corresponds qualitatively to the departure
from rapid distortion that CMK report, measured here against the model's
rigid-translation kinematics rather than against exact projected RDT,
which is itself not rigid (\S\,\ref{sec:spec-linref}); it is the spectral
signature of the broadband distortion of \S\,\ref{sec:spec-bvc}. The
comparison is qualitative by construction (initial spectrum, spectral
projection and scale separation all differ from the experiment's); the
full account, with both figures and the stated limits, is given in the
supplementary material (\S\,S1), and the quantitative spectral
assessment is carried out against the strained-box DNS below, where
projection and initial spectrum are matched by construction.

\subsection{Validation against a strained-box DNS with an exact-RDT
companion}\label{sec:box-dns}

The comparisons above are either model-internal
(\S\,\ref{sec:spec-bvc}) or, where external, qualitative by construction
(\S\,\ref{sec:cmk}) or at the Reynolds-stress level
(\S\,\ref{sec:lr-validation}). The referent quantity for the
leading-edge application is spectral, and the objection that the
emergent-spectrum results compare the model only with itself must be met
with an independent, projection-matched spectral benchmark under the
same strain. To that end
we performed direct numerical simulations of plane-strained homogeneous
turbulence in a deforming-frame (Rogallo-type) pseudo-spectral formulation:
a $128^3$ box, decaying isotropic precursor, then plane strain
$A=S\,\mathrm{diag}(\tfrac12,-\tfrac12,0)$ applied at strain-rate
parameters $\chi\in\{0.8,16\}$ at onset, carried to total
strain $e=1$, four independent realisations
($\kappa_{\max}\eta\approx0.76$ at $e=1$: a pilot resolution, so
small-scale statements below are qualitative). Alongside every DNS
realisation runs an \emph{exact-RDT companion}: the same initial field
evolved under the closed-form Cauchy solution of the linearised equations
and projected onto lines identically to the DNS and to ODT. The companion
reproduces the analytic linear reference of \S\,\ref{sec:spec-linref} to
within $0.05$ in the transverse splitting; here it completes a three-way comparison in which model, linear theory and
DNS are reduced to the same one-dimensional observables. On the ODT side
the ensembles use an unstrained precursor: the line is evolved as ordinary
ODT until the initialisation transient has relaxed and the eddy population
is statistically active, and the strain is switched on only then, so that
the strained evolution starts from a physically conditioned state rather
than from the synthetic initial condition. All ODT statistics in this
subsection are medians over $1024$ realisations per case with bootstrap
confidence intervals; median statistics are used because the strained band
energies are strongly intermittent across realisations, and ensemble means
of spectral ratios are dominated by rare events at any practical sample
size.

\subsubsection{Distance from linear theory as a function of strain
rapidity}\label{sec:envelope}

A scalar measure complementary to the band-resolved anisotropy of
\S\,\ref{sec:validated} tracks the \emph{whole} spectral shape across
strain and strain rapidity: at each
total strain the median ODT component line spectra are fitted, jointly in
the upwash and transverse components, by the exactly distorted spectrum of
linear theory---the closed-form Cauchy distortion of an isotropic von
K\'arm\'an field, projected onto the line, with only the pre-distortion
amplitude and energy wavenumber free. The root-mean-square log-residual of
that two-parameter fit is a scale-resolved distance between the model's
strained spectra and exact linear theory. Figure~\ref{fig:rdt-distance}
shows it for precursor-conditioned ensembles at $\chi\approx0.4$
and $\approx8$ at onset ($1024$ realisations each, identical precursor
state: the two fits coincide at $e=0$ at $13$ per cent, the residual
shape mismatch of a relaxed ODT line against the von K\'arm\'an form).

The distance grows with strain in both cases, but \emph{faster at rapid
strain}: to approximately $22$ per cent at $e=2$ for
$\chi\approx0.4$ against approximately $24$ per cent,
saturating early, at $\chi\approx8$, with the separation
established already by $e=0.5$. If the model's rapid response were
spectrally faithful, the rapid curve would be the one to approach zero:
eddy events are subdominant there by construction. That the deviation is
largest precisely where the eddies matter least localises the deficiency in
the rapid kinematics themselves---the wavenumber-uniform operator and the
rigid dilatation---and not in the eddy or relaxation statistics. This
measured validity envelope
is carried with the model into applications: predictions inherit a
spectral-shape uncertainty of order the plotted distance at the operating
point's strain and rapidity.

\begin{figure}
  \centering
  \includegraphics[width=0.62\textwidth]{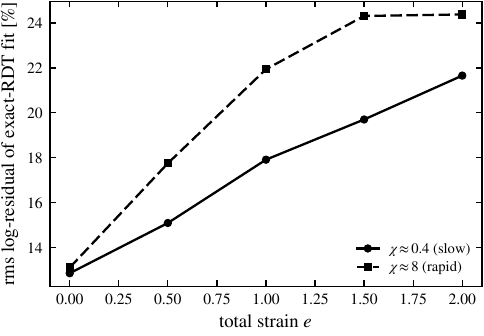}
  \caption{Scale-resolved distance between the strained SC-ODT line spectra
  and exact linear theory: rms log-residual of the two-parameter
  exact-RDT-distorted von K\'arm\'an fit, against total strain, for low
  and high onset strain-rate parameters (medians over $1024$
  realisations; identical precursor state at $e=0$).}
  \label{fig:rdt-distance}
\end{figure}

\subsection{The two eddy types against the benchmark}\label{sec:validated}

The on-line production is carried by the eddy events of
\S\,\ref{sec:event-strain}, which deposit it at the scale of the overturn
rather than uniformly across the spectrum. Measured against the
strained-box benchmark, the model returns the wavenumber-resolved
anisotropy of table~\ref{tab:validated} and figure~\ref{fig:eddytypes}.
Exact linear theory concentrates the strain-induced upwash anisotropy at
the energy-containing scales and \emph{removes} it from the small scales,
because modes aligned with the line are amplified least; the DNS shows the
same shape, weakened at the large scales by the nonlinear return---that
nonlinear effects absent from linear theory shape the strained anisotropy
across scales is likewise the finding of the $4096^3$
axisymmetric-contraction simulations of \citet{ClayYeung2016}.

The two eddy types of \S\,\ref{sec:eddy-types} place the model on this
benchmark. With the one-eddy-type kernel the strained anisotropy is
negative in every band but shallow, about half the DNS magnitude at the
fine scales (figure~\ref{fig:eddytypes}, `one eddy type'). Sampling the
planar type with probability $p=0.5$ and exchange fraction
$\alpha_\mathrm{p}=0.5$---the SC-ODT model used
throughout---roughly doubles the fine-scale deficit and brings it onto the
DNS, with the correct sign of the high-wavenumber upwash anisotropy and its
wavenumber dependence following the benchmark rather than lying flat across
scales. The planar type withholds the out-of-plane sink for the amplified
upwash, which is what deepens the fine-scale deficit; $p$ and
$\alpha_\mathrm{p}$ are set once by this comparison and carried unchanged
through the rest of the paper.

\begin{table}
  \centering\small
  \begin{tabular}{lccccccc}
  $\Delta b_{22}(\kappa_2)$, band $\kappa_2/\kappa_c$: & 0.5--0.75 & 0.75--1.2 & 1.2--1.9 & 1.9--3 & 3--4.8 & 4.8--7.6 & 7.6--12\\[2pt]
  DNS ($128^3$)   & $+0.032$ & $-0.033$ & $-0.029$ & $-0.032$ & $-0.036$ & $-0.040$ & $-0.046$\\
  one eddy type   & $-0.011$ & $-0.014$ & $-0.009$ & $-0.005$ & $-0.016$ & $-0.017$ & $-0.020$\\
  SC-ODT          & $-0.010$ & $-0.019$ & $-0.021$ & $-0.033$ & $-0.033$ & $-0.027$ & $-0.045$\\
  \end{tabular}
  \caption{Strain-induced upwash anisotropy per wavenumber band at $e=1$,
  $\chi=0.8$, for the one-eddy-type kernel and for SC-ODT (the
  two-eddy-type model) against the $128^3$ DNS. Each system is referenced to
  its own unstrained state; $1024$ realisations for the model rows, four for
  the DNS.}
  \label{tab:validated}
\end{table}

At the rapid rate $\chi=16$ the DNS carries the fine-scale
upwash \emph{below} the transverse mean while the energy-containing scales
stay strongly positive (figure~\ref{fig:eddytypes}\textit{b}). The model
does not reproduce that sign reversal: at this rate fewer than one eddy
event acts per line over $e\le1$ (\S\,\ref{sec:doublecount}), so the
scale-wise redistribution that produces the negative fine-scale anisotropy
at $\chi=0.8$ has no time to act, and the continuous rapid operator acts
uniformly across wavenumber.
That regime is the rapid-distortion limit, outside the $\chi=O(1)$
operating point of the leading-edge application (\S\,\ref{sec:lenoise});
capturing it would require a rapid operator that discriminates by
wavenumber, using the angular and anisotropy information of
\eqref{eq:Piclosed}--\eqref{eq:g3} that the moment closure discards
(\S\,\ref{sec:lrr-reduction}).

\begin{figure}
  \centering
  \includegraphics[width=\textwidth]{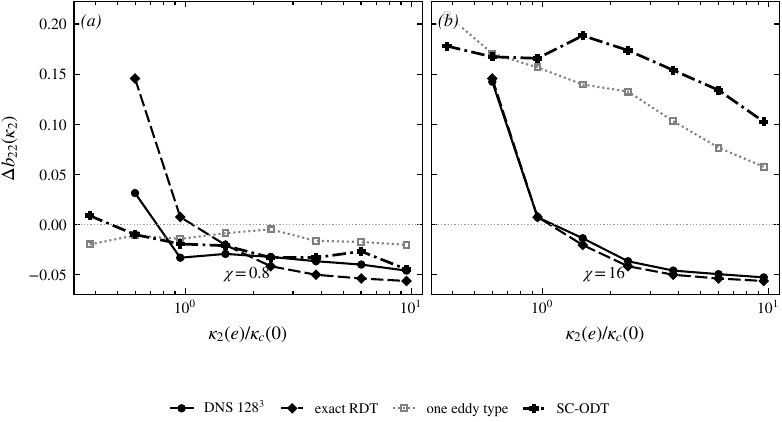}
  \caption{Strain-induced upwash anisotropy $\Delta b_{22}(\kappa_2)$ at
  $e=1$: exact projected RDT (dashed, diamonds) and the $128^3$ DNS (solid,
  circles) against the one-eddy-type kernel (dotted, open squares) and
  SC-ODT, the two-eddy-type model (dash-dot, filled plus), at
  $\chi=0.8$ (\textit{a}) and $16$ (\textit{b}). The two eddy
  types deepen the fine-scale deficit of the one-type kernel toward the DNS
  at moderate strain; at the rapid rate the close agreement of the DNS with
  RDT is reproduced by neither model variant.
  ODT rows: $1024$ realisations, referenced to the unstrained state.}
  \label{fig:eddytypes}
\end{figure}

\section{Consequence for leading-edge noise}\label{sec:lenoise}

The application that motivates the model can now be served with its
output. In the leading-edge-noise theory of \citet{Amiet1975} the
far-field pressure spectral density of a large-aspect-ratio aerofoil
factorises as a product of the flow and observer geometry, the aerofoil
response $|\mathcal L|^2(K_x)$, the upwash wavenumber spectrum
$\Phi_{ww}(K_x,0)$ and the spanwise correlation length
$\ell_y(K_x)=\pi\,\Phi_{ww}(K_x,0)/\!\int\Phi_{ww}(K_x,K_y)\,
\mathrm{d}K_y$, with $K_x$ the streamwise and $K_y$ the spanwise
wavenumber. For a \emph{change of inflow spectrum} at fixed aerofoil,
flow and observer, everything except the turbulence kernel
\begin{equation}
  \mathcal K(K_x)
  \;=\;
  \pi\,\frac{\Phi_{ww}(K_x,0)^2}{\int\Phi_{ww}(K_x,K_y)\,\mathrm{d}K_y}
  \label{eq:knoise}
\end{equation}
cancels exactly, so the change of radiated sound pressure level (SPL) is
$\Delta\mathrm{SPL}(K_x)=10\log_{10}
[\mathcal K^{\rm distorted}/\mathcal K^{\rm frozen}]$, independent of
the aerofoil geometry and response. This is the most direct question the
distortion model can be asked without importing an acoustic solver:
by how much does the prediction change when the inflow spectrum is the
\emph{computed, distorted} one instead of the frozen upstream spectrum
that standard practice prescribes?

The distorted input is taken from the model's own strained ensembles in
the representation already fitted in \S\,\ref{sec:envelope}: the
exact-RDT-distorted von K\'arm\'an family, two parameters
$(A_0,\kappa_e)$ per strain, fitted to the median line spectra of the SC-ODT
$1024$-realisation ensembles at $\chi\approx0.4$ and
$\approx8$. The
line data pin the resolved-wavenumber structure; the
transverse structure that \eqref{eq:knoise} requires is supplied by the
representation---this is \S\,\ref{sec:gate}'s statement about line
knowledge, operating in the application---and the representation's rms
residual, $13$--$24\%$ over the strain range (the distance envelope of
figure~\ref{fig:rdt-distance}), is the corresponding uncertainty
carried by the result. The two-dimensional spectrum follows by direct
integration of the distorted components,
$\Phi_{ww}(K_x,K_y;e)=A_0\!\int\!\Phi^{\rm RDT}_{22}
(K_x,\kappa_2,K_y;e)\,\mathrm{d}\kappa_2$, with the machinery verified
against the wavevector-ensemble integrator of
\S\,\ref{sec:finite-strain} (upwash-variance growth reproduced to
under $1\%$). The frozen baseline is the same inflow's relaxed $e=0$
state carried undistorted. The level and shape contributions are
separated: the kernel scales linearly with the upwash variance, so
dividing it out isolates the spectral reallocation, and the variance ratio itself is taken from the measured line spectra
rather than from the fit. The chain is laid out step by step in the
supplementary material (\S\,S2).

\begin{figure}
  \centering
  \includegraphics[width=\textwidth]{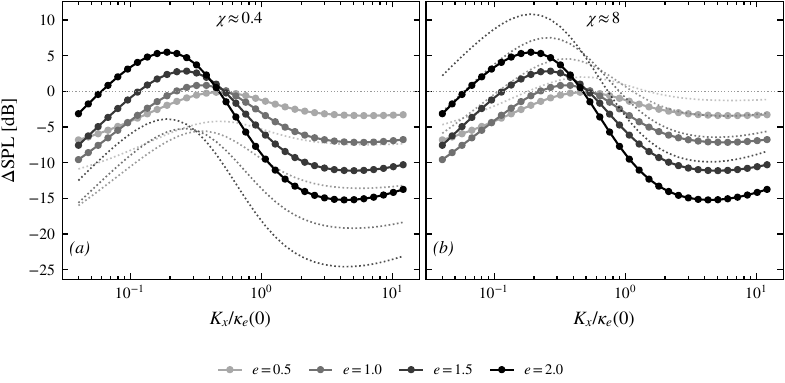}
  \caption{Change of the Amiet leading-edge-noise kernel when the
  frozen inflow spectrum is replaced by the computed SC-ODT spectrum,
  against reduced streamwise wavenumber, at total strains
  $e=0.5$--$2$ (shade light to dark with increasing strain):
  (\textit{a}) $\chi\approx0.4$,
  (\textit{b}) $\approx8$. Solid: spectral reallocation at matched
  upwash variance (shape only); dotted: total, including the measured
  variance change ($-4.1$ to $-9.4$\,dB at the slow rate, where decay
  dominates the transit time; $+2.1$ to $+5.3$\,dB at the rapid rate, where
  the upwash amplification survives). Geometry, aerofoil response and
  observer position cancel identically in this ratio.}
  \label{fig:lenoise}
\end{figure}

Figure~\ref{fig:lenoise} shows the result, and its sign and structure
are the finding. Each panel plots, against reduced streamwise
wavenumber $K_x/\kappa_e(0)$, the change in the noise kernel when the
frozen inflow is replaced by the distorted one: the $\Delta\mathrm{SPL}=0$
line is the frozen reference, so a curve below it marks a weaker source,
and the curves darken with increasing total
strain. The solid traces are the shape change at fixed upwash variance
and the dotted traces add the measured variance term, so the offset
between a solid curve and its dotted partner is that variance
contribution---deepening the drop at the slow rate (panel~\textit{a})
and lifting it at the rapid rate (panel~\textit{b}). Because the
representation distorts a fixed von K\'arm\'an shape by exact RDT, the
shape change at matched variance depends on $e$ alone and is the same in
both panels; the rapidity enters through the measured variance. Including
the measured variance change (dotted), at the slow rate (panel~\textit{a})
the distorted spectrum is a \emph{weaker} leading-edge source than the
frozen one at every resolved reduced frequency: $\approx-5$\,dB at
$K_x\approx\kappa_e(0)$ for $e=0.5$, deepening to beyond $-20$\,dB
by $e=2$, the reduction shallowest near $K_x\approx0.3\,\kappa_e(0)$ and
deepest at high $K_x$. At matched upwash variance (solid) the shape
change alone is a reallocation: it weakens the high-$K_x$ source
($\approx3$--$15$\,dB from $e=0.5$ to $2$ for $K_x\gtrsim\kappa_e(0)$)
while strengthening the low-to-mid band (by up to $\approx5$\,dB at
$e=2$, near $K_x\approx0.2\,\kappa_e(0)$), so the net reduction at the slow rate is
carried by the variance collapse, reinforced at high frequency by the
shape. The mechanism is
the orientation physics of \S\,\ref{sec:gate} acting on the noise
kernel: \eqref{eq:knoise} weights the spanwise-coherent upwash content
$\Phi_{ww}(K_x,0)$, which lives on wavevectors aligned with the
line---exactly the orientations whose upwash is smallest by
solenoidality and whose pre-images retreat into the vanishing
low-wavenumber tail of the initial spectrum as the distortion
proceeds. The total change adds the variance term (dotted):
at the rapid rate the surviving upwash amplification lifts the
low-to-mid band above the frozen reference, leaving a net reduction only
at high $K_x$; at the slow rate viscous decay over the longer transit
time deepens the reduction across the band.

\begin{figure}
  \centering
  \includegraphics[width=0.7\textwidth]{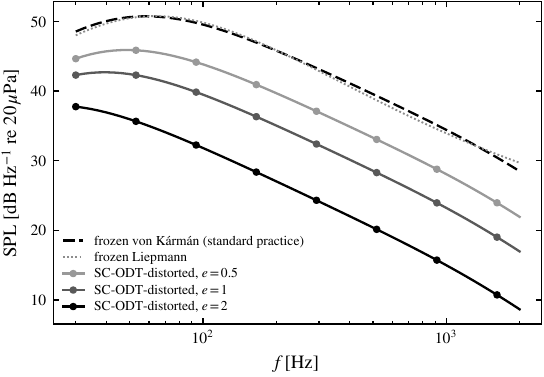}
  \caption{The assembled absolute prediction at the flow conditions of the
  reference configuration of \S\,\ref{sec:spec-op} ($U_\infty$,
  $u'/U_\infty$, $\mathcal{L}$): far-field sound pressure
  level from the compact lift-dipole form of \citet{Amiet1975} with
  the exact Sears response, for a flat plate of chord $0.39$\,m
  ($Re_c\approx5\times10^5$) and span $0.8$\,m at
  $U_\infty=20$\,m\,s$^{-1}$, $u'/U_\infty=6\%$,
  $\mathcal{L}=55$\,mm, observer $1$\,m overhead. Dashed: the frozen
  von K\'arm\'an inflow of standard practice (the $e=0$ member of the
  fitted representation, which is exactly the isotropic von K\'arm\'an
  spectrum); grey dotted: the frozen Liepmann spectrum at the same
  $(u',\mathcal{L})$, an independent analytic cross-check of the
  physical-unit bridge. Solid (grey light to black with increasing
  strain, circles): the same prediction with the SC-ODT-distorted inflow of
  the slow family of figure~\ref{fig:lenoise} ($\chi\approx0.4$,
  table~\ref{tab:resolution}) at total strains $e=0.5$, $1$ and $2$, level
  and shape both included.}
  \label{fig:amiet-abs}
\end{figure}

Figure~\ref{fig:amiet-abs} places the kernel result in the quantity
practice actually reports: the absolute far-field spectrum at the flow
conditions of the reference configuration of \S\,\ref{sec:spec-op}, with
the inflow of the slow family ($\chi\approx0.4$), assembled from the
compact lift-dipole form of Amiet's theory with the exact Sears
response~ \citep{Sears1941}. The spectra enter in physical units through two matching
conditions only---the integral scale fixes the length unit through the
von K\'arm\'an relation $\kappa_e=0.7468/\mathcal{L}$ and the upwash
variance at $e=0$ fixes the amplitude---so the frozen baseline
\emph{is} the standard Amiet--von K\'arm\'an prediction, not a
recalibration of it: the $e=0$ member of the representation is the
isotropic von K\'arm\'an spectrum itself, and it reproduces the
analytic one-dimensional form to within the quoted fit residual
(ratio $0.79$--$1.14$ across $100$\,Hz--$1$\,kHz), with the Liepmann
spectrum shown as an independent check of the unit bridge. The figure
is read vertically, at fixed frequency: the two upper curves---dashed
and dotted, all but coincident---are the frozen baselines, and their
overlap is that unit-bridge check made visible; the three solid curves
below are the same chain fed the distorted inflow at $e=0.5$, $1$ and
$2$, so the downward offset of each from the dashed baseline is the
level reduction the closure predicts at that strain, deepening as the
strain increases. Against
this baseline the computed distortion lowers the predicted level by
$6$\,dB at $e=0.5$, $11$\,dB at $e=1$ and $19$\,dB at $e=2$ at $1$\,kHz
($5$, $8$ and $15$\,dB at $50$\,Hz), within the resolved band
($f\lesssim1.8$\,kHz at these flow conditions), the
variance decay and spectral reallocation of figure~\ref{fig:lenoise} now carried
through geometry, aerofoil response and radiation to the observer.
Both this figure and the kernel ratio of figure~\ref{fig:lenoise}
carry the SC-ODT inflow as the computed spectrum.

\begin{figure}
  \centering
  \includegraphics[width=\textwidth]{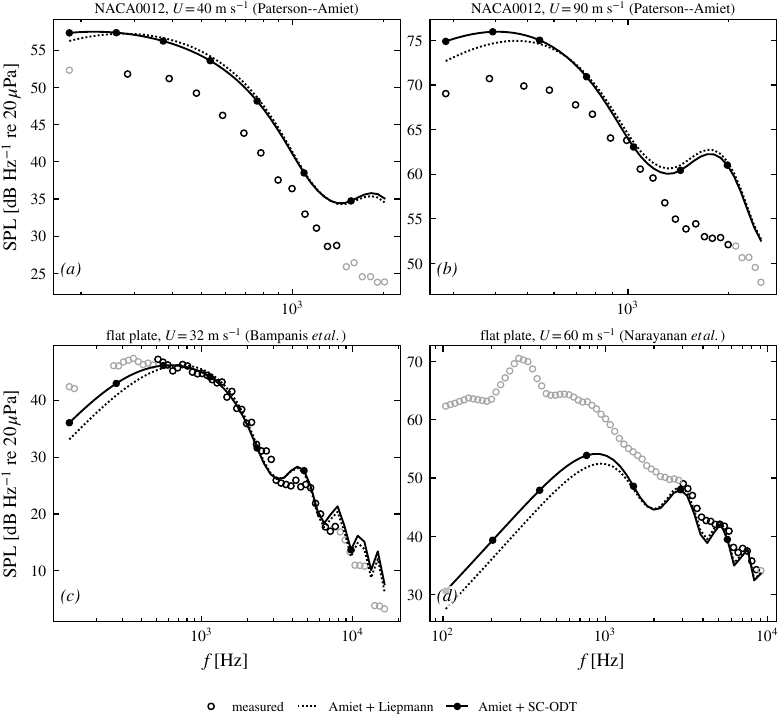}
  \caption{Absolute far-field predictions against published
  measurements, the inflow spectra through the identical
  non-compact response \citep{Amiet1975,RogerMoreau2010}: frozen
  Liepmann (dotted) and the computed SC-ODT spectrum (solid, circles),
  each entered at the case's effective free-distortion strain; the frozen
  von K\'arm\'an prediction is not shown, as it overlaps the SC-ODT
  prediction almost exactly (band-mean bias within $0.4$\,dB). (\textit{a,b}) NACA0012 of \citet{PatersonAmiet1976}
  ($c=0.23$\,m, $u'_w/U\approx4$--$5\%$, $\Lambda_f\approx30$\,mm,
  observer $2.25$\,m overhead; Gershfeld thickness admittance applied
  identically to all predictions); (\textit{c}) ECL flat plate of
  \citet{Bampanis2022} ($c=0.10$\,m, $3$\,mm thick, $Tu=4.5\%$,
  $\Lambda=9$\,mm, $\Theta=90^\circ$, $r=1.25$\,m); (\textit{d}) ISVR
  flat plate of \citet{Narayanan2015} ($c=0.15$\,m, $Tu=2.5\%$,
  $\Lambda=6$\,mm, $r=1.2$\,m). Black measured symbols lie in each
  case's stated comparison band (grey outside: background flags,
  facility contamination, thickness-breakdown frequency); spectra
  digitised from the published figures with per-figure calibration
  checks ($\pm0.5$\,dB). The effective strain is a geometric estimate (potential-flow stagnation-line accumulation truncated at the eddy scale, varied over $\Lambda/4$--$\Lambda$); the resulting spread in predicted level is at most ${\approx}1$\,dB, comparable to the digitisation uncertainty, and is not shown.}
  \label{fig:expval}
\end{figure}

The chain is then confronted with measurement.
Figure~\ref{fig:expval} shows three published leading-edge-noise
experiments (four panels) with fully documented inflow turbulence---the NACA0012 of
\citet{PatersonAmiet1976} at $40$--$90$\,m\,s$^{-1}$ (higher speeds
excluded only by their Mach number), and the flat plates of
\citet{Bampanis2022} and \citet{Narayanan2015}---each with three
predictions computed through the identical response: the frozen
von K\'arm\'an and Liepmann inflows of standard practice, and the
model's computed spectrum, evaluated at the
effective free-distortion strain accumulated along the potential-flow
stagnation streamline down to the eddy scale ($e_{\rm eff}=0.09$--$0.16$
for the NACA0012; below $0.05$ for the thin plates). Three statements
result. First, the assembled chain is quantitatively right where the
experiment is cleanest: on the flat-plate cases every formulation sits
within $1$--$2$\,dB rms of the measured spectra with the
non-compactness interference lobes reproduced---the fidelity
\citet{Bampanis2022} report for their own implementation---so the
comparison discriminates inflow physics, not implementation. Second,
the computed correction scales with the effective strain: for the thin
plates the effective strain is negligible and the computed inflow returns
the frozen spectrum, while for the NACA0012 it lowers the prediction
towards the data at every speed (bias $+5.7\!\to\!+5.3$\,dB at $40$ and
$+5.1\!\to\!+4.7$\,dB at $90$\,m\,s$^{-1}$ in panels \textit{a,b}, and
$+1.8\!\to\!+1.4$\,dB at the intermediate $60$\,m\,s$^{-1}$ of
figure~\ref{fig:frozen-cost}), by the modest amount the geometric
effective strain implies; the shift is independent of
the digitisation uncertainty ($\pm0.5$\,dB), common to all three. Third,
the correction stays small at these laboratory strains, which is a limit
of the test, not of the model: laboratory rigs are designed to minimise
inflow distortion, so no published far-field measurement samples the
strains at which the distortion correction is large. The residual
NACA0012 overprediction is
the finite-thickness effect that \citet{PatersonAmiet1976} themselves
identified (breakdown near $U/f t_A\approx1$) and
that \citet{Gershfeld2004} modelled, of a piece with the blocking mechanism
excluded from the present scope; it is insensitive to the choice
of inflow spectrum.

\begin{figure}
  \centering
  \includegraphics[width=\textwidth]{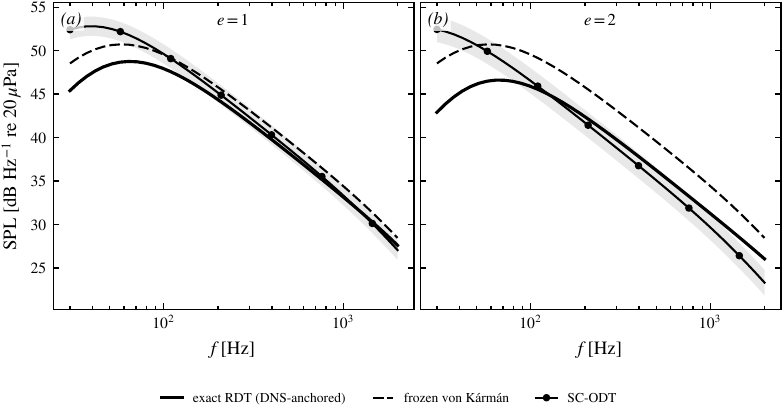}
  \caption{Computed against prescribed inflow at application strains,
  judged by the linear reference: absolute SPL at the reference
  configuration for the rapid family ($\chi\approx8$) at
  (\textit{a}) $e=1$ and (\textit{b}) $e=2$. Solid black: exact
  (inviscid) RDT of the model's relaxed initial state---the reference the
  strained-box DNS of \S\,\ref{sec:box-dns} validates at rapid
  rates---with level set by the kernel's exact upwash amplification
  (reproducing the moment-level value to $0.7\%$). Black dashed: the
  frozen von K\'arm\'an inflow of standard practice. Black circles: the
  computed SC-ODT inflow (fitted representation, measured ensemble level),
  with a grey band giving its measured spectral rms distance from the RDT
  family mapped to decibels.}
  \label{fig:amiet-rapid}
\end{figure}

Where the correction reaches the prediction is shown in
figure~\ref{fig:amiet-rapid}. Each panel shows absolute SPL against frequency at
one application strain; the black solid curve is the exact-RDT
reference that the strained-box DNS validates at this rate, so the
vertical distance of any other curve from it, read at fixed frequency,
is that inflow's error. The dashed frozen-von-K\'arm\'an curve sits
above the reference---the over-prediction of standard practice---while
the computed SC-ODT inflow (circles) falls on the reference within its
grey band over $100$\,Hz--$1$\,kHz (to within $2$\,dB), the band used for
the band means below; at the lowest frequencies it lies above both the
reference and the frozen inflow, by $7$\,dB ($e=1$) and $9.5$\,dB ($e=2$)
against the reference at $30$\,Hz. At total strains of $1$--$2$ the frozen-spectrum
error against the linear reference grows from $+1.4$ to $+3.4$\,dB
(band mean, figure~\ref{fig:frozen-cost}),
while the computed inflow tracks the reference within about
$1$\,dB (band mean)---a residual that contains the physical relaxation at
finite rapidity as well as the model deficiency measured in
\S\,\ref{sec:envelope}---so an evolved inflow outperforms a prescribed one
at the strains where the correction matters. A far-field configuration
that would sharpen the test samples the fine scales at accumulated
strains of order unity---turbulence-ingesting
rotors~\citep{MajumdarPeake1998}, where the streamtube contraction supplies the strain and the
blade response the fine-scale weighting---and is left as a falsifiable future target.

\begin{figure}
  \centering
  \includegraphics[width=0.8\textwidth]{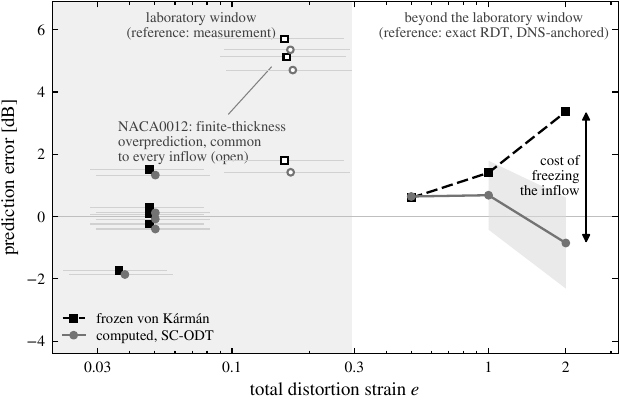}
  \caption{The cost of the frozen-inflow assumption against accumulated
  distortion: signed far-field prediction error, judged by the best
  available reference at each strain. Laboratory window (shaded): the
  eight measured cases (the four of figure~\ref{fig:expval} and four
  further speeds), bias over each
  case's comparison band---flat plates filled, NACA0012 open (its
  common overprediction is the finite-thickness effect); horizontal
  bars span the effective-strain truncation range, and the two
  inflow predictions are offset slightly in $e$ for legibility. Beyond the
  window: band-mean error over $100$--$1000$\,Hz against the exact-RDT
  reference of figure~\ref{fig:amiet-rapid} at $e=0.5$, $1$ and $2$;
  shaded band gives the $\pm10\log_{10}(1+\mathrm{rms})$
  spectral-distance uncertainty where measured ($e=1$, $2$).}
  \label{fig:frozen-cost}
\end{figure}

Figure~\ref{fig:frozen-cost} assembles the two regimes on a single
axis and summarises the section. In the laboratory window the frozen
and computed inflows agree with measurement to the same accuracy: the
window in which Amiet's theory has been validated is the window in which
its frozen-spectrum input costs little, consistent with the practice
having gone untested beyond it. What the
methodology adds is the right-hand half of the figure: an inflow that
continues to predict where the frozen assumption fails. Beyond the window the frozen
error grows monotonically---$+0.6$, $+1.4$, $+3.4$\,dB in the band
mean at $e=0.5$, $1$, $2$---with no measurement there to catch it,
while the computed inflow remains within about $1$\,dB of the
reference and carries its measured uncertainty with it.

The statement concerns the \emph{free} distortion by the mean strain: the
surface-blocking effect, which measurements and simulations show
\emph{enhancing} the near-edge upwash at low frequency
\citep{HuntGraham1978,dosSantos2023,Piccolo2024}, is a distinct and
partially compensating mechanism outside the scope of the present line
model; the calculation also replaces the strain history along
individual streamlines by the canonical constant-rate distortion. What the calculation establishes
is therefore not a noise prediction but the magnitude and sign of the
term that frozen-spectrum practice omits: order $5$--$20$\,dB at the
energy-containing reduced frequencies for total strains of
$0.5$--$2$---far too large to neglect, and of the opposite sign to the
blocking correction with which it must ultimately be combined.

Figure~\ref{fig:amiet-abs} carries the kernel result through to an
absolute far-field spectrum in physical units, the frozen baseline
coinciding with the standard Amiet--von K\'arm\'an prediction. A
predictive comparison against measurement at application strains, where
no direct measurement exists, additionally requires the blocking
correction and the streamline strain history; the present model supplies
the distorted-inflow input, with its uncertainty quantified by the
envelope of \S\,\ref{sec:envelope}.

\section{Conclusions}\label{sec:conclusions}

We have presented a strain-coupled formulation of one-dimensional
turbulence and, with it, an account of what a single line can and cannot
determine about rapid distortion. The formulation applies the event
machinery of standard ODT to a line under imposed mean strain: the
off-axis production and the rapid pressure--strain enter as a
continuous-in-time forcing built from an explicit spectral closure, and
the wavevector kinematics as a dilatation of the domain, while the eddy
events, acting on the full velocity, capture the cascade, the slow return
to isotropy and the on-line upwash production, delivered at the scale of
the overturn. Because the events act on the full velocity, each conserves
momentum and total kinetic energy exactly, the one-dimensional analogues
of the Navier--Stokes invariants, so the construction extends
one-dimensional turbulence rather than departing from it.

The model, denoted SC-ODT, reproduces the constant-free rapid-distortion
onset rate and is verified computationally to the level of the closure
it evolves ($10^{-4}$ in the component fractions); at the moment level
it reproduces the strained components of the \citet{LeeReynolds1985}
DNS quantitatively, with the one exception---the neutral-direction
anisotropy---traced to the orientation information that every
single-point closure discards, not to the line. The laboratory
counterpart is the plane-strain experiment of
\citet{ChenMeneveauKatz2006}: the model reproduces qualitatively the
departure from rapid distortion that the experiment reports, the
component centroids held far below the model's rigid-translation
kinematics, the upwash most (\S\,\ref{sec:cmk}; supplementary material,
\S\,S1).

Three findings bear on how such a model is to be assessed. First, the
correct linear reference for line-projected spectra is broadband: exact
rapid-distortion theory, projected onto a line, rescales the
projected upwash spectrum by a factor of $2.4$ across the resolved range
at $e=1$, so a rigid spectral translation---the model's own
kinematics---is not the linear prediction, and comparisons built on it
conflate a structural approximation with nonlinear physics. Second,
against a strained-box DNS with exact-RDT companions reduced to the
same observables, SC-ODT reproduces the sign and wavenumber dependence of
the strained upwash anisotropy at moderate strain rate, because the
on-line production is deposited at the eddy scale rather than uniformly
across wavenumber (\S\,\ref{sec:event-strain}). At rapid strain the
model stays close to the DNS at the energy-containing scales but does not
reproduce the sign reversal of the fine-scale upwash anisotropy
(figure~\ref{fig:eddytypes}\textit{b}), and a
distance-from-linear-theory envelope, growing with strain rapidity,
is carried with the model into applications. Third, the single-line
closure question admits an exact, falsifiable answer in place of a
closure claim: line data bound the rapid term sharply, the bounds are
valid up to $e\approx0.25$--$0.5$ under plane strain and, by
construction, under distortion axisymmetric about the line, and the axisymmetry
error of the axisymmetry hypothesis is measured: it grows with accumulated
strain and rapidity under plane strain, while the acoustically relevant
upwash amplification changes little with it.

SC-ODT reproduces the strain-induced fine-scale upwash anisotropy of the
strained-box DNS at the moderate operating point
(figure~\ref{fig:eddytypes}), the two eddy types supplying the scale-wise
component allocation that a single kernel does not. The one regime it does not
reach is rapid strain, in which the wavenumber-uniform rapid operator acts
essentially alone and the DNS carries the fine-scale upwash \emph{below}
the transverse mean; that regime lies outside the $\chi=O(1)$ operating point of the leading-edge
application.

For the application that motivated the work, the model supplies a
strained, scale-resolved upwash spectrum at a cost compatible with
parametric study, and \S\,\ref{sec:lenoise} carries it through Amiet's
theory to the radiated level. The assembled chain reproduces three published far-field experiments to
$1$--$2$\,dB rms for the facilities with the best-characterised inflow; the
computed correction vanishes for thin plates and lowers the prediction
towards the data for the NACA0012. Against the linear reference at
application strains the evolved inflow lies closer than the frozen one, by
a margin that grows with strain. Relative to the
frozen-spectrum input of standard practice, at the slow rate
($\chi\approx0.4$) the computed distorted inflow is a weaker leading-edge
source at every resolved frequency---by $5$--$20$\,dB at the
energy-containing scales for total strains $0.5$--$2$, in the geometry-free
kernel ratio and in the assembled absolute prediction alike. Most of this
reduction is the decay of the upwash variance over the transit time at this
slow rate; the redistribution of the spectrum adds to it at high frequency,
where the distortion drains the spanwise-coherent upwash content the
aerofoil responds to. The omitted term is too large to neglect and opposite in sign
to the surface-blocking correction that a complete prediction must
likewise include. Predictions inherit a spectral-shape uncertainty of
order the distance envelope at the operating point's strain and rapidity;
that envelope is itself a quantity that measurement can test.

Relative to the corrections reviewed in \S\,\ref{sec:intro}, which prescribe
the distorted inflow spectrum from linear theory or from measurement
\citep{deSantana2016,dosSantos2023,Piccolo2024}, the model computes it
with the nonlinear relaxation retained, at a cost compatible with
parametric study, and its departures from linear theory and from the DNS
are measured rather than assumed. On the evidence presented, its useful
range is moderate strain rapidity, $\chi=O(1)$, the regime of the
leading-edge stagnation region: there it reproduces the strained moment
components and the fine-scale upwash anisotropy of the DNS, and its
spectral-shape distance from linear theory stays between $13$ and $24$ per
cent up to $e=2$ (figure~\ref{fig:rdt-distance}). At rapid strain
($\chi\approx16$) the fine-scale anisotropy is not reproduced. Extending
this range requires a rapid operator that discriminates by wavenumber;
extending the geometry requires a treatment of rotational mean flows and
of lines not aligned with a principal axis, together with coupling to the
surface-blocking correction.

\appendix
\section{The forward map and the bound kernels}\label{app:forward}

This appendix derives the three ingredients that \S\,\ref{sec:gate} uses:
the two-scalar representation of the axisymmetric solenoidal spectrum with
its realisability cone, the forward map from the two scalars to the line
spectra (with the falsifiers T1 and T2), and the closure kernels
\eqref{eq:g1}--\eqref{eq:g3}, together with the linear programs that turn
these linear relations into the sharp bounds of \S\,\ref{sec:bounds}.

\subsection{The axisymmetric representation and realisability}
\label{app:twoscalar}

The spectrum tensor of statistically homogeneous turbulence is Hermitian,
positive semi-definite at each $\boldsymbol{\kappa}$, and solenoidal,
$\kappa_m\Phi_{mn}(\boldsymbol{\kappa})=0$; it therefore has rank at most
two, supported on the plane normal to $\boldsymbol{\kappa}$. Upon invoking axisymmetry about $\boldsymbol{e}_2$, mirror symmetry, and zero helicity, here denoted Assumption~A1, it is fixed by two scalar functions of $(\kappa_2,\kappa_\perp)$
\citep{Batchelor1946,ChandrasekharAxisym1950,Lindborg1995}:
\begin{equation}
  \Phi_{mn}(\boldsymbol{\kappa})
  = a(\kappa_2,\kappa_\perp)\,P_{mn}(\boldsymbol{\kappa})
  + c(\kappa_2,\kappa_\perp)\,\xi_m\xi_n,
  \qquad
  \boldsymbol{\xi} = \boldsymbol{e}_2 - \mu\,
  \frac{\boldsymbol{\kappa}}{\kappa},
  \quad \mu=\frac{\kappa_2}{\kappa},
  \label{eq:twoscalar}
\end{equation}
with $P_{mn}=\delta_{mn}-\kappa_m\kappa_n/\kappa^2$ the solenoidal
projector, so that $\kappa_m\xi_m=0$ and
$|\boldsymbol{\xi}|^2=1-\mu^2=1-x$ with $x=\kappa_2^2/\kappa^2$:
solenoidality is exact by construction, which is how the third would-be
scalar of a general axisymmetric tensor is eliminated. For a unit vector
$\boldsymbol{v}$ normal to $\boldsymbol{\kappa}$,
$\Phi_{mn}v_mv_n = a + c\,(\boldsymbol{\xi}\cdot\boldsymbol{v})^2$, so the
two eigenvalues of the polarisation matrix on that plane are $a$ (attained
for $\boldsymbol{v}\perp\boldsymbol{\xi}$) and $a+(1-x)\,c$ (for
$\boldsymbol{v}\parallel\boldsymbol{\xi}$). Positive semi-definiteness is
therefore the pointwise \emph{realisability cone}
\begin{equation}
  a(\kappa_2,\kappa_\perp)\ \ge\ 0,
  \qquad
  a(\kappa_2,\kappa_\perp) + (1-x)\,c(\kappa_2,\kappa_\perp)\ \ge\ 0 ,
  \label{eq:cone}
\end{equation}
quoted in \S\,\ref{sec:bounds}. Isotropy is the special case $c=0$ with
$a$ a function of $\kappa$ alone.

\subsection{The forward map to the line spectra}\label{app:fwdmap}

Writing
$\boldsymbol{\kappa}=(\kappa_\perp\cos\theta,\ \kappa_2,\
\kappa_\perp\sin\theta)$, the azimuthal moments at fixed
$(\kappa_2,\kappa_\perp)$ are
$\langle\kappa_1^2\rangle_\theta=\langle\kappa_3^2\rangle_\theta
=\tfrac12\kappa_\perp^2$,
$\langle\kappa_1^4\rangle_\theta=\langle\kappa_3^4\rangle_\theta
=\tfrac38\kappa_\perp^4$,
$\langle\kappa_1^2\kappa_3^2\rangle_\theta=\tfrac18\kappa_\perp^4$, and
every odd moment vanishes. Applied to \eqref{eq:twoscalar} the diagonal
weights average to
\begin{equation}
  \langle P_{11}\rangle_\theta=\langle P_{33}\rangle_\theta
  =\tfrac12(1+x), \quad
  P_{22}=1-x, \quad
  \langle \xi_1^2\rangle_\theta=\langle \xi_3^2\rangle_\theta
  =\tfrac12 x(1-x), \quad
  \xi_2^2=(1-x)^2 ,
  \label{eq:azweights}
\end{equation}
and every mixed weight to zero. Substituting into the definition
\eqref{eq:phi1d} of the line spectra,
$\phi_{mn}(\kappa_2)=\int_0^\infty\langle\Phi_{mn}\rangle_\theta\,
2\pi\kappa_\perp\,\mathrm{d}\kappa_\perp$, gives the \emph{forward map}
from the two scalars to the data a line carries:
\begin{align}
  \phi_{11}(\kappa_2)=\phi_{33}(\kappa_2)
  &= 2\pi\!\int_0^\infty\!
     \Bigl[\,a\,\tfrac{1+x}{2} + c\,\tfrac{x(1-x)}{2}\,\Bigr]
     \kappa_\perp\,\mathrm{d}\kappa_\perp ,
  \label{eq:fwd11}\\
  \phi_{22}(\kappa_2)
  &= 2\pi\!\int_0^\infty\!
     \Bigl[\,a\,(1-x) + c\,(1-x)^{2}\,\Bigr]
     \kappa_\perp\,\mathrm{d}\kappa_\perp ,
  \label{eq:fwd22}\\
  \phi_{mn}(\kappa_2) &= 0 \qquad (m\ne n).
  \label{eq:fwdcross}
\end{align}
Equations \eqref{eq:fwd11}--\eqref{eq:fwdcross} contain the two free
falsifiers quoted in \S\,\ref{sec:bounds}: under A1 the two transverse
line spectra coincide at every resolved wavenumber (T1), and the
cross-spectra vanish (T2); both are properties the line itself reports,
scale by scale, with no reference to unresolved data. The
$\phi_{22}$ bracket is the same combination of $a$ and $c$ that appears in
the kernel $g_2$ below---the conventions mesh, as they must, because
$\phi_{22}$ and the $22$-component of the rapid term weight the spectrum
through the same projection geometry.

They also explain why the isotropic intuition about line data fails under
A1. For isotropic turbulence ($c=0$, $a=a(\kappa)$) the map is
over-determined: one unknown function of one variable against two measured
functions, with $\phi_{22}$ determining $a$ by the classical Abel-type
inversion and the transverse spectrum slaved to it by the derivative
relation. Under A1 the bookkeeping reverses: \emph{two unknown functions
of two variables} against two measured functions of one variable. At each
fixed $\kappa_2$ the data supply exactly two numbers---the two weighted
$\kappa_\perp$-integrals \eqref{eq:fwd11}--\eqref{eq:fwd22}---about two
unknown functions of $\kappa_\perp$; any perturbation
$(\delta a,\delta c)(\kappa_\perp)$ annihilating both integrals while
respecting \eqref{eq:cone} is invisible to the line. This
infinite-dimensional null space is the quantitative content of the
obstruction of \S\,\ref{sec:obstruction} after the axisymmetric collapse.

\subsection{Derivation of the closure kernels}\label{app:kernels}

For the diagonal components the exact result \eqref{eq:Mijkl} reads
\begin{equation}
  \Pi^{(\mathrm r)}_{nn}
  = 2\,A_{kl}\int_{\mathbb{R}^3}
    \frac{\kappa_n\kappa_k}{\kappa^2}\,
    \Phi_{nl}(\boldsymbol{\kappa})\,\mathrm{d}\boldsymbol{\kappa}
  \qquad(\text{no sum on }n).
  \label{eq:pinn-exact}
\end{equation}
Substituting \eqref{eq:twoscalar}, carrying out the azimuthal average with
the moments above, and specialising to
$A=\mathrm{diag}(\tfrac12,-\tfrac12,0)$ reduces each component to a
$2\pi\kappa_\perp$-weighted integral over the half-plane of a polynomial
combination in $x$, linear in $(a,c)$---the kernels
\eqref{eq:g1}--\eqref{eq:g3} of the main text. The upwash component
illustrates the mechanics: only $k=l=1$ and $k=l=2$ contribute, with
azimuthal averages
$\langle\kappa_1\Phi_{21}\rangle_\theta
=-\tfrac12\kappa_\perp^2\,\kappa_2\,[\,a+(1-x)\,c\,]/\kappa^2$
and
$\langle\Phi_{22}\rangle_\theta=(1-x)[\,a+(1-x)\,c\,]$, so that both enter
through the same bracket and
\begin{equation}
  \bigl\langle\text{integrand of \eqref{eq:pinn-exact}}\bigr\rangle_\theta
  \Big|_{n=2}
  = 2\Bigl[-\tfrac12\cdot\tfrac{x(1-x)}{2}
           -\tfrac12\,x(1-x)\Bigr]\bigl[a+(1-x)c\bigr]
  = -\tfrac32\,x\,(1-x)\bigl[a+(1-x)c\bigr],
  \label{eq:g2derive}
\end{equation}
which is $g_2(x)$. The transverse components involve in addition the
fourth moments of \eqref{eq:azweights}'s underlying table and yield
$g_1$ and $g_3$; the full reduction has been verified symbolically against
\eqref{eq:pinn-exact}. Three identities check the algebra:
$g_1+g_2+g_3=0$ at every $x$ for both scalars (pointwise energy
conservation); each kernel vanishes at $x=1$, where the wavevector aligns
with the line and the projection degenerates; and averaging over the
isotropic angular measure $\mathrm{d}\mu$ on $[0,1]$ gives
$\langle g_1\rangle:\langle g_2\rangle:\langle g_3\rangle
=+\tfrac15:-\tfrac15:0$, reproducing
$\Pi^{(\mathrm r)}_{ij}=\tfrac45 k_tS_{ij}$ \citep{Crow1968} as in
\S\,\ref{sec:lrr-reduction}.

\subsection{The bounds as linear programs}\label{app:lp}

The target \eqref{eq:Piclosed}, the data constraints
\eqref{eq:fwd11}--\eqref{eq:fwd22}, and the cone \eqref{eq:cone} are all
linear in $(a,c)$, and everything decomposes over the resolved wavenumber:
writing $\Pi^{(\mathrm r)}_{nn}=\int_0^\infty
\pi_n(\kappa_2)\,\mathrm{d}\kappa_2$ with
$\pi_n=2\pi\int_0^\infty g_n(x)[a,c]\,\kappa_\perp\,\mathrm{d}\kappa_\perp$,
the extremisation over all spectra consistent with the line data is a
family of \emph{independent} infinite-dimensional linear programs, one per
$\kappa_2$:
\begin{equation}
  \pi_n^{\pm}(\kappa_2)
  = \max/\min_{(a,c)(\kappa_\perp)}\;
    2\pi\!\int_0^\infty g_n(x)\,[a,c]\,\kappa_\perp\,\mathrm{d}\kappa_\perp
  \quad\text{s.t.}\quad
  \text{\eqref{eq:fwd11}--\eqref{eq:fwd22} hold, \eqref{eq:cone}
  pointwise},
  \label{eq:LP}
\end{equation}
whence $\Pi^{\pm}_{nn}=\int\pi_n^{\pm}\,\mathrm{d}\kappa_2$ and the
indeterminacy $\mathcal I_n=(\Pi^{+}_{nn}-\Pi^{-}_{nn})/2k_tS$ of
\S\,\ref{sec:bounds}, normalised by the natural rapid-term scale so that
$\mathcal I_n$ is comparable across strains (the isotropic value is
$\Pi^{(\mathrm r)}_{11}=\tfrac25 k_t$ per unit $S_{11}=\tfrac12$).
Discretised on a logarithmic $\kappa_\perp$ grid of order $10^2$ points,
\eqref{eq:LP} is a linear program with two equality constraints and the
cone, solved per $\kappa_2$ in milliseconds; the bounds are global and
sharp by LP duality---no fitting and no local minima. The side conditions
of \S\,\ref{sec:bounds} enter as further linear constraints: the
log-Lipschitz slope cap
$|\mathrm{d}\ln(a,c)/\mathrm{d}\ln\kappa_\perp|\le\lambda$ and the
polarisation cap $|c|\le\gamma a$. The LP extremals concentrate on
few-point supports in $\kappa_\perp$ (bang--bang spectra), which is why
the raw bounds are conservative and why smoothness alone barely narrows
them: a log-Lipschitz spectrum can still vary like
$\kappa_\perp^{\pm\lambda}$, so the angular reweighting that drives the
indeterminacy survives smoothing, while the polarisation cap attacks it
directly. The dual variables of \eqref{eq:LP} certify, per resolved
wavenumber, which linear combinations of the measured line spectra control
each $\pi_n$; finite-window bias on $\phi_{mn}$ at small $\kappa_2$
propagates linearly and can be carried as interval data in place of the
equalities with no change of machinery.

\section{Numerical parameters, initial conditions and statistics}
\label{app:numerics}

\subsection{ODT configurations}\label{app:odtconf}

All line ensembles share the solver settings of standard ODT
\citep{StephensLignell2021}: eddy-rate constant $C=7$, viscous-penalty
parameter $Z=450$, kernel parameter $\alpha=2/3$, adaptive mesh with
$\mathrm{d}x\in[2\times10^{-4},2\times10^{-2}]$ on a unit domain
($\mathrm{d}x_{\min}=5\times10^{-5}$ at the operating point of
\S\,\ref{sec:spec-op}), explicit advancement with diffusive CFL $0.5$
and a strain-resolved bound of $10^{-2}$ on the total-strain increment
per step (\S\,\ref{sec:level1a}). The initial condition is the resolved
spectral field of \S\,\ref{sec:level1a}: a band-limited isotropic
superposition of $64$ Fourier modes peaked at eight wavelengths per
domain, which places no energy at the grid scale. The $64$ initial
Fourier modes of each line are whitened by the \emph{symmetric} inverse
square root $C^{-1/2}$ of the sample covariance, so the initial
anisotropy vanishes per realisation. This choice matters: an
ordered-Cholesky whitening used initially left a bias
$b_{33}=+0.0148\pm0.0046$, detected by the nominal-turbulence null test of
the ODT line (strain machinery active, $A=0$, eddy events on) and traced
by a reversal experiment to the ordering; with the symmetric form all
$|b_{ii}|\le0.005$ and the on-line falsifier T1 exceeds its confidence
band at the chance rate. The ensembles of
\S\,\ref{sec:validated} and of \S\,\ref{sec:lenoise} use an unstrained precursor: the
line evolves as standard ODT until the eddy population is statistically
active and the initialisation transient has relaxed, and the strain is
switched on only then (\S\,\ref{sec:box-dns}); those of
\S\S\,\ref{sec:spec-bvc}--\ref{sec:spec-op} are strained from the initial
field, and their statistics are read after the initialisation transient.
The molecular viscosity is $10^{-5}$ for all ensembles except the
operating point of \S\,\ref{sec:spec-op}, $3\times10^{-6}$. The number of
cells after the initialisation transient and the resolution relative to
the Kolmogorov length are given for every ensemble in
table~\ref{tab:resolution}. The strained line spectrum ensembles carried into the leading-edge-noise chain of \S\,\ref{sec:lenoise} are $1024$ realisations per case at each rapidity $\chi\approx0.4$ and $\approx8$. The resolved scale separation of the $\nu=10^{-5}$ ensembles, used for the CMK comparison of supplementary \S\,S1, is $\mathcal L/\eta\approx250$.

\begin{table}
  \centering\small
  \begin{tabular}{lcccc}
  ensemble & $\nu$ & cells & $\overline{\mathrm{d}x}/\eta$ & $\mathrm{d}x_{\min}/\eta$\\[2pt]
  strain on/off, \S\,\ref{sec:spec-bvc} & $10^{-5}$ & $760\to235$ & $1.9\to0.8$ & $0.28\to0.04$\\
  operating point, \S\,\ref{sec:spec-op} & $3\times10^{-6}$ & $980\to320$ & $3.0\to1.4$ & $0.16\to0.03$\\
  eddy types, \S\,\ref{sec:validated}, $\chi\approx0.8$ & $10^{-5}$ & $1630\to820$ & $2.2\to1.7$ & $0.71\to0.28$\\
  eddy types, \S\,\ref{sec:validated}, $\chi\approx16$ & $10^{-5}$ & $1660\to1570$ & $2.2\to2.2$ & $0.72\to0.67$\\
  \S\,\ref{sec:lenoise}, $\chi\approx0.4$ & $10^{-5}$ & $1630\to330$ & $2.2\to1.3$ & $0.71\to0.10$\\
  \S\,\ref{sec:lenoise}, $\chi\approx8$ & $10^{-5}$ & $1630\to1140$ & $2.2\to1.7$ & $0.73\to0.38$\\
  \end{tabular}
  \caption{Median number of cells after the initialisation transient
  (strain onset for the precursor-conditioned ensembles, $e=1$ for those
  of \S\S\,\ref{sec:spec-bvc}--\ref{sec:spec-op}) and at the end of strain
  ($64$ realisations per ensemble), with the mean and smallest cell sizes
  relative to the Kolmogorov length $\eta=(\nu^3/\varepsilon)^{1/4}$
  formed from the line-averaged dissipation rate.}
  \label{tab:resolution}
\end{table}

\subsection{Strained-box DNS and exact-RDT companions}\label{app:dns}

The benchmark of \S\,\ref{sec:box-dns} is a pseudo-spectral $128^3$
simulation in a Rogallo deforming frame: a decaying isotropic precursor,
then plane strain $A=S\,\mathrm{diag}(\tfrac12,-\tfrac12,0)$ at
$\chi\in\{0.8,16\}$ at onset, carried to $e=1$ over four
independent realisations, with $\kappa_{\max}\eta\approx0.76$ at
$e=1$---a pilot resolution, so small-scale statements are qualitative.
Every realisation has an exact-RDT companion: the same initial field
evolved under the closed-form Cauchy solution of the linearised
equations and projected onto lines identically. The estimator pipeline
of \S\,\ref{sec:bounds} passes the null battery of Appendix~\ref{app:odtconf}
on the Johns Hopkins
Turbulence Database isotropic turbulence dataset
\citep{Li2008,Perlman2007} ($1024^3$, $128$ lines): T1 exceedances at the chance rate,
cross-spectral coherences $\sim10^{-2}$, the longitudinal--transverse
derivative relation to $4\%$, and the LP bounds bracketing the exact
isotropic value at every constraint level.

\subsection{Statistics policy}

Spectral band quantities of the strained line are strongly intermittent
across realisations---individual realisations can carry $10^2$--$10^3$
times the median band energy---so ensemble means of band ratios converge
too slowly to be useful at practical sample sizes. The band statistics of
\S\,\ref{sec:validated} and \S\,\ref{sec:box-dns} are therefore medians over
realisations with $95\%$ bootstrap confidence
intervals of the median, and the two-eddy-type comparison of
\S\,\ref{sec:validated} is assessed by the median of \emph{paired}
same-seed contrasts, quoted significant only when the bootstrap
interval excludes zero. Ensemble sizes: $1024$ realisations for the
eddy-type comparison of \S\,\ref{sec:validated} and the distance envelope, $128$--$255$ for the
spectrum ensembles of \S\S\,\ref{sec:spec-bvc}--\ref{sec:spec-op}, and four seeds for
the DNS and its companions, whose derived quantities carry jackknife
errors over realisations.

\begin{acknowledgments}
\end{acknowledgments}

\vskip6pt\noindent{\bf Supplementary data.} The supplementary material follows the
references in this version.

\vskip6pt\noindent{\bf Funding.} This research received no specific grant
from any funding agency, commercial or not-for-profit sectors.

\vskip6pt\noindent{\bf Declaration of interests.} The authors report no
conflict of interest.

\vskip6pt\noindent{\bf Author ORCIDs.} S. Sharma,
https://orcid.org/0009-0000-1911-9317; A. R. Kerstein,
https://orcid.org/0000-0002-6694-3525.

\bibliographystyle{jfm}
\bibliography{jfm}

\clearpage
\normalsize
\makeatletter\global\appendixfalse\makeatother
\setcounter{section}{0}\setcounter{subsection}{0}
\setcounter{equation}{0}\setcounter{figure}{0}\setcounter{table}{0}
\renewcommand{\thesection}{S\arabic{section}}
\renewcommand{\thesubsection}{\thesection.\arabic{subsection}}
\renewcommand{\theequation}{S\arabic{equation}}
\renewcommand{\thefigure}{S\arabic{figure}}
\renewcommand{\thetable}{S\arabic{table}}
\begin{center}{\large\bfseries Supplementary material}\end{center}
\medskip

This supplement carries supporting material referenced from the main
text above: the validation against the strained-turbulence experiment of
\citet{ChenMeneveauKatz2006} (\S\,\ref{sec:s-cmk}), the operational steps
of the leading-edge-noise chain (\S\,\ref{sec:s-chain}), the
verification of the implementation (\S\,\ref{sec:s-level1a}), a
space--time view of a single realisation of the strained line
(\S\,\ref{sec:s-line}), and a cross-check of the Rogallo-frame
implementation against a lab-frame one (\S\,\ref{sec:s-lab}).

\section{Validation against the strained-turbulence experiment of
Chen, Meneveau \& Katz}\label{sec:s-cmk}

The strain-on/strain-off comparison of the main text (\S\,5.3) is model-internal;
an external benchmark is needed for the departure from rapid distortion. We use the plane-strain
experiment of \citet{ChenMeneveauKatz2006} (CMK), in which initially
near-isotropic turbulence at $R_\lambda\approx400$ is subjected to a
controlled planar straining cycle and the one-dimensional component
spectra are compared with rapid-distortion theory. We reproduce the
\emph{straining} phase of their cycle. The formulation of \S\,2 of the
main text applies equally to the destraining phases, in which the strain
reverses with $x_2$ remaining a principal axis, but the present
implementation imposes a constant $A_{ij}$, so those phases are not
computed here; the end of the straining phase corresponds to a total
deformation
$\mathcal{C}=\exp(\tfrac12\!\int S\,\mathrm{d}t)\approx3$, i.e. $e\approx2.20$ in
the notation of the main text. Unless stated otherwise, the line
ensembles here use the configuration of Appendix~B of the main text
(molecular viscosity $\nu=10^{-5}$, scale separation
$\mathcal L/\eta\approx250$).

In the rapid-distortion limit (eddy events suppressed, inviscid) the
model follows its kinematic reference and the moment closure
(figure~\ref{fig:s-cmk-rdt}): the component spectral centroids follow the
rigid-translation law---the model's own kinematics, not the projected
linear-theory prediction (main text \S\S\,5.1--5.2)---reaching the
deformation ratio $3.00$ at $e=2.20$ with the components coincident, and
the kinetic energy reaches $1.72$ against the exact amplification $1.69$
of the analytic reference of \citet{LeePiomelliReynolds1986} that CMK
use---the residual being the moment closure's, not the
implementation's (main text \S\,4.3), since with the
events suppressed the solver evolves the closure trajectory ($1.72$)
exactly. The truly inviscid setting is required: any nonzero viscosity
dissipates the compressed small scales and degrades the rigid
translation, an effect we verified directly.

\begin{figure}
  \centering
  \includegraphics[width=\textwidth]{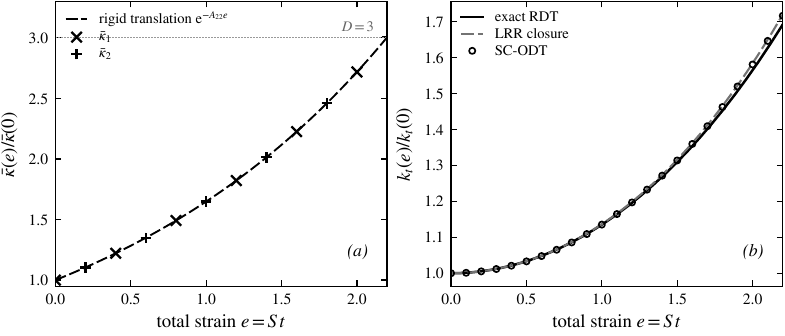}
  \caption{Rapid-distortion limit (eddy events suppressed, inviscid)
  against the analytic reference of \citet{LeePiomelliReynolds1986}
  used by \citet{ChenMeneveauKatz2006}. (\textit{a}) Component
  spectral-centroid migration: the streamwise ($\times$) and upwash
  ($+$) centroids both follow the rigid-translation law (the model's
  kinematic reference) and coincide
  throughout (the markers are interleaved so the two coincident series
  stay legible), reaching the deformation ratio $\mathcal{C}=3$ (marked $D$ in the figure) at the end of the
  straining phase. (\textit{b}) Kinetic-energy amplification: the solver follows
  the closure trajectory, within $2\%$ of the exact rapid-distortion
  value of \citet{LeePiomelliReynolds1986} at $e=2.20$.}
  \label{fig:s-cmk-rdt}
\end{figure}

With the eddy events and viscosity active, the full model (SC-ODT) is no longer constrained to rigid translation, and the strain drives a measurable anisotropy between the components. Figure~\ref{fig:s-cmk-aniso} shows
the strain-induced migration of the streamwise and upwash component
centroids, each measured against the matched no-strain reference (the per-component form of the strain-on/strain-off construction, which
cancels the initialisation transient and the eddies' own relaxation):
the components diverge progressively under continued strain---the
streamwise migration reaching $1.9$ at the end of the straining phase
against $1.5$ for the upwash, both far below the rigid value $3$---so
the upwash component is held back most. This corresponds qualitatively
to what CMK report, the components responding differently and departing
from the rapid-distortion prediction, measured here against the
rigid-translation kinematics; it is the spectral signature of the broadband
distortion analysed in the main text (\S\,5.3). The comparison is restricted to
its qualitative content, for three reasons: the
band-limited initial spectrum differs from CMK's, so absolute levels
are not comparable; the line-resolved one-dimensional spectra and the
experiment's transverse/longitudinal projections are not the same
projection of the spectrum tensor; and the scale separation achieved
here ($\mathcal{L}/\eta\approx250$) is below the experimental
$\approx930$, whose compressed dissipation range exceeds feasible line
resolution. Within these limits the model reproduces its
kinematic reference and the closure's energy amplification in the rapid
limit, and qualitatively the departure reported by the benchmark; the quantitative spectral assessment is carried out in
the main text against the strained-box DNS (\S\S\,5.6--5.7), where the projection and
the initial spectrum are matched by construction. A direct numerical
simulation of the same straining--destraining cycle
\citep{GualtieriMeneveau2010}, at a Reynolds number below the
experiment's, reports the component variances and the pressure--strain
budget through the cycle and compares them with the Launder--Reece--Rodi
(LRR) closure \citep{LaunderReeceRodi1975}; it is the DNS counterpart against which the present model can
be set quantitatively on this configuration. A direct overlay is not
attempted here: their observable is the anisotropy through the full
straining--\emph{destraining} cycle, whereas only the straining phase is
computed here, so the comparison is kept at the level of the
independent confirmation their DNS provides for the experiment.

\begin{figure}
  \centering
  \includegraphics[width=0.62\textwidth]{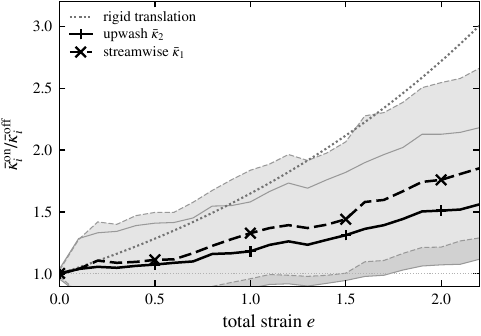}
  \caption{Full model (eddies and viscosity active; ensembles over
  $128$ realisations at $\nu=10^{-5}$): strain-induced migration of the
  streamwise and upwash component centroids, each as the geometric-mean ratio of the strained ensemble to the matched no-strain reference, against the rigid-translation law (grey dotted); the shaded bands span $\pm1$ geometric standard deviation of the ratio across the ensemble. The components diverge under
  strain---the upwash held back most---reproducing qualitatively the
  departure from rapid distortion reported by \citet{ChenMeneveauKatz2006},
  measured here against the rigid-translation law. The
  comparison is qualitative; see text.}
  \label{fig:s-cmk-aniso}
\end{figure}

\section{Operational steps of the leading-edge-noise chain}\label{sec:s-chain}

This section collects, in operational order, the steps of \S\,6 of
the main text, from strained line ensembles to radiated level; every
step restates the main text.
\begin{enumerate}
\item \emph{Inflow representation.} For each total strain $e$ and
  strain rapidity $\chi$ ($\approx0.4$ and $\approx8$), fit the exact-RDT-distorted von K\'arm\'an
  family to the median component line spectra of the
  $1024$-realisation SC-ODT ensemble, jointly in the upwash and transverse
  components, with the pre-distortion amplitude $A_0$ and energy
  wavenumber $\kappa_e$ the only free parameters. The rms log-residual
  of the fit is the distance envelope of the main text (\S\,5.6.1,
  figure~12), carried
  through the chain as the input uncertainty.
\item \emph{Two-dimensional upwash spectrum.} Form
  \[\Phi_{ww}(K_x,K_y;e)=A_0\int\Phi^{\rm RDT}_{22}
  (K_x,\kappa_2,K_y;e)\,\mathrm{d}\kappa_2\]
  by direct integration of
  the distorted components of the fitted family; the integration is
  verified against the wavevector-ensemble integrator of the main
  text (\S\,4.2; upwash-variance growth reproduced to under $1\%$).
\item \emph{Level and shape.} The noise kernel is linear in the
  upwash variance, so the variance ratio
  $\overline{u_2^2}(e)/\overline{u_2^2}(0)$ is taken from the
  measured line spectra and applied as the level; the fitted family
  supplies only the shape.
\item \emph{Geometry-free ratio.} Form
  $\mathcal K(K_x)=\pi\,\Phi_{ww}(K_x,0)^2\big/\!\int\Phi_{ww}
  (K_x,K_y)\,\mathrm{d}K_y$ for the distorted and for the frozen
  ($e=0$) input;
  $\Delta\mathrm{SPL}(K_x)=10\log_{10}
  [\mathcal K^{\rm distorted}/\mathcal K^{\rm frozen}]$, from which
  aerofoil geometry, response and observer cancel.
\item \emph{Absolute prediction.} At the flow conditions of the reference configuration (main text \S\,5.4),
  physical units enter through two matching conditions only---the
  integral scale through $\kappa_e=0.7468/\mathcal{L}$ and the $e=0$
  upwash variance---and the far-field spectrum is assembled from the
  compact lift-dipole form of \citet{Amiet1975} with the exact Sears response \citep{Sears1941}; the frozen baseline is then the standard
  Amiet--von~K\'arm\'an prediction with no recalibration.
\item \emph{Validation chain.} Against the published measurements,
  the three inflows (frozen von K\'arm\'an, frozen Liepmann, SC-ODT)
  pass through the identical non-compact response,
  with the thickness admittance of \citet{Gershfeld2004} applied to all three chains
  on the NACA0012 cases and the effective strain taken as the
  potential-flow stagnation-line accumulation truncated at the eddy
  scale.
\end{enumerate}

\section{Verification of the implementation}\label{sec:s-level1a}

The rapid-distortion checks of the main text verify the
\emph{formulation}; this section verifies the \emph{implementation},
in which the operator is not prescribed but reconstructed at every
substep: the line Reynolds stress is formed from the resolved field,
the closure supplies $\Pi^{(\mathrm r)}_{ij}$, the Lyapunov equation
(2.21) of the main text is solved for $B_{ij}$, and
$\mathcal{A}_{ij}=-A_{ij}+B_{ij}$ is applied pointwise. To isolate
this chain the solver is run in the rapid-distortion limit---plane
strain, eddy events suppressed, zero molecular diffusion---from a
statistically isotropic resolved initial field, and the component
energies are evaluated as length-weighted line averages. With the
events inactive the line's component-energy evolution must coincide
with the closed moment equation across the whole strain range, not
merely at onset.

Figure~\ref{fig:s-level1a} shows that it does: the solver output lies
on the LRR trajectory over $0\le e\le4$ to within $10^{-4}$ in each
component fraction, so the per-substep stress reconstruction, Lyapunov
solve and operator assembly are correct at every state of anisotropy.
The solver reproduces the \emph{closure}, as it should: with the events and diffusion
suppressed the line carries no information beyond the single-point
statistics the closure itself uses. The offset from the exact solution
is the single-point closure error quantified in the main text
(\S\,4.2), not an implementation error; agreement with
the exact solution here would indicate a fault. The fractions test
only the anisotropy; the kinetic energy, which the traceless rapid
term cannot change directly and which grows through the production
integral $\mathrm{d}\ln k_t/\mathrm{d}e=f_{22}-f_{11}$, provides the
independent magnitude check, and the solver reproduces the LRR
amplification to within $0.3\%$. The test requires a resolved spectral
initialisation (a white-noise start places energy at the grid scale,
where mesh-adaption interpolation damps it) and a strain-resolved
timestep bound in the inviscid limit; both effects act nearly
isotropically, so the energy panel is what distinguishes them.

\begin{figure}
  \centering
  \includegraphics[width=\textwidth]{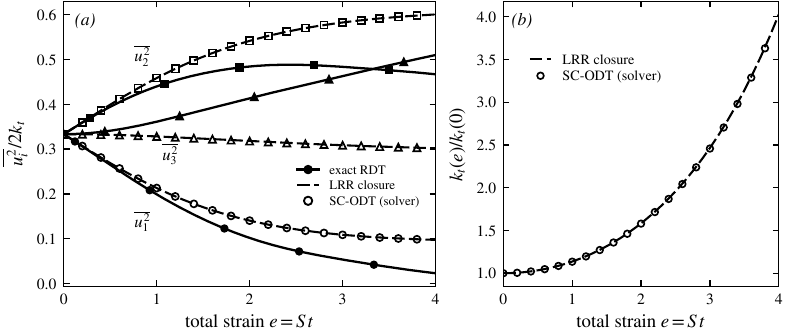}
  \caption{Verification of the SC-ODT implementation under plane
  strain, eddy events suppressed and zero molecular diffusion.
  (\textit{a}) Component energy fractions; (\textit{b}) kinetic-energy
  amplification $k_t(e)/k_t(0)$. Open symbols: the SC-ODT line solver
  (length-weighted line averages), by component
  ($\circ\,\overline{u_1^2}$, $\square\,\overline{u_2^2}$,
  $\triangle\,\overline{u_3^2}$); dashed: the LRR closure integrated as
  a moment equation; solid with filled symbols: exact rapid-distortion
  theory. The solver reproduces the closure to $10^{-4}$ in the
  fractions and $0.3\%$ in the energy; the offset from the exact
  solution is the closure error quantified in the main text (\S\,4.2).}
  \label{fig:s-level1a}
\end{figure}

\section{The strained ODT line, in space and time}\label{sec:s-line}

Figure~\ref{fig:s-hero} shows the upwash velocity along a single line
against time, for two runs from the same seed and initial field with the
eddy events active in both. Under standard ODT (strain off) the line keeps
its full width and the turbulence is statistically homogeneous and
component-symmetric: eddy events of all sizes punctuate a slowly decaying
field. Under the strain-coupled model the plane strain, with the
accompanying domain dilatation, compresses the line by about a factor of
seven over the strain range while the on-line production amplifies the
upwash component. The triplet-map eddy structure is carried through the
compression rather than erased by it---the line-level picture of the
strained state that the main text characterises spectrally.

\begin{figure}
  \centering
  \includegraphics[width=\textwidth]{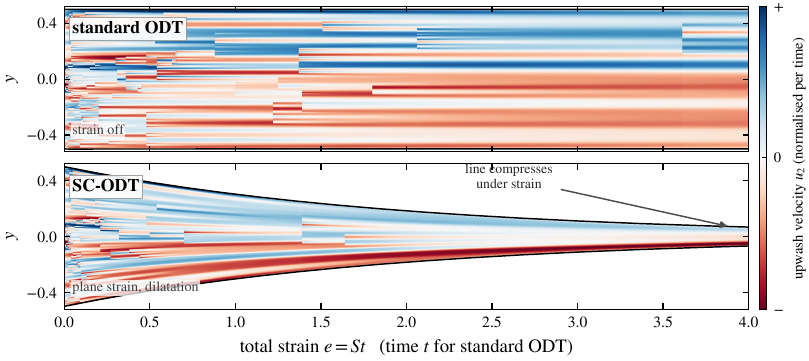}
  \caption{Space--time view of a single realisation of the ODT line: the
  upwash velocity $u_2$ along the line (vertical) against total strain
  (horizontal), normalised per time so the eddy structure stays visible as
  the fluctuations decay (blue positive, red negative). (\textit{Top})
  standard ODT, strain off---the full-width, statistically homogeneous
  line. (\textit{Bottom}) SC-ODT under plane strain with domain dilatation;
  the black envelope is the physical line half-width, so the line visibly
  compresses under strain while the eddy structure persists. Same seed and
  initial field in both.}
  \label{fig:s-hero}
\end{figure}

\section{Lab-frame cross-check of the Rogallo-frame implementation}\label{sec:s-lab}

All SC-ODT results of the main text are computed in the Rogallo frame
(main text \S\,2.9). An independent implementation in the physical (lab)
frame carries the mean $U_2=A_{22}(y-y_c)$ in the line velocity, applies
the event (2.2) of the main text to it directly and maintains the mean
between events (main text \S\,2.7); the line Reynolds stress entering the
rapid operator is formed from the pointwise de-meaned velocity. The two
were run on $32$ paired realisations (identical seeds and initial fields)
under plane strain with $S=2$, and the lab-frame results were compared
after back-transformation. Up to $e\approx0.4$ the total fluctuation
energy agrees to $0.5$--$1.9\%$ and each component energy to within
$4\%$, and the number of eddy events in the strained phase is $86$
against $84$. Started from laminar flow (no fluctuation, the mean
gradient alone), the paired runs fire the identical first eddy event in
both frames, at the same time and with the same size, with per-event
energy budgets equal to round-off; thereafter the realisations
decorrelate and the two frames agree statistically.

\end{document}